%% file: long_cycles_main.tex
\documentclass[11pt,reqno]{amsart}
 \pdfoutput=1
\usepackage{multirow}
\usepackage{mathtools}

\DeclarePairedDelimiter\floor{\lfloor}{\rfloor}
\usepackage{amssymb}
\usepackage{amsmath}
\usepackage{amsfonts}
\usepackage{mathrsfs}
\usepackage[nice]{nicefrac}
\usepackage{amsaddr}
\usepackage[dvipdfmx]{graphicx} 
\usepackage{subcaption}
\usepackage{bmpsize}
\usepackage{color}
\usepackage{setspace}
\usepackage{caption}
\usepackage{natbib}

\usepackage{booktabs} 
\usepackage{diagbox}
\usepackage{makecell}
\usepackage{adjustbox}
\usepackage{float}
\usepackage{threeparttable}
\usepackage{rotating}
\usepackage{pdflscape}
\usepackage[shortlabels]{enumitem}
\usepackage{charter}
\usepackage[colorlinks=true,citecolor=blue,urlcolor=blue,pdfpagemode=UseNone,pdfstartview=FitH]{hyperref}
\usepackage{apptools}
\usepackage{physics}
\usepackage{bigints}
\makeatletter
\def\section{\@startsection{section}{1}
	\z@{1.0\linespacing\@plus\linespacing}{.8\linespacing}{\Large}}

\def\subsection{\@startsection{subsection}{2}
	\z@{.8\linespacing\@plus.7\linespacing}{.7\linespacing}{\large}}

\def\subsubsection{\@startsection{subsubsection}{3}
	\z@{.5\linespacing\@plus.7\linespacing}{-.5em}{\normalfont\bfseries}}
\makeatother

\numberwithin{equation}{section}

\newtheorem{proposition}{Proposition}[section]
\newtheorem{lemma}{Lemma}[section]
\newtheorem{corollary}{Corollary}[section]
\theoremstyle{definition}
\newtheorem{definition}{Definition}[section]

\theoremstyle{definition}
\newtheorem{assumption}{Assumption}[section]

\theoremstyle{definition}

\theoremstyle{definition}

\usepackage{afterpage}
\usepackage{pdflscape}

\allowdisplaybreaks
\expandafter\let\expandafter\oldproof\csname\string\proof\endcsname
\let\oldendproof\endproof
\renewenvironment{proof}[1][\proofname]{%
  \oldproof[\normalfont \bfseries  #1]%
}{\oldendproof}

\DeclareMathOperator*{\argmin}{arg\,min}

\begin{document}
%
%
%
%
%

\title{\large Modeling Long Cycles}

\date{\today}

\author{Natasha Kang\and Vadim Marmer\textsuperscript{\textasteriskcentered} }
\thanks{\textsuperscript{\textasteriskcentered}Corresponding author. Address: Vancouver School of Economics, University of British Columbia, 6000 Iona Drive, Vancouver, BC Canada, V6T 1L4. Phone: +1 (604) 822-8217, Email: vadim.marmer@ubc.ca}
\thanks{We thank our co-editor, Serena Ng, the associate editor, and an anonymous referee for their comments that helped us to improve the quality of the paper. We thank Paul Beaudry for numerous discussions and comments, Kevin Song for his contributions and discussions in the early stages of the project, and Jasmine Hao and Hiro Kasahara for their helpful comments and suggestions. Bruno Esposito Acosta provided outstanding research assistance. Vadim Marmer gratefully acknowledges the financial support of the Social Sciences and Humanities Research Council of Canada under grants 435-2017-0329 and 435-2021-0189. This research was enabled in part by the support provided by Compute Canada.}

\begin{abstract}
{\footnotesize }
 Recurrent boom-and-bust cycles are a salient feature of economic and financial history. Cycles found in the data are stochastic, often highly persistent, and span substantial fractions of the sample size. We refer to such cycles as ``long''. 
 In this paper, we develop a novel approach to modeling cyclical behavior specifically designed to capture long cycles. We show that existing inferential procedures may produce misleading results in the presence of long cycles and propose a new econometric procedure for the inference on the cycle length. Our procedure is asymptotically valid regardless of the cycle length.  
 We apply our methodology to a set of macroeconomic and financial variables for the U.S. We find evidence of long stochastic cycles in the standard business cycle variables, as well as in credit and house prices. However, we rule out the presence of stochastic cycles in asset market data. Moreover, according to our result, financial cycles, as characterized by credit and house prices, tend to be twice as long as business cycles. 

{\footnotesize \ }

\bigskip

\end{abstract}

\maketitle


\input{Introduction}

\input{Model}

\input{Asymptotics}

\input{Extension_DC}

\input{Inference}

\input{Empirical_application}

\appendix
\input{Periodogram}

\input{data_description}
\input{Proofs}

\input{unit_root_tests}

%

%

\input{Size_distortion}

\input{Supplement}

\bibliographystyle{elsart-harv}
\bibliography{long_cycles}


%

\end{document}

%% file: Introduction.tex
\section{Introduction}

This paper develops an econometric framework for inference on the cyclical properties of time series. We are particularly interested in stochastic cycles arising from persistent low-frequency oscillatory impulse responses.  The period of such cycles spans a substantial fraction of a sample, and the econometrician would be able to observe only a handful of peaks and troughs in the data. We refer to such cycles as ``long''. 

Long cycles are prevalent in macroeconomic and financial data. In a recent paper, \cite{beaudry2020aer} estimated that many variables have cycles of approximately 32--40 quarters, corresponding to 15\% and even 20\% of their observed samples.  Using data from 1960 to 2011, \citet{DrehmannBorioTsatsaronis2012} estimated that the length of the credit cycle is 18 years or approximately 35\% of their sample.  The first contribution of our paper is to show that statistical inference may be distorted in such cases. We find that substantial distortions occur when the cycle length exceeds 25\% of the sample size.\footnote{See Appendix \ref{sec:distortions}.}

In our second contribution, we propose a new econometric procedure for inference on the periodicity of cycles. The novel aspect of our methodology is that it is specifically designed to take into account the possibility of long persistent stochastic cycles. Our procedure produces confidence intervals for the cycle length that have the following property: their asymptotic coverage probability is correct regardless of the cycle length. Thus, the confidence intervals are asymptotically valid both
when the period is small relative to the sample size and when it spans a substantial fraction of observed data.
No other procedure in the existing literature has this property. When a data-generating process (DGP) is acyclical, our procedure is expected to produce empty confidence intervals for the cycle length in large samples. Hence, the procedure can be used to rule out cyclical behavior.

Stochastic cycles arise naturally in the AR(2) model  $y_t=\phi_1 y_{t-1}+\phi_2 y_{t-2} +u_t$ with complex roots. For example,  \cite{sargent1987macroeconomic} shows that such a process has a peak in its spectrum in the interior of the $[0,\pi]$ range provided that the autoregressive coefficients satisfy $|\phi_1(1-\phi_2)/4\phi_2|<1$.\footnote{The region with an interior spike in the spectrum is a subset of the region with complex roots  \citep[see][pages 261--265]{sargent1987macroeconomic}} The period of such cycles is determined by the autoregressive coefficients through  ${2\pi}/{\omega}$, where $\omega=\cos^{-1}(-\phi_1(1-\phi_2)/4\phi_2)$ is the spectrum peak frequency. According to this model, the cycle length would amount to only a negligible fraction of the sample size $n$ in large samples: $\frac{2\pi/\omega}{n}\to 0$. Thus,  the long-cycle characteristics of the data are not preserved asymptotically: while in a finite sample the cycle length can represent a substantial fraction of the sample size, it would be negligible in the asymptotic approximation. As a result, conventional asymptotic approximations to the finite sample distributions of estimators and statistics would be inaccurate.

To preserve long-cycle characteristics asymptotically, the spectrum peak frequency must be local to zero in the sense that $n\omega_n$ converges to a positive constant as $n\to\infty$, i.e. the spectrum peak frequency $\omega_n$ ``drifts'' closer to zero with the sample size. This property can be obtained by modeling the two conjugate complex roots of the AR(2) model as local to one, that is, the real and complex parts of the roots ``drift'' closer to one with the sample size. Following the literature \citep[e.g.][]{andrews2020generic}, we refer to such specifications and coefficients as ``drifting'', while specifications with constant parameters independent of $n$ are referred to as ``fixed''.\footnote{Drifting  specifications can be used to verify the uniform size properties of inferential procedures \citep{andrews2020generic}.} The corresponding autoregressive coefficients are too drifting. 

Our resulting long-cycle model is a restricted version of the nearly-twice integrated model in \citet{perron1996useful}, and the restriction is imposed to generate persistent oscillatory behavior. Although the resulting processes are near I(2), they are stationary in finite samples.


Figure \ref{fig:fixed vs drift} illustrates how drifting specifications allow us to preserve the long-cycle feature asymptotically. It shows the difference between the simulated sample paths of cyclical processes generated using the fixed-coefficient  and drifting-coefficient AR(2) DGPs for small and large sample sizes.
 The figure demonstrates that in the model with fixed autoregressive coefficients, by relying on asymptotic approximations, one would distort the cyclical properties of the data. However, the long-cycle properties are preserved in the limit by relying on asymptotic approximations with drifting coefficients.

\begin{figure}[t]
\begin{subfigure}{0.38\textwidth}
\centering
\includegraphics[width=\linewidth]{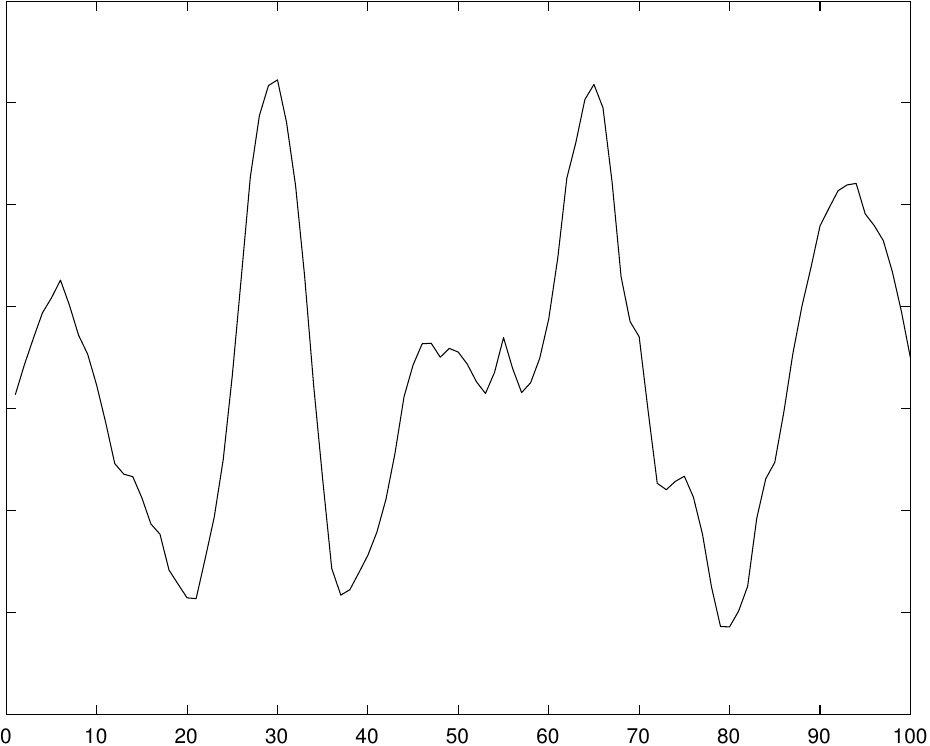}
\caption{fixed coefficients, n = 100}
\end{subfigure} 
\begin{subfigure}{0.38\textwidth}
\centering
\includegraphics[width=\linewidth]{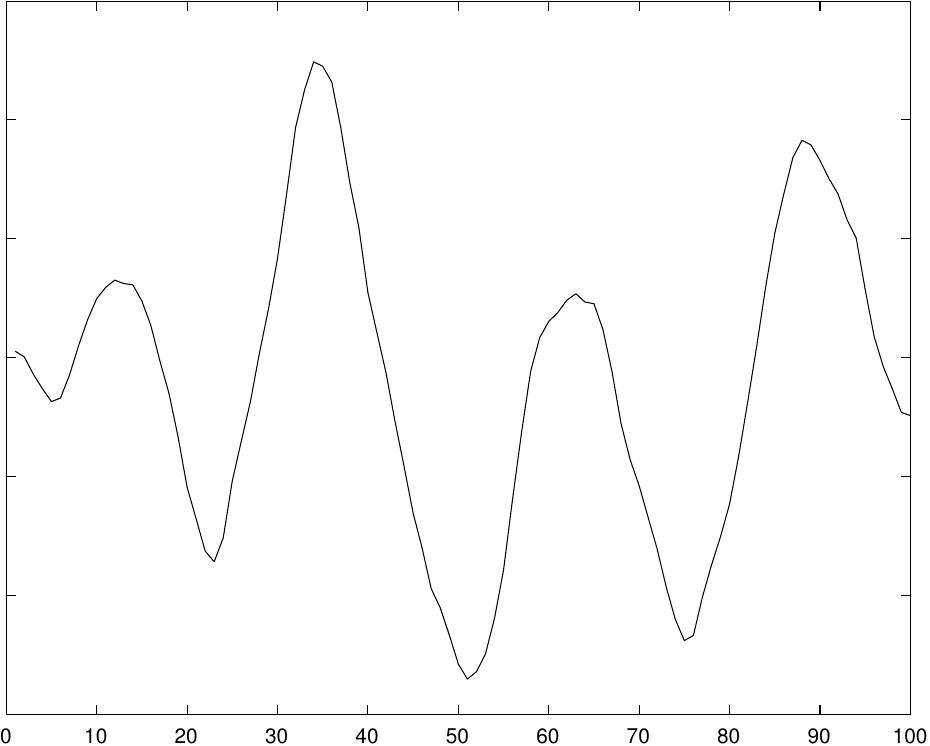}
\caption{drifting coefficients, n = 100}
\end{subfigure}
\begin{subfigure}{0.38\textwidth}
\centering
\includegraphics[width=\linewidth]{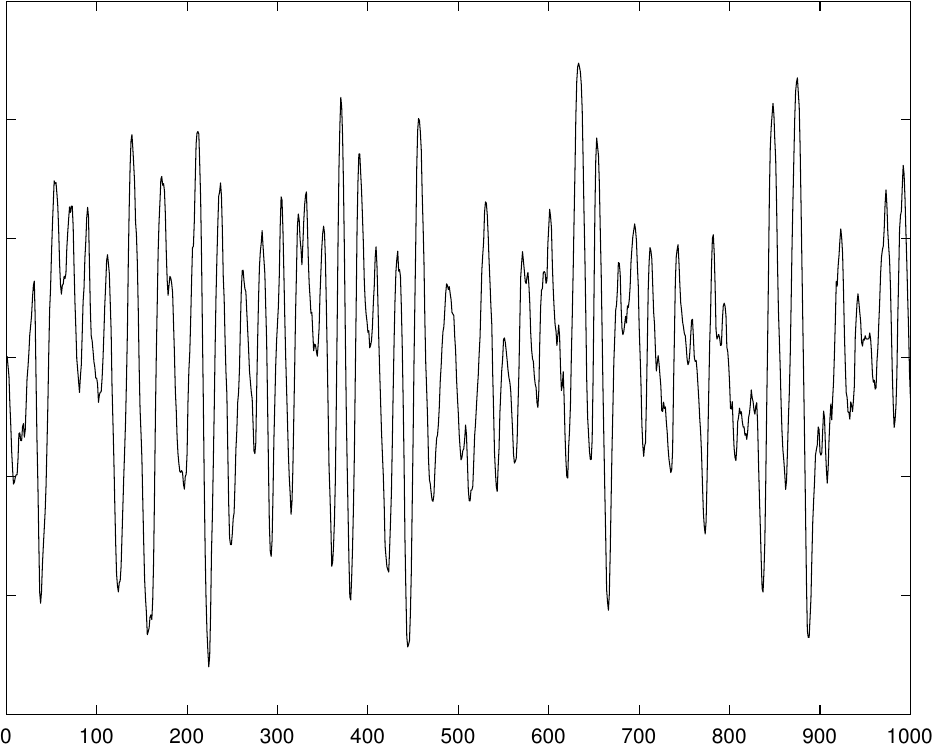}
\caption{fixed coefficients, n = 1000}
\end{subfigure}
\begin{subfigure}{0.38\textwidth}
\centering
\includegraphics[width=\linewidth]{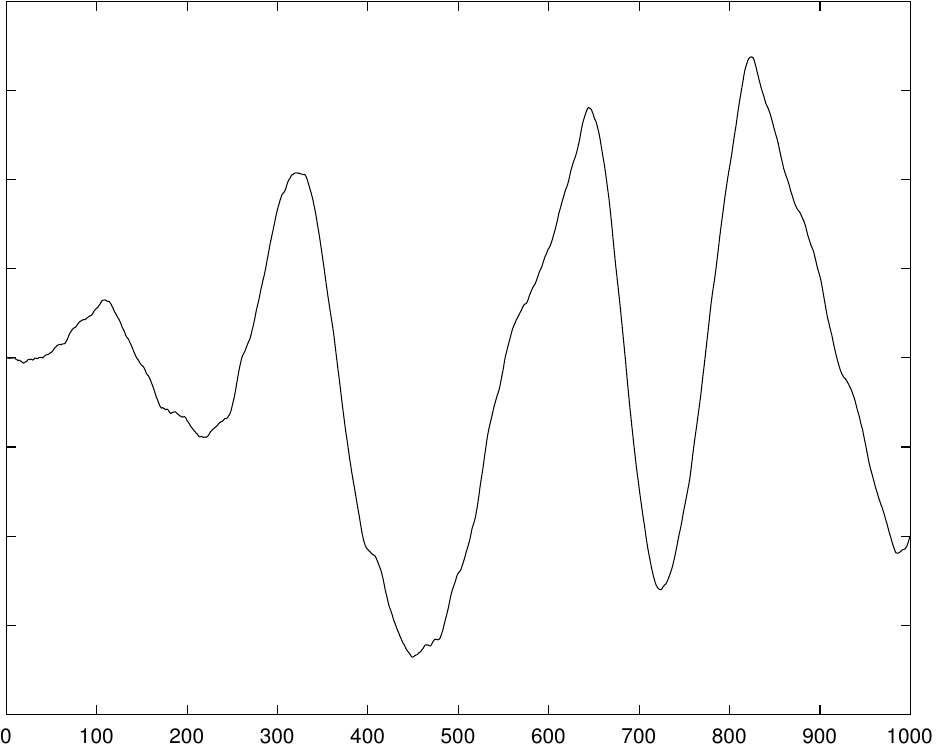}
\caption{drifting coefficients, n = 1000}
\end{subfigure} 
\caption{Time plots of AR(2) processes with fixed and drifting coefficients. In the standard AR(2) specification with fixed coefficients, the period remains the same as the sample size increases. In contrast, the period of an AR(2) process with drifting coefficients grows proportionally  with the sample size. } \label{fig:fixed vs drift}
\end{figure}

The problem is closely related to that in the literature concerned with inference on the largest autoregressive root \citep{stock1991confidence,andrews1993exactly,hansen1999grid,elliott2001confidence,mikusheva2007uniform,mikusheva2012one,dou2019generalized}.\footnote{Equivalently, the sum of the autoregressive coefficients.} 
It has been shown in this literature that when autoregressive roots are close to unity, the conventional asymptotic theory does not provide an accurate approximation to the finite sample distributions of  estimators and statistics. More accurate approximations can be obtained using the so-called   local-to-unity asymptotics developed in \cite{phillips1987towards,phillips1988regression}. However, despite many similarities, the existing results for local-to-unity processes cannot accommodate long cycles because of the presence of two complex conjugate roots. Our third contribution is to develop a novel asymptotic theory for such processes.\footnote{In a recent paper, \cite{dou2019generalized} propose a generalized local-to-unity ARMA model with multiple local-to-unity autoregressive roots balanced local-to-unity roots in the moving average component. They do not consider cyclical behavior, and their limiting distributions are different from those arising in our case.}
 Our results also lay out the foundation for a new econometric framework that, besides the inference on cyclicality, can also be used to study cointegrating long cycles and phase shifts in macro-financial aggregates. 

Our paper is also related to the literature on complex unit roots \citep{bierens2001complex,gregoir2006efficient}. Unlike \citet{bierens2001complex}, our data generating process is stationary in finite samples, and persistent oscillations are achieved through local-to-unity modeling. Local-to-unity modeling with complex roots has been previously considered by \citet{gregoir2006efficient}. However, \citet{gregoir2006efficient} only considers oscillations at fixed frequencies, while we focus on oscillations at local-to-zero frequencies. This crucial feature allows us to accommodate arbitrary long cycles with persistent oscillations at very low frequencies, which is an important attribute of many macroeconomic and financial time series, as demonstrated in Section 7. 

%

The fourth contribution of this paper is empirical, where we implement our procedure to study the cyclical properties of key macroeconomic and financial indicators using U.S. data. Recurrent boom-and-bust cycles are a salient feature of economic and financial history. Long-standing interest in understanding these ups and downs in macro-financial aggregates has led to a vast body of literature on business cycles \citep[see e.g.][]{bry1971front,Harvey1985,AHearnWoitek2001JME,harding2002dissecting,comin2006medium,DrehmannBorioTsatsaronis2012,Aikman_etal2015EJ,Strohsal_etal2019JBF,Runstler2018business}.
Using our methodology, we find that long cycles cannot be ruled out for macroeconomic series, such as the real GDP per capita, the unemployment rates, and the hours per capita. Our results suggest the possibility of cycles that are much longer than those previously reported in the literature. In addition, we find that financial variables such as credit to the nonfinancial sector and home prices exhibit long cycles that are even longer than those for the macro variables. Our results support the position that financial cycles operate at a lower frequency than business cycles. However, our most striking result is that we decisively reject stochastic cycles for asset market variables such as volatility index, credit risk premium, and equity prices. This suggests that the mechanism for asset market fluctuations is different from that of macroeconomic variables and financial variables such as credit and home prices. Importantly, this finding rejects the view suggested in the macro-finance literature that asset prices and economic fluctuations are driven by the same underlying forces: time-varying risk premiums and risk-bearing capacity \citep[see][]{cochrane2017macro}.

The remainder of the paper is organized as follows. In Section \ref{section:model}, we present our modeling approach for long cycles. Section \ref{sec:Asy} presents our core asymptotic results. The results are extended in Section \ref{sec:extensions} to allow for linear time trends and deterministic cycles. Section \ref{sec:inference} describes our procedure for constructing confidence intervals for the cycle length. Section \ref{sec:empirical} presents our empirical results. In Appendix \ref{sec:periodogram}, we show that the periodogram is an asymptotically biased estimator in the presence of long cycles. In Appendix \ref{sec:distortions}, we discuss the size distortion from using conventional $\chi^2$ critical values for inference. In Appendix \ref{sec:BIC}, we show the consistency of the Bayesian Information Criterion (BIC) for the specification of models with long cycles, and then study the finite-sample properties of our proposed inferential procedures using Monte Carlo simulations in \ref{sec}.

%

%% file: Model.tex

\section{A model for long cycles}\label{section:model}

In this section, we present a model for processes that exhibit long cycles. Our objective is to develop a parsimonious modeling approach that allows for cycles with periods spanning nonnegligible fractions of observed samples. More formally, the model should allow the period as a fraction of the sample size $n$ to converge to a nonzero constant as $n\to \infty$.

Following \citet{sargent1987macroeconomic}, we consider the class of autoregressive models. Since in this class cyclical behavior requires complex roots that come as conjugate pairs, the AR(2) model with serially uncorrelated errors\footnote{We extend the approach later in the paper to allow for serially correlated errors.}  is a natural starting point. Thus, consider a process $\{y_t\}$ generated according to
\begin{equation}\label{eq:DGP}
(1-\phi_1 L-\phi_2 L^2) y_t =u_t ,
\end{equation}
where $L$ denotes the lag operator and $\{u_t\}$ is a mean-zero i.i.d. sequence with a finite variance. 
Let $\lambda_1$ and $\lambda_2$ denote the roots of the characteristic equation 
$z^2-\phi_{1} z - \phi_{2} =0$
for the lag polynomial in \eqref{eq:DGP}.
When $|\lambda_1|<1$ and $|\lambda_2|  < 1$,  $\{y_t\}$ has the following MA($\infty$) representation:
\begin{equation*}
y_t = \frac{1}{\lambda_1 - \lambda_2} \sum_{j=0}^{\infty} \left(\lambda_1^{j+1} - \lambda_2^{j+1}\right)u_{t-j}. \label{eq:MA_DGP}
\end{equation*}
Suppose the roots $\lambda_1,\lambda_2$  are complex, and consider their polar form representation:
\begin{equation}
\lambda_1, \lambda_2 = r e^{\pm i \theta} \label{eq:char_roots},
\end{equation}
where $r$ denotes the modulus, $\theta$ is the argument of the complex roots, and $i=\sqrt{-1}$ is the imaginary number.

Given the polar coordinate representation for the roots and the MA($\infty$) representation for the process,
we can write $\{y_t\}$ as
\begin{equation}
y_t =  \sum_{j=0}^{\infty} r^j \frac{\sin(\theta(j+1))}{\sin(\theta)}u_{t-j}. \label{eq:MA_DGP_complex}
\end{equation}
According to \eqref{eq:MA_DGP_complex}, the realized value of $y_t$ is a weighted infinite sum of past realizations of the innovation sequence $\{u_t\}$.  When the characteristic roots are complex, the weights or impulse responses are given by a damped sine wave: the impulse response of $y_t$ to $u_{t-j}$ is
\begin{equation*}
w_j = r^j \frac{\sin(\theta(j+1))}{\sin(\theta)} \label{eq:weights_sine},
\end{equation*}
where the modulus $r$ indicates the rate of decay\footnote{In a more common exponential decay representation, $r^j = e^{\ln(r)j}$. Restricting to processes with non-explosive roots, i.e. $r\leq 1$, we have $\ln(r) \leq 0$ and $-\ln(r)$ is known as the decay constant. } or the persistence of the sine wave, and the argument $\theta$ corresponds to the angular frequency and determines the period of the sine wave. The latter has been used in the literature as a measure of the frequency of the cycle \citep{Harvey1985}.

The stochastic process $\{y_t\}$ inherits its oscillatory behavior precisely from this damped periodic sine weighting function. The closer $r$ is to one, the more persistent $\{y_t\}$, and
the closer $\theta$ is to zero, the lower the oscillating frequency and the longer the length of cycles in $\{y_t\}$. Stochastic cycles are therefore conveniently captured in an AR(2) model with a pair of complex conjugate roots. 

A cyclical process generated according to \eqref{eq:DGP} with roots given by \eqref{eq:char_roots} has the expected cycle length of $2\pi/\theta$. With any fixed parameter value $\theta$, the period as a fraction of the sample size is negligible for large $n$. Hence, asymptotic approximations assuming fixed values for the argument $\theta$ can produce distinctly different cyclical behavior from that observed in finite samples.\footnote{This point is illustrated in Figure \ref{fig:fixed vs drift}.} In other words, 
conventional asymptotics with a fixed complex root argument $\theta$ can distort the cyclical properties of the process. As a result, such asymptotic theory would provide a poor approximation to the actual behavior of the process in finite samples. 
Since the expected period of the process as a fraction of the sample size is given by $2\pi/(\theta n)$, to preserve the cyclical properties  in the limit as $n\to\infty$, one has to consider a drifting sequence of the arguments $\{\theta_n\}$ and allow for $n\theta_n\to d\in[0,\infty]$.\footnote{The approach can still accommodate conventional asymptotics by allowing $n\theta_n \to \infty$.}

We rewrite the AR(2) model in \eqref{eq:DGP} as follows:\footnote{As in \cite{phillips1987towards,phillips1988regression}, the solution to the difference equation \eqref{eq:DGP_n} is a triangular array of the form $\{y_{n,t}:t=1,\ldots,n;n\geq 1\}$. However, we suppress the subscript $n$ to simplify the notation.}
\begin{equation}
(1-\phi_{1,n} L-\phi_{2,n} L^2) y_{t} =u_t,\label{eq:DGP_n}
\end{equation}
where $\phi_{1,n}$ and $\phi_{2,n}$ are now drifting coefficients that can change with $n$. We denote the corresponding characteristic roots by $\lambda_{1,n}$ and $\lambda_{2,n}$ and make the following assumption.
\begin{assumption}\label{a:roots}
The characteristic roots $\lambda_{1,n},\lambda_{2,n}$ associated with the lag polynomial equation in \eqref{eq:DGP_n} are given by
\begin{equation}
\lambda_{1,n}=e^{(c+id)/n},\qquad \lambda_{2,n}=e^{(c-id)/n} \label{eq:char_roots_n}
\end{equation}
where $c\leq 0$ and $d>0$ are fixed localization parameters.
\end{assumption}
\begin{definition}
In what follows, we refer to processes satisfying \eqref{eq:DGP_n} and Assumption \ref{a:roots} as long-cycle.
\end{definition}
Assumption \ref{a:roots} excludes positive values of $c$ as they correspond to explosive roots. Negative values of $d$ can be excluded as $d$ and $-d$ define the same pair of roots.
The autoregressive coefficients are related to the characteristic roots through the following equations:
\begin{align}\
\phi_{1,n} &= \lambda_{1,n}+\lambda_{2,n} = 2 e^{c/n} \cos(d/n),  \label{eq:phi_1n}\\
\phi_{2,n} &= -\lambda_{1,n}\lambda_{2,n}= - e^{2c/n}. \label{eq:phi_2n}
\end{align}
Hence, the autoregressive coefficient $\phi_{1,n}$ is local to 2 while $\phi_{2,n}$ is local to $-1$. The sum of the autoregressive coefficients is local to one, and the process can be mistaken for those considered in the local-to-unity literature. As we discuss below, processes defined by \eqref{eq:phi_1n}--\eqref{eq:phi_2n} exhibit persistent stochastic oscillations and, as a result,  their asymptotic properties are different from those considered in the local-to-unity literature.

The expressions for the characteristic roots in \eqref{eq:char_roots} and  \eqref{eq:char_roots_n} are equivalent since one can replace $r$ and $\theta$ with $r_n\equiv\exp(c/n)$ and  $\theta_n\equiv d/n$ respectively.  The modulus $r_n$ in \eqref{eq:char_roots_n}  has the same representation as the autoregressive parameter in the local-to-unity model in \cite{phillips1987towards,phillips1988regression}. Since there are two roots near one, the process defined by Assumption \ref{a:roots} can be viewed as near I(2). The model is a restricted version of the nearly-twice integrated model of \citet{perron1996useful}, where the restriction is imposed to generate complex roots and persistent cycles.\footnote{\citet{perron1996useful} are primarily concerned with unit root testing and propose modifications designed to improve the size properties of some commonly used tests.}

The localization parameter $d$  controls the length of the cycle, where long cycles correspond to values of $d$ near zero, while  $c$ controls its persistence. Under Assumption \ref{a:roots}, the expected period as a fraction of the sample size is given by
\begin{equation}\label{eq:tau_theta}
\tau_\theta\equiv \frac{2\pi}{n\theta_n  }=\frac{2\pi}{d}.
\end{equation}
With this parameterization, the length of the cycle as a fraction of the sample size is independent of the sample size, and the resulting asymptotic approximations preserve the cyclical properties of the process. 


%


A process with a cyclical factor in its impulse response coefficients may not have visible cyclical oscillations when the damping effect of $r^j$ in \eqref{eq:MA_DGP_complex} is too strong. An alternative but related measure of the periodicity of a process can be constructed from its spectral properties \citep[see e.g.][]{kaiser2001measuring}. The advantage of this measure is that it also takes into account the persistence properties, unlike those based solely on the argument $\theta_n$ of the complex roots. Let $\omega^*_n$ denote the frequency that maximizes the spectral density of the process in \eqref{eq:DGP_n}. As in \citet{sargent1987macroeconomic}, 
\begin{equation*}\label{eq:omega_star}
\omega^*_n = \cos^{-1}\left( -\frac{\phi_{1,n}(1-\phi_{2,n})}{4\phi_{2,n}}\right),
\end{equation*}
provided that 
\begin{equation}\label{eq:peak_cond}
\left| -\frac{\phi_{1,n}(1-\phi_{2,n})}{4\phi_{2,n}}\right|<1.
\end{equation}
The condition in \eqref{eq:peak_cond} together with \eqref{eq:phi_1n}--\eqref{eq:phi_2n} imply that for sufficiently large sample sizes $n$, the spectrum has a peak away from the origin if\footnote{The result in \eqref{eq:peak_cond_cd} follows from \eqref{eq:peak_cond} and \eqref{eq:phi_1n}-\eqref{eq:phi_2n} using second-order expansions of $\cos(d/n)$ and $\exp(c/n)$.}
\begin{equation}\label{eq:peak_cond_cd}
d>|c|.
\end{equation}
The corresponding period of the process as a fraction of the sample size is given by
\begin{equation*}
\tau_{\omega^*_n}\equiv\frac{2 \pi /\omega^*_n}{n}.
\end{equation*}
The following proposition provides an asymptotic approximation of the length of a cycle as a fraction of the sample size when measured using the spectrum-based approach.
\begin{proposition}\label{p:period}
Suppose that $\{y_{n,t}\}$ is generated according to \eqref{eq:DGP_n} with characteristic roots satisfying Assumption \ref{a:roots} and serially uncorrelated $\{u_t\}$ with a finite variance. Suppose further that $d\geq \lvert c\rvert$. Then its spectrum-maximizing frequency $\omega^*_n$ satisfies
\begin{equation*}
n\omega^*_n = \sqrt{d^2-c^2} +O(n^{-2}),
\end{equation*}
and its corresponding spectrum-based period as a fraction of the sample size satisfies
\begin{equation*}
\tau_{\omega^*_n} = \frac{2 \pi}{\sqrt{d^2-c^2}} +O(n^{-2}).
\end{equation*}
\end{proposition}

The proposition shows that, when using spectrum-based measures of the period, the length of a cycle relatively to the sample size can be approximated by 
\begin{equation}\label{eq:tau_omega}
\tau_\omega\equiv\frac{2\pi}{\sqrt{d^2-c^2}}.
\end{equation} 
Unlike the angular frequency-based measure $\tau_\theta$, the spectrum-based measure $\tau_\omega$ takes into account the persistence of the cycle as captured by the localization parameter $c$. Note that larger negative values of $ c < 0$ produce less persistent cycles. In such cases, the spectrum's peak is closer to the origin, and as a result, less persistent processes may not exhibit any visible cyclical behavior. 

In this paper, we consider both $\tau_\theta$ and $\tau_\omega$ since both types of cyclicality measures are used in the literature \citep[see e.g.][]{Harvey1985,sargent1987macroeconomic,kaiser2001measuring}. If one is only concerned with the presence of a cyclical factor in the impulse response coefficients, $\tau_\theta$ is appropriate. However, if in addition one is interested in cyclicality with persistence strong enough to produce a visible spectrum spike away from zero, then more stringent conditions on the autoregressive coefficients are required, and $\tau_\omega$ is more appropriate. The restriction $d\geq |c|$, which is required for $\tau_\omega$ to be defined, ensures that the cycle is persistent enough to produce a peak in the spectrum.

In Appendix \ref{sec:periodogram}, we show that the periodogram-based estimation approach produces biased estimates of $\tau_\omega$. Therefore, in the following we develop an inference procedure for $\tau_\theta$ and $\tau_\omega$ based on the estimates of the autoregressive coefficients $\phi_{1,n}$ and $\phi_{2,n}$. For this purpose, we proceed in two steps. First, we develop a procedure to construct asymptotically valid confidence sets for the autoregressive parameters $\phi_{1,n}$ and $\phi_{2,n}$. In the second step, we use projection arguments to build confidence intervals for the $\tau_\theta$ and $\tau_\omega$ measures of the length of a cycle. The theory developed below also allows serially correlated $\{u_t\}$.

%% file: Asymptotics.tex
\section{Asymptotics for long-cycle processes}\label{sec:Asy}
We now provide the asymptotic theory for the process defined in equations \eqref{eq:DGP_n}--\eqref{eq:char_roots_n}. The theory will be subsequently used to establish the asymptotic distributions of regression-based statistics involving long-cycle time series. It is also necessary to develop robust and asymptotically valid inference about the cyclical properties of a process.

The specification proposed in equations \eqref{eq:DGP_n}--\eqref{eq:char_roots_n} is similar to the first-order autoregressive local-to-unity root model in \citet{phillips1987towards}. Assuming that a process $\{x_{t}\}$ is generated according to $x_{t} = a_n x_{t-1} + u_t$ with $a_n = \exp(c/n)$, \citet{phillips1987towards} shows that the distribution of $\{x_{t}\}$ can be approximated by an Ornstein-Uhlenbeck diffusion process:
\begin{equation}\label{eq:near-unit}
n^{-1/2} x_{\floor{nr}}=n^{-1/2}\sum_{t=1}^{\floor{nr}} e^{c(\floor{nr}-t)/n}u_t\Rightarrow \sigma J_c(r),
\end{equation}
where
\[J_c(r) \equiv \int_0^{r}e^{c(r-s)}\dd W(s),
\]
$r\in[0,1]$, $\floor{x}$ denotes the largest integer less or equal to $x$, $W(\cdot)$ is a standard Brownian motion, $\sigma^2$ denotes the limit of the long-run variance of $\{u_t\}$, ``$\Rightarrow$" denotes the weak convergence of probability measures, and it is assumed that $\{u_t\}$ satisfies a Functional Central Limit Theorem (FCLT). Note that the distribution of the Ornstein-Uhlenbeck process $J_c(r)$ depends on the localization parameter $c$. In what follows, we build on these insights.


We make the following assumption on the innovation sequence $\{u_t\}$.\footnote{Assumption \ref{a:FCLT}  holds, for example, when $\{u_t\}$ is a mixing process such that  $E(u_t) =0$ for all $t$, $\sup_t E|u_t|^{\beta + \epsilon} \leq \infty$ for some $\beta < 2$ and $\epsilon >0$, and $\{u_t\}$ is $\alpha$-mixing of size $-\beta/(\beta-2)$ \citep{phillips1987towards}.
 Alternatively, it holds when $\{u_t\}$ is a linear MA($\infty$) process satisfying the conditions in \citet[][Theorem 3.15]{phillips1992asymptotics}.}
\begin{assumption}[FCLT]\label{a:FCLT} 
 Let $W(\cdot)$ denote the standard Brownian motion, and let $\sigma^2$ be the limit of the long-run variance of $\{u_t\}$: $\sigma^2 \equiv \lim_{n\to\infty} Var(n^{-1/2}\sum_{t=1}^n u_t)$.  Then
 for $r\in[0,1]$,  
\begin{equation*}
n^{-1/2}\sum_{t=1}^{\floor{nr} }u_t \Rightarrow \sigma W(r).
\end{equation*}
\end{assumption}

In the case of long cycles, a different limiting process arises from that of the local-to-unity case, with cyclicality reflected by the sine function. However, similarly to the results in \citet{phillips1987towards}, the process can be described as an integral with respect to a Brownian motion and depends on the localization parameters $c$ and $d$. We define:
\begin{equation}\label{eq:limit_process}
J_{c,d}(r)\equiv  \frac{1}{d} \int_{0}^{r} e^{c(r-s)}\sin(d(r-s))\dd W(s).
\end{equation}
The next proposition shows that in large samples and after appropriate scaling, the distribution of a long-cycle process can be approximated by that of $J_{c,d}(\cdot)$.
\begin{proposition}\label{prop:convergence_J}
Suppose that $\{y_{t}\}$ is generated according to equation \eqref{eq:DGP_n}, and Assumptions \ref{a:roots} and  \ref{a:FCLT} hold. Then,
\begin{equation*}\label{eq:limit_process_J}
n^{-3/2} y_{{\floor{nr}}} \Rightarrow \sigma J_{c,d}(r).
\end{equation*}
\end{proposition}
\begin{proof}
The solution to \eqref{eq:DGP_n} can be expressed in terms of the characteristic roots as
\begin{eqnarray} 
y_{t} &=&  \frac{1}{\lambda_{n,1}-\lambda_{n,2}}\displaystyle\sum_{k=1}^{t} \left(\lambda_{n,1}^{t-k+1} - \lambda_{n,2}^{t-k+1}\right)u_k \notag\\
&=& \frac{1}{2i\cdot e^{c/n}\sin(d/n)} \displaystyle\sum_{k=1}^{t} \left(e^{(c+id)(t-k+1)/n} - e^{(c-id)(t-k+1)/n}\right)u_k, \label{aeq:AR_DGP_n_roots_sum} \notag
\end{eqnarray}
where the second equality follows by Assumption \ref{a:roots}. By Assumption \ref{a:FCLT} and as in \eqref{eq:near-unit},
\begin{eqnarray}
&& n^{-1/2}\sum_{k=1}^{\floor{nr}} (e^{(c+id)(t-k+1)/n} - e^{(c-id)(t-k+1)/n})u_k \notag\\
&  \Rightarrow & \sigma \int_0^r \left(e^{(c+id)(r-s)}-e^{(c-id)(r-s)}\right)\dd W(s) \notag\\
& = & 2i \sigma\int_0^r e^{c(r-s)}\sin(d(r-s))\dd W(s).\label{aeq:thediff}\notag
\end{eqnarray}
The result follows 
since $\sin(d/n) = d/n + O(n^{-2})$.
\end{proof}

The continuous-time Gaussian process $J_{c,d}(\cdot)$ plays the central role in our analysis.
It can be viewed as a continuous-time version of the MA($\infty$) representation in \eqref{eq:MA_DGP_complex}: past shocks are weighted by a damped sine wave. Again, the parameters $c$ and $d$ control the persistence and frequency of the cycle, respectively. Note also that long-cycle processes require stronger scaling than local-to-unity: $n^{-3/2}$ instead of $n^{-1/2}$. This is a reflection of the fact that long-cycle processes are near I(2).  

We now turn to the properties of the least-squares estimators and the corresponding test statistics for the second-order autoregressive model with long cycles.  Let $\widehat{\phi}_{1,n}$ and $\widehat{\phi}_{2,n}$ denote the least-squares estimator of \eqref{eq:DGP_n}:
\begin{equation}\label{eq:OLS}
\begin{pmatrix}
\widehat{\phi}_{1,n} \\
\widehat{\phi}_{2,n} 
\end{pmatrix}
 = \begin{pmatrix}
\sum y^2_{t-1}  &\sum y_{t-1}y_{t-2}\\
\sum y_{t-1}y_{t-2} & \sum y^2_{t-2}
\end{pmatrix}^{-1} \begin{pmatrix}
\sum y_{t-1}y_t \\
\sum y_{t-2}y_t
\end{pmatrix}. 
\end{equation}
As it turns out, the matrix on the right-hand side is asymptotically singular \label{singularity} because all three elements $\sum y^2_{t-1}$, $\sum y^2_{t-2}$, and $\sum y_{t-1}y_{t-2}$ converge to the same random limit when properly scaled. This is because $\sum y_{t-1}y_{t-2}=\sum y^2_{t-1}+$ smaller order terms, which follows formally from Lemma  \ref{lem:convergence_moments}(b) below. The singularity complicates the derivation of the limiting distributions of the estimators and the corresponding test statistics.

To eliminate the singularity arising in the limit, we consider the following transformation of the equation in \eqref{eq:DGP_n}: 
\begin{equation}
y_t = (\phi_{1,n} +\phi_{2,n}) y_{t-1} - \phi_{2,n}\Delta y_{t-1} + u_t, \label{eq:ARn_transformed}
\end{equation}
where  $\Delta y_{t-1} = y_{t-1} - y_{t-2}$. Since \eqref{eq:ARn_transformed} is obtained from the original equation through a non-singular linear transformation of the regressors and parameters, the OLS estimator of $\phi_{1,n} +\phi_{2,n}$ is given by $\widehat{\phi}_{1,n} + \widehat{\phi}_{2,n}$. Moreover, the usual Wald test statistic for testing joint hypotheses about $\phi_1$ and $\phi_2$ is the same for both regressions. Thus, we have the following:
\begin{eqnarray}
\lefteqn{\begin{pmatrix}
\widehat{\phi}_{1,n}+\widehat{\phi}_{2,n} - \phi_{1,n} -\phi_{2,n} \\
\widehat{\phi}_{2,n} - \phi_{2,n}
\end{pmatrix}
 =} \notag\\
  && \quad \begin{pmatrix}
\sum y^2_{t-1}  & -\sum y_{t-1}\Delta y_{t-1}\\
-\sum y_{t-1}\Delta y_{t-1} & \sum (\Delta y_{t-1})^2
\end{pmatrix}^{-1} 
\begin{pmatrix}
\sum y_{t-1}u_t \\
-\sum \Delta y_{t-1}u_t
\end{pmatrix}. \label{eq:phi_12}
\end{eqnarray}
As we show below, the matrix on the right-hand side of \eqref{eq:phi_12} is no longer singular in the limit.

It follows from the representation in \eqref{eq:phi_12} that the asymptotic theory of the OLS estimator involves the sample moments of  $(y_{t-1},\Delta y_{t-1})$. Therefore, in addition to the asymptotic approximation of $y_{t-1}$, we also need the asymptotic approximation for $\Delta y_{t-1}$. The latter involves two additional continuous-time processes. We define:
\begin{align}
K_{c,d}(r) &  \equiv \frac{1}{d} \int_{0}^{r} e^{c(r-s)}\cos(d(r-s))\dd W(s), \notag\\
G_{c,d}(r) & \equiv c\cdot J_{c,d}(r) + d \cdot K_{c,d}(r). \label{eq:G}
\end{align}
Note that the diffusion process $K_{c,d}(r)$ is akin to the process $J_{c,d}(r)$ except that it is defined with a cosine function instead of a sine function. The next proposition shows that in large samples and after scaling, the distribution of $\Delta y_{\floor{nr}}$ can be approximated by that of $G_{c,d}(r)$.

\begin{proposition}\label{prop:convergence_G}
Suppose that $\{y_{t}\}$ is generated according to equation \eqref{eq:DGP_n}, and Assumptions \ref{a:roots} and  \ref{a:FCLT} hold. Then,
\begin{equation*}\label{eq:limit_process_G}
n^{-1/2} \Delta y_{{\floor{nr}}} \Rightarrow \sigma G_{c,d}(r),
\end{equation*}
where the result holds jointly with that in Proposition \ref{prop:convergence_J}.
\end{proposition}
Note that in contrast to the scaling $n^{-3/2}$ applied to $y_{t-1}$, its first difference $\Delta y_{t-1}$ requires scaling by $n^{-1/2}$.  Therefore, the first differences of long-cycle processes have convergence rates of $O(n^{1/2})$ tantamount to those of local-to-unity processes. However, due to cyclicality, the large-sample distribution of $\Delta y_{t-1}$ differs from that in the local-to-unity model.

Based on the results of Proposition \ref{prop:convergence_J} and \ref{prop:convergence_G}, we can now provide the asymptotic theory for the sample moments of long-cycle processes. Parts of the lemma below require the following ergodicity property for $\{u_t\}$.
\begin{assumption}\label{a:var} 
 Let $\sigma^2_u \equiv \lim_{n\to\infty} n^{-1} \sum_{t=1}^n Eu_t^2$ be the average variance of $\{u_t\}$ over time. We assume that 
  $n^{-1} \sum_{t=1}^n u^2_t \to_p \sigma^2_u$.
\end{assumption}

\begin{lemma}
\label{lem:convergence_moments}
Suppose that $\{y_{t}\}$ is generated according to equation \eqref{eq:DGP_n}, and Assumptions \ref{a:roots} and  \ref{a:FCLT} hold. The following results hold jointly.
\begin{enumerate}[(a)]
\item $n^{-4} \sum y^2_{t-1} \Rightarrow \sigma^2 \int_0^1 J^2_{c,d}(r) \dd r$.
\item $n^{-3} \sum y_{t-1}\Delta y_{t-1} \Rightarrow \sigma^2 \int_0^1 J_{c,d}(r)G_{c,d}(r) \dd r$.
\item $n^{-2} \sum (\Delta y_{t-1})^2 \Rightarrow \sigma^2 \int_0^1 G^2_{c,d}(r) \dd r$.
\end{enumerate}
Suppose in addition that Assumption \ref{a:var} holds. The following results hold jointly with (a)--(c).
\begin{enumerate}
\item[(d)] $n^{-2} \sum y_{t-1}u_t \Rightarrow \sigma^2 \int_0^1 J_{c,d}(r)\dd W(r)$.
\item[(e)]  $n^{-1} \sum \Delta y_{t-1}u_t \Rightarrow \sigma^2 \int_0^1 G_{c,d}(r)\dd W(r)+ \frac{1}{2}(\sigma^2 - \sigma^2_u)$.
\end{enumerate}
\end{lemma}

Note that in part (e) of the lemma, the limiting distribution of the sample covariance between $\Delta y_{t-1}$ and $u_t$ depends on the difference between the long-run and the average over time variances of $\{u_t\}$. This reflects the serial correlation in $\{u_t\}$ and is standard in the unit root literature. However, despite the serial correlation, the difference $\sigma^2 -\sigma^2_u$ does not appear in the limiting expressions in part (d) for the sample covariance between $y_{t-1}$ and $u_t$. This is due to the stronger scaling factor required for the near I(2) long-cycle process $\{y_t\}$. 

To simplify the notation, in the rest of the paper we use \label{notation}
$\int J_{c,d}^2 $ to denote $\int_0^1 J_{c,d}^2(r) \dd r$ and $\int J_{c,d} \dd W$ to denote $\int_0^1 J_{c,d}(r) \dd W(r)$. We use the same convention for
the integral expressions with $G_{c,d}(r)$ with $J_{c,d}$ replaced by $G_{c,d}$. Lastly, we use $\int J_{c,d}G_{c,d}$ to denote $\int_0^1 J_{c,d}(r) G_{c,d}(r) \dd r$.

Equipped with the results of Lemma \ref{lem:convergence_moments}, we can now describe the asymptotic distribution of the least-squares estimators of $\phi_{1,n}$ and $\phi_{2,n}$. 
\begin{proposition} \label{prop:LSphi}
Suppose that $\{y_{n,t}\}$ is generated according to equation \eqref{eq:DGP_n}, and Assumptions \ref{a:roots}, \ref{a:FCLT}, and \ref{a:var} hold. The following results hold jointly with the results of Lemma \ref{lem:convergence_moments}.
\begin{enumerate}[(a)]
\item $n^2(\widehat{\phi}_{1,n} +\widehat{\phi}_{2,n}- \phi_{1,n}-\phi_{2,n}) \Rightarrow $
\begin{equation*}
\frac{\int G_{c,d}^2 \cdot \int J_{c,d} \dd W -  \bigg(\int G_{c,d} \dd W + \frac{1}{2}(1 - \sigma^2_u/\sigma^2)\bigg)\cdot \int J_{c,d} G_{c,d} }{\int J^2_{c,d}  \cdot \int G^2_{c,d}  - (\int J_{c,d} G_{c,d} )^2}.
\end{equation*}
\item $ n \begin{pmatrix}
\widehat{\phi}_{1,n} - \phi_{1,n}\\
\widehat{\phi}_{2,n} - \phi_{2,n}
\end{pmatrix} \Rightarrow $
\begin{equation*}
\begin{pmatrix}
-1\\
1
\end{pmatrix} \times 
 \frac{   \int J_{c,d}G_{c,d} \cdot \int J_{c,d}\dd W-\bigg(\int G_{c,d}\dd W  + \frac{1}{2}(1 - \sigma^2_u/\sigma^2)\bigg)\cdot \int J^2_{c,d}}{\int J^2_{c,d} \cdot \int G_{c,d}^2  - \big(\int J_{c,d}G_{c,d} \big)^2}.
\end{equation*}
\end{enumerate}
\end{proposition}

According to part (b) of the proposition, the joint asymptotic distribution of the least-squares estimators of $\phi_{1,n}$ and $\phi_{2,n}$ is singular and determined by the same random variable. Furthermore, their convergence rate is $O_p(n^{-1})$ despite that the process $\{ y_t \}$ is near I(2). This is a consequence of the asymptotic singularity in \eqref{eq:OLS} as previously discussed on page \pageref{singularity}. However, in part (a) of the proposition, the least-squares estimator of $\phi_{1,n}+\phi_{2,n}$ has the faster convergence rate $O_p(n^{-2})$ characteristic of I(2) processes. Note that the limiting distributions depend on the localization parameters $c$ and $d$.

In Section \ref{sec:inference} below and following \citet{hansen1999grid}, we consider a grid-based approach for constructing confidence sets for the autoregressive coefficients. Since  \citet{hansen1999grid} relies on the $t$-statistic for his procedure for the largest autoregressive root, we construct our procedure around the Wald statistic.\footnote{The approach in \citet{elliott2001confidence} can be used to construct alternative statistics.} 

Consider testing a joint hypothesis $H_0:\phi_1=\phi_{1,0}, \phi_2=\phi_{2,0}$ against $H_1:\phi_1\ne\phi_{1,0}\;\text{or}\; \phi_2\ne\phi_{2,0}$. Construction of grid-based confidence sets involves testing a sequence of hypothesis with different null values $\phi_{1,0},\phi_{2,0}$ and then collecting those that are not rejected. 
The usual Wald statistic is given by
\begin{equation}\label{eq:Wald}
W_n(\phi_{1,0},\phi_{2,0}) \equiv \begin{pmatrix}
\widehat{\phi}_{1,n} + \widehat{\phi}_{2,n} - \phi_{1,0}- \phi_{2,0}\\
\widehat{\phi}_{2,n} - \phi_{2,0}
\end{pmatrix}^\top    \widehat{V}^{-1}_n\begin{pmatrix}
\widehat{\phi}_{1,n} + \widehat{\phi}_{2,n} - \phi_{1,0}- \phi_{2,0}\\
\widehat{\phi}_{2,n} - \phi_{2,0}
\end{pmatrix},
\end{equation}
where
\begin{align*}
\widehat{V}_n &\equiv \widehat{\sigma}^2_n \begin{pmatrix}
\sum y^2_{t-1}  & -\sum y_{t-1}\Delta y_{t-1}\\
-\sum y_{t-1}\Delta y_{t-1} & \sum (\Delta y_{t-1})^2
\end{pmatrix}^{-1},
\end{align*}
and $\widehat{\sigma}^2_n$ is a consistent estimator of the long-run variance $\sigma^2$ constructed using $\hat u_t =y_t-\widehat \phi_{1,n} y_{t-1} - \widehat \phi_{2,n} y_{t-2}$, see \cite{NWK/WKD'87} and \cite{ADWK'91}. The infeasible estimator of $\sigma^2$ that uses $u_t$ is constructed as
$\tilde \sigma^2_n=n^{-1}\sum_{t=1}^n u_t^2 + 2\sum_{h=1}^{m_n} w_n(h) n^{-1}\sum_{t=h+1}^n u_t u_{t-h}$, 
where $m_n=o(n)$ is the lag truncation parameter, and $w_n(\cdot)$ is a bounded weight function such that $\lim_{n\to\infty}w_n(h)=1$ for all $h$. The feasible estimator $\widehat\sigma^2_n$ is constructed similarly using the estimated residuals $\hat u_t$ instead of $u_t$. We make the following assumption.
\begin{assumption}\label{a:lrvar} The infeasible estimator $\tilde \sigma^2_n$ of the long-run variance $\sigma^2$ is consistent: $\tilde \sigma^2_n \to_p \sigma^2$. 
\end{assumption}
The conditions for consistency of the infeasible estimator can be found in \cite{NWK/WKD'87} and \cite{ADWK'91}.  Our next result describes the asymptotic null distribution of the Wald statistic for long-cycle processes.
\begin{proposition} \label{prop:wald}
Suppose that $\{y_{n,t}\}$ is generated according to equation \eqref{eq:DGP_n}, Assumptions \ref{a:roots} and  \ref{a:FCLT}-\ref{a:lrvar} hold, and $m_n=o(n)$. Then,
\begin{align*}
W_n(\phi_{1,n},\phi_{2,n}) \Rightarrow
\frac{\bigintsss \bigg\{J_{c,d} \cdot \bigg(\int G_{c,d} \dd W + \frac{1}{2}(1 - \sigma^2_u/\sigma^2)\bigg)- G_{c,d}  \cdot \int J_{c,d}\dd W\bigg\}^2 }{\int J^2_{c,d}\cdot \int G^2_{c,d} - \big(\int J_{c,d}G_{c,d} \big)^2}. 
\end{align*}
\end{proposition}

The asymptotic null distribution of the Wald statistic is non-standard and non-pivotal: it depends on the ratio of the average over time and long-run variances $\sigma^2_u/\sigma^2$, and on the unknown localization parameters $c$ and $d$. While the ratio $\sigma^2_u/\sigma^2$ does not play a role when $\{u_t\}$ are serially uncorrelated and can be estimated consistently otherwise,\footnote{See the proof of Proposition \ref{prop:wald}.} the dependence on $c$ and $d$ remains. 
Hence, the quantiles of the limiting distribution can only be simulated given the values of $c$ and $d$.

In Appendix \ref{sec:distortions}, we discuss the differences between the conventional $\chi^2_2$ critical values and the quantiles of the asymptotic distribution in Proposition \ref{prop:wald}. Depending on the values of $c$ and $d$, the differences can be substantial, especially when the cycle length exceeds 25\% of the sample size and the model includes deterministic components that are introduced in the next section.

%% file: Extension_DC.tex

\section{Extensions to models with deterministic components}\label{sec:extensions}
For practical applications, it is important to allow the DGP to include nonzero means, trends, and deterministic cycles. We discuss such extensions in this section. As the results below show, the limiting distributions of the regression estimators and test statistics take a similar form to those in Section \ref{sec:Asy}, but with $J_{c,d}$ and $G_{c,d}$ replaced with their residuals from appropriate continuous-time projections. This property is standard in the unit-root literature and continues to hold in our case.

Formally, we assume that the data $\{y_t: t = 1,\ldots,n\}$ are generated according to 
\begin{align}
(1-\phi_{1,n} L-\phi_{2,n} L^2)(y_t -D_t)= u_t, \label{eq:yc+d}
\end{align}
where  $D_t$ is non-random, can vary with $t$, and depends on unknown parameters.  To control for deterministic regressors, estimation of the autoregressive coefficients requires projecting against the components of $D_t$. The asymptotic distributions of the estimators and test statistics change accordingly. We consider the following three formulations of $D_t$:\begin{enumerate}[(i)]
\item Constant mean: $D_t = \mu$ for some unknown parameter $\mu$.
\item Deterministic cycles: $D_t = \mu + \sum_k \{ \eta_{1k} \cos(2\pi k t/n) + \eta_{2k} \sin(2\pi k t/n) \}$, where $k$'s are known positive integers, and $\eta_{1k},\eta_{2k}$ are unknown coefficients.
\item Linear time trend: $D_t = \mu + \xi t/n$, where $\xi$ is the unknown coefficient.
\end{enumerate}

The specification\footnote{We do not consider specifications in which deterministic components may dominate the stochastic long-cycle component.} in (i)  allows $\{y_t\}$ to have a constant over time with a nonzero mean. The DGP in (ii) can be used, for example, to distinguish between very low frequency fluctuations and long cycles, as many time series in economics exhibit such patterns, see \cite{beaudry2020aer}.\footnote{We thank Paul Beaudry for pointing our attention to this fact.} The $D_t$ component in (ii) generates cosine and sine oscillations at frequencies $2\pi k/n$. The period of such oscillations relative to the sample size is $1/k$, and they can capture very low-frequency cycles in data that are outside the range of interest of the econometrician. For practical purposes, we consider $k=1,2,3$. Inclusion of such components can be viewed as detrending of data by removing fluctuations at the frequencies corresponding to the values of $k$. The asymptotic results developed in this section can be used to account for detrending in inferential procedures.

The DGP in (iii) allows for linear time trends, and such adjustments have a long history in the unit root literature. The division by $n$ is required to derive the asymptotic properties and can be absorbed into the unknown coefficient $\xi$. Hence, observationally, the model in (iii) is identical to the model with no adjustment by $n$. 

The empirical application in Section \ref{sec:empirical} also considers the case where $D_t$ consists of seasonal dummies and a constant. However, as shown in \cite{phillips2002kpss} for unit root testing, the arising asymptotic distributions have the same form as those in the constant mean case.

%
%

As in the previous section and to avoid singularities in the limit, we use the  transformed version of the model with $y_{t-1}$ and $\Delta y_{t-1}$:
\begin{equation}
y_{t} = (\phi_{1,n} +\phi_{2,n})y_{t-1} - \phi_{2,n}\Delta y_{t-1} + (1-\phi_{1,n} L-\phi_{2,n} L^2)D_t + u_t. \label{eq:ARn_dc_transformed}
\end{equation}

\subsection{Constant mean}\label{subsec:mean}
In this section, we consider case (i) of a constant unknown mean. When $D_t=\mu$, equation \eqref{eq:ARn_dc_transformed} becomes
\begin{equation} 
y_{t} = \alpha_{n} + (\phi_{1,n} +\phi_{2,n})y_{t-1} - \phi_{2,n}\Delta y_{t-1} + u_t \label{eq:ARn_mean},
\end{equation}
where $\alpha_n \equiv (1 - \phi_{1,n} - \phi_{2,n})\mu=O(n^{-2})$.\footnote{See Lemma \ref{alem:phi_taylor} in the Appendix.} 
Let $\widehat{\phi}_{1,n}$ and $\widehat{\phi}_{2,n}$ be the least-squares estimator of the corresponding coefficients in \eqref{eq:ARn_mean}, and define $\widetilde{y}_{t-1} = y_{t-1} - \bar{y}$ and $\widetilde{\Delta {y}}_{t-1} = \Delta y_{t-1} - \overline{\Delta {y}}$, where $\bar{y}_n$ and $\overline{\Delta {y}}_n$ denote the sample averages of $y_{t-1}$ and $\Delta y_{t-1}$ respectively. Then,
\[
\begin{pmatrix}
\widehat{\phi}_{1,n}+\widehat{\phi}_{2,n} - \phi_{1,n}-\phi_{2,n}\\
\widehat{\phi}_{2,n} - \phi_{2,n}
\end{pmatrix} = \begin{pmatrix}
\sum \widetilde{y}_{t-1}^2 & - \sum\widetilde{y}_{t-1}\widetilde{\Delta {y}}_{t-1} \\
- \sum\widetilde{y}_{t-1}\widetilde{\Delta {y}}_{t-1} & \sum \widetilde{\Delta {y}}_{t-1}^2
\end{pmatrix}^{-1} \begin{pmatrix}
\sum \widetilde{y}_{t-1}u_t \\
-\sum \widetilde{\Delta {y}}_{t-1}u_t
\end{pmatrix}. 
\]
and we have the following analogue of Lemma \ref{lem:convergence_moments}.

\begin{lemma} \label{lem:convergence_moments_mean}
Suppose that $\{y_{t}\}$ is generated according to equation \eqref{eq:ARn_mean}, and Assumptions \ref{a:roots} and  \ref{a:FCLT} hold. 	Define
$\widetilde{J}_{c,d}(r) \equiv J_{c,d}(r) - \int_0^1 J_{c,d}(s) \dd s$ and $\widetilde{G}_{c,d} (r)\equiv G_{c,d}(r)-\int_0^1 G_{c,d}(s) \dd s$.
The following results hold jointly.
\begin{enumerate}[(a)]
\item $n^{-4} \sum \widetilde{y}_{t-1}^2 \Rightarrow \sigma^2\int \widetilde{J}^2_{c,d}$.
\item $n^{-2} \sum \widetilde{\Delta {y}}_{t-1}^2 \Rightarrow \sigma^2\int \widetilde{G}^2_{c,d}$.
\item $n^{-3} \sum \widetilde{y}_{t-1}\widetilde{\Delta{y}}_{t-1} \Rightarrow \sigma^2 \int \widetilde{J}_{c,d}\widetilde{G}_{c,d} $.
\end{enumerate}
In addition, suppose that Assumption \ref{a:var} holds. The following results hold jointly with (a)--(c).
\begin{enumerate}
\item[(d)] $n^{-2} \sum \widetilde{y}_{t-1}u_t \Rightarrow \sigma^2 \int \widetilde{J}_{c,d}\dd W $.
\item[(e)] $n^{-1} \sum \widetilde{\Delta {y}}_{t-1}u_t \Rightarrow \sigma^2 \int \widetilde{G}_{c,d}\dd W+ \frac{1}{2}(\sigma^2 - \sigma^2_u)$.
\end{enumerate}
\end{lemma}

The results in Lemma \ref{lem:convergence_moments_mean} are parallel to those in Lemma  \ref{lem:convergence_moments}. However, instead of $J_{c,d}$ and $G_{c,d}$, the distributions arising in the limit depend on $\widetilde J_{c,d}$ and $\widetilde G_{c,d}$. Note that the latter processes are obtained from $J_{c,d}$ and $G_{c,d}$ by subtracting their respective continuous-time averages, which matches the construction of $\widetilde y_t$ and $\widetilde {\Delta y}_t$ in finite samples.

\subsection{Deterministic cycles}\label{subsec:det_cycle}
In this section, we consider case (ii) of deterministic cycles. When $D_t$ includes deterministic cycles, equation \eqref{eq:ARn_dc_transformed} takes the form
\begin{eqnarray}
\lefteqn{y_t = \alpha_n + \sum_k \left\{ \gamma_{1k,n} \cos\left(\frac{2\pi k t}{n}\right) + \gamma_{2k,n} \sin\left(\frac{2\pi k t}{n}\right)\right\} } \qquad\qquad\qquad \notag \\
&&\qquad\qquad\mbox{}+ (\phi_{1,n}+\phi_{2,n})y_{t-1} - \phi_{2,n}\Delta y_{t-1} + u_t, \label{eq:ARn_cycles}
\end{eqnarray}
where the intercept $\alpha_n$ is as defined in the case of a constant mean. The lags of the cosine and sine components can be written as linear combinations of $\cos(2\pi k t/n)$ and $\sin(2 \pi k t/n)$ with coefficients depending on $n$ and, therefore, can be omitted. 
The least-squares estimators $\widehat{\phi}_{1,n}$ and $\widehat{\phi}_{2,n}$ can be obtained by estimating
\begin{equation*}
{y}_t = (\phi_{1,n}+\phi_{2,n})\widetilde{y}_{t-1} - \phi_{2,n}\Delta \widetilde{y}_{t-1} + {u}_{t},
\end{equation*}
where $\widetilde{y}_{t-1}$ and $\widetilde{\Delta {y}}_{t-1}$ are the residuals from the regressions of  $y_{t-1}$ and $\Delta y_{t-1}$ respectively on  $\cos(2\pi k t/n)$, $\sin(2 \pi k t/n)$, and a constant.

The following result describes the asymptotic distributions of the sample moments of $\widetilde y_{t-1}$, $\widetilde{\Delta y}_{t-1}$, and $u_t$. 
\begin{lemma} \label{lem:convergence_moments_cycles}
Suppose that $\{y_{t}\}$ is generated according to equation \eqref{eq:ARn_cycles}, and Assumptions \ref{a:roots} and  \ref{a:FCLT} hold. Define
\begin{align*}
\widetilde{J}_{c,d}(r) &\equiv  J_{c,d}(r) - \int_0^1 J_{c,d}(s)\dd s - \sum_k \left \{\psi_{1k}\cos(2\pi kr) -  \psi_{2k}\sin(2\pi kr) \right \},\\
\widetilde{G}_{c,d}(r) &\equiv  G_{c,d}(r) - \int_0^1 G_{c,d}(s)\dd s -\sum_k \left \{\varphi_{1k}\cos(2\pi kr) -  \varphi_{2k}\sin(2\pi kr) \right \},
\end{align*}
where
\begin{alignat*}{3}
\psi_{1k}  &\equiv 2\int_0^1 \cos(2 \pi k s) J_{c,d}(s) \dd s,\quad
\psi_{2k} &&\equiv 2\int_0^1 \sin(2 \pi k s) J_{c,d}(s)\dd s,\\
\varphi_{1k}  &\equiv 2\int_0^1 \cos(2 \pi k s) G_{c,d}(s) \dd s, \quad
\varphi_{2k} &&\equiv 2\int_0^1 \sin(2 \pi k s) G_{c,d}(s)\dd s.
\end{alignat*}
The following results hold jointly.
\begin{enumerate}[(a)]
\item $n^{-4} \sum \widetilde{y}_{t-1}^2 \Rightarrow \sigma^2\int \widetilde{J}_{c,d}^2$.
\item $n^{-2} \sum \Delta \widetilde{y}_{t-1}^2 \Rightarrow \sigma^2\int \widetilde{G}_{c,d}^2$.
\item $n^{-3} \sum \widetilde{y}_{t-1}\Delta \widetilde{y}_{t-1} \Rightarrow \sigma^2 \int \widetilde{J}_{c,d}\widetilde{G}_{c,d} $.
\end{enumerate}
In addition, suppose that Assumption \ref{a:var} holds. The following results hold jointly with (a)-(c).
\begin{enumerate}
\item[(d)] $n^{-2} \sum \widetilde{y}_{t-1}u_t \Rightarrow \sigma^2 \int \widetilde{J}_{c,d}\dd W $.
\item[(e)] $n^{-1} \sum \Delta\widetilde{y}_{t-1}u_t \Rightarrow \sigma^2 \int \widetilde{G}_{c,d}\dd W+ \frac{1}{2}(\sigma^2 - \sigma^2_u)$.
\end{enumerate}
\end{lemma}

Lemma \ref{lem:convergence_moments_cycles} is the analogue of Lemma \ref{lem:convergence_moments} for the model with deterministic cycles. The coefficients $\psi_{1,k}$  and $\psi_{2,k}$ can be viewed as the least-squares coefficients in the continuous time regression of $J_{c,d}(s)$ against $\cos(2\pi k s)$, $\sin(2 \pi k s)$, and a constant with $s$ varying over the interval $[0,1]$. The coefficients $\varphi_{1k}$ and $\varphi_{2k}$ have a similar interpretation with $J_{c,d}$ replaced by $G_{c,d}$. Therefore, the processes $\widetilde J_{c,d}$ and $\widetilde G_{c,d}$ are the residuals from the corresponding continuous-time regressions. They are continuous-time versions of $\widetilde y_{t-1}$ and $\widetilde{\Delta y}_{t-1}$ respectively. Therefore, the results of Lemma \ref{lem:convergence_moments} continue to hold with processes $J_{c,d}$ and $G_{c,d}$ replaced by their respective residuals from continuous-time regressions.

\subsection{Linear time trend}\label{subsec:ltt}
In this section, we consider case (iii) of a linear time trend. The model in equation \eqref{eq:ARn_dc_transformed} now takes the form
\begin{equation}
y_t = \delta_n + \beta_n (t/n)  + (\phi_{1,n}+\phi_{2,n})y_{t-1} - \phi_{2,n}\Delta y_{t-1} + u_t \label{eq:ARn_ltt}
\end{equation}
where  $\delta_n \equiv \alpha_n + (\phi_{1,n} + 2\phi_{2,n})\xi/n=O(n^{-2})$, and $\beta_n\equiv \xi(1-\phi_{1,n}- \phi_{2,n})=O(n^{-2})$.\footnote{See Lemma \ref{alem:phi_taylor} in the Appendix.} 
Similarly to the previous cases, the least-squares estimators $\widehat{\phi}_{1,n}$ and $\widehat{\phi}_{2,n}$ can be obtained by estimating
\begin{equation*}
\widetilde{y}_t = (\phi_{1,n}+\phi_{2,n})\widetilde{y}_{t-1} - \phi_{2,n}\widetilde{\Delta y}_{t-1} + \widetilde{u}_{t}
\end{equation*}
where $\widetilde{y}_{t-1}$ and $\widetilde{\Delta y}_{t-1}$ are now  the residuals from the regressions of $y_{t-1}$ and $\Delta y_{t-1}$ respectively against  $t/n$ and a constant.

\begin{lemma} \label{lem:convergence_moments_ltt}
Suppose that $\{y_{t}\}$ is generated according to equation \eqref{eq:ARn_ltt}, and Assumptions \ref{a:roots} and  \ref{a:FCLT} hold. Define
\begin{align*}
\widetilde{J}_{c,d}(r) &\equiv  J_{c,d}(r) - (4-6r)\int_0^1 J_{c,d}(s)\dd s - (12r- 6 )\int_0^1 sJ_{c,d}(s) \dd s   ,\\
\widetilde{G}_{c,d}(r) &\equiv  G_{c,d}(r) - (4-6r)\int_0^1 G_{c,d}(s)\dd s - (12r- 6 ) \int_0^1 sG_{c,d}(s)\dd s .
\end{align*}
The following results hold jointly.
\begin{enumerate}[(a)]
\item $n^{-4} \sum \widetilde{y}_{t-1}^2 \Rightarrow \sigma^2\int \widetilde{J}_{c,d}^2$.
\item $n^{-2} \sum\widetilde{ \Delta y}_{t-1}^2 \Rightarrow \sigma^2\int \widetilde{G}_{c,d}^2$.
\item $n^{-3} \sum \widetilde{y}_{t-1} \widetilde{\Delta y}_{t-1} \Rightarrow \sigma^2 \int \widetilde{J}_{c,d}\widetilde{G}_{c,d} $.
\end{enumerate}
In addition, suppose that Assumption \ref{a:var} holds. The following results hold jointly with (a)-(c).
\begin{enumerate}
\item[(d)] $n^{-2} \sum \widetilde{y}_{t-1}u_t \Rightarrow \sigma^2 \int \widetilde{J}_{c,d}\dd W $.
\item[(e)] $n^{-1} \sum \widetilde{\Delta y}_{t-1}u_t \Rightarrow \sigma^2 \int \widetilde{G}_{c,d}\dd W+ \frac{1}{2}(\sigma^2 - \sigma^2_u)$.

\end{enumerate}

\end{lemma}
Lemma \ref{lem:convergence_moments_ltt} is the analogue of Lemmas \ref{lem:convergence_moments_mean} and \ref{lem:convergence_moments_cycles} for the case of the linear time trend. The processes $\widetilde J_{c,d}(r)$ and $\widetilde G_{c,d}(r)$ can be interpreted similarly as the residuals from the continuous-time regressions of $ J_{c,d}(r)$ and $ G_{c,d}(r)$, respectively, against a constant and $r$ varying over the interval $[0,1]$.

\subsection{Asymptotic distributions of the estimators and test statistics}
The results of Lemmas \ref{lem:convergence_moments_mean}--\ref{lem:convergence_moments_ltt} can now be used to describe the asymptotic distributions of the least-squares estimators of the autoregressive coefficients and the corresponding Wald statistics for the models with constant mean, deterministic cycles, and a linear time trend, respectively. 

Under the same assumptions as those in Proposition \ref{prop:LSphi}, however with the model in equation \eqref{eq:DGP_n} replaced by that in either \eqref{eq:ARn_mean}, \eqref{eq:ARn_cycles}, or \eqref{eq:ARn_ltt}, the asymptotic distribution of the least-squares estimators of the autoregressive coefficients now satisfies
\begin{eqnarray}
\lefteqn{ n^2(\widehat{\phi}_{1,n} +\widehat{\phi}_{2,n}- \phi_{1,n}-\phi_{2,n}) \Rightarrow } \notag\\
&& \frac{\int\widetilde G_{c,d}^2 \cdot \int \widetilde J_{c,d} \dd W -  \bigg(\int \widetilde G_{c,d} \dd W + \frac{1}{2}(1 - \sigma^2_u/\sigma^2)\bigg)\cdot \int \widetilde J_{c,d} \widetilde G_{c,d} }{\int \widetilde J^2_{c,d}  \cdot \int \widetilde G^2_{c,d}  - (\int\widetilde J_{c,d}\widetilde G_{c,d} )^2}, \notag \\
 \lefteqn{n \begin{pmatrix}
\widehat{\phi}_{1,n} - \phi_{1,n}\\
\widehat{\phi}_{2,n} - \phi_{2,n}
\end{pmatrix} \Rightarrow } \notag\\
&&
\begin{pmatrix}
-1\\
1
\end{pmatrix} \times 
 \frac{   \int\widetilde J_{c,d}\widetilde G_{c,d} \cdot \int\widetilde J_{c,d}\dd W-\bigg(\int\widetilde G_{c,d}\dd W  + \frac{1}{2}(1 - \sigma^2_u/\sigma^2)\bigg)\cdot \int\widetilde J^2_{c,d}}{\int\widetilde J^2_{c,d} \cdot \int\widetilde G_{c,d}^2  - \big(\int \widetilde J_{c,d} \widetilde G_{c,d} \big)^2}, \label{eq:limit_phi_tilde}
\end{eqnarray}
where the convergence holds jointly with the results of either Lemma \ref{lem:convergence_moments_mean}, \ref{lem:convergence_moments_cycles}, or \ref{lem:convergence_moments_ltt} respectively with the correspondingly defined residual processes 
$\widetilde J^2_{c,d}$  and $\widetilde G^2_{c,d}$.

For all three specifications in Sections \ref{subsec:mean}--\ref{subsec:ltt}, the Wald  statistic for testing $H_0:\phi_1=\phi_{1,0}, \phi_2=\phi_{2,0}$ against $H_1:\phi_1\ne\phi_{1,0}\;\text{or}\; \phi_2\ne\phi_{2,0}$ takes the same form as in equation \eqref{eq:Wald}. However,  $\widehat V_n$ is now given by 
\[
\widehat{V}_n = \widehat{\sigma}_n^2
\begin{pmatrix}
\sum \widetilde{y}_{t-1}^2 & - \sum\widetilde{y}_{t-1}\widetilde{ \Delta y}_{t-1} \\
- \sum\widetilde{y}_{t-1}\widetilde{\Delta y}_{t-1} & \sum (\widetilde{\Delta y}_{t-1})^2
\end{pmatrix}^{-1},
\]
with $\widetilde{y}_{t-1}$ and $\widetilde{\Delta y}_{t-1}$ defined respectively for each specification.
Provided that the assumptions of Proposition \ref{prop:wald} hold with the model in \eqref{eq:DGP_n} replaced by that in either \eqref{eq:ARn_mean}, \eqref{eq:ARn_cycles}, or \eqref{eq:ARn_ltt}, the asymptotic null distribution of the Wald statistic is given by
\begin{align}
W_n(\phi_{1,n},\phi_{2,n}) \Rightarrow
\frac{\bigintsss \bigg\{\widetilde J_{c,d} \cdot \bigg(\int\widetilde G_{c,d} \dd W + \frac{1}{2}(1 - \sigma^2_u/\sigma^2)\bigg)-\widetilde G_{c,d}  \cdot \int\widetilde J_{c,d}\dd W\bigg\}^2 }{\int\widetilde J^2_{c,d}\cdot \int\widetilde G^2_{c,d} - \big(\int\widetilde J_{c,d} \widetilde G_{c,d} \big)^2} \label{eq:wald_distribution}
\end{align}
with the correspondingly defined residual processes 
$\widetilde J^2_{c,d}$  and $\widetilde G^2_{c,d}$.

As in the base case with no deterministic components, the asymptotic null distributions of the Wald statistics are nonstandard and depend on the unknown parameters $c$ and $d$. Differences between the quantiles of these asymptotic distributions and the critical values of the $\chi^2_2$ distribution are discussed in Appendix \ref{sec:distortions}. Compared to the base case, the inclusion of deterministic components may result in more substantial deviations from the $\chi^2_2$ critical values.

%% file: Inference.tex

\section{Inference for cyclicality}\label{sec:inference}

In this section, we propose a procedure for inference on the cycle length in terms of the angular frequency-based measure $\tau_{\theta}$  and the spectrum-based measure $\tau_{\omega}$ that were introduced in Section \ref{section:model}. Recall that the two measures can be deduced from the autoregressive coefficients $\phi_{1n}$ and $\phi_{2,n}$ through the relationships in \eqref{eq:phi_1n}--\eqref{eq:tau_omega}. Therefore, we first construct confidence sets for the autoregressive parameters by collecting values $(\phi_1,\phi_2)$ consistent with cyclical behavior and not rejected by data. In the second step, we use projection arguments to construct confidence intervals for $\tau_{\theta}$ and $\tau_{\omega}$. By multiplying the values of $\tau_{\theta}$ and $\tau_{\omega}$ in the confidence intervals by $n$, the length of the cycle can also be expressed in time units instead of fractions of sample size.

The proposed confidence sets have the following property: If the true DGP is indeed cyclical, the coverage probability is at least $1-\alpha$ asymptotically whether the roots of the autoregressive equation are close to or far from one. However, if the true DGP is inconsistent with the cyclical behavior, we expect the confidence sets to be empty in large samples. Therefore, the proposed procedure can be used to detect cyclical specifications consistent with the data or to rule out cyclical behavior. However, our procedure is not designed for inference on acyclical specifications.

When the roots of the autoregressive polynomial are local to unity as in Assumption \ref{a:roots}, the least-squares estimators of the autoregressive coefficients are consistent regardless of whether $\{u_t\}$ is serially correlated or not. This is established in Proposition \ref{prop:LSphi} for the base case and in \eqref{eq:limit_phi_tilde} for the cases with deterministic components. The serial correlation in $\{u_t\}$ and the resulting correlation between $(y_{t-1}, \Delta y_{t-1})$ and $u_t$ is reflected by the noncentrality term $0.5(1-\sigma^2_u/\sigma^2)$ in the asymptotic distributions. The non-centrality proliferates from the estimators into the asymptotic null distribution of the Wald statistic. This is standard for the unit root literature and continues to hold in our framework. 

However, when the roots of the autoregressive polynomial are sufficiently far from unity, that is, under I (0) specifications, the least squares estimators of the autoregressive coefficients are no longer consistent if $\{u_t\}$ is serially correlated. Although the process is I(0), due to the inconsistency of the least-squares estimators, the null asymptotic distribution of the Wald statistic is no longer a central $\chi^2$. Therefore, to design an inferential procedure that remains valid regardless of the magnitude of the roots, we need to be able to accommodate a potential serial correlation in $\{ u_t\}$. 

We proceed as follows. First, in Section \ref{subsec:uncorr} we discuss how to construct confidence intervals for $\tau_\theta$ and $\tau_\omega$ when $\{ u_t\}$ are serially uncorrelated. Then, in Section \ref{subsec:corr} we extend the procedure to a serially correlated innovation process $\{ u_t\}$ by assuming that it satisfies an AR($p$) formulation with real roots bounded away from one. We use the BIC selection procedure to choose the appropriate number of lags $p$ and the specification for the deterministic part $D_t$.

\subsection{Serially uncorrelated $\{u_t\}$}\label{subsec:uncorr}
Suppose that $\{u_t\}$ is serially uncorrelated and, therefore, the least-squares estimators of the autoregressive coefficients $\phi_{1,n}$ and $\phi_{2,n}$ are consistent whether the roots are close to unity or far from it. Recall that the expression on the right-hand side of \eqref{eq:wald_distribution} with $1-\sigma^2_u/\sigma^2=0$ approximates well the asymptotic distribution of the Wald statistic for any configuration of the localization parameters $c$ and $d$. Moreover, recall that given the sample size $n$, there is a one-to-one relationship between $(\phi_{1,n}, \phi_{2,n})$ and the localization parameters $(c,d)$, and let
\begin{alignat*}{3}
\phi_{1,n}& =\Phi_{1,n}(c,d)\;\text{and}\; \phi_{2,n} &&=\Phi_{2,n}(c,d), 
\end{alignat*}
where the functions $\Phi_{1,n}(c,d)$ and  $\Phi_{2,n}(c,d)$ are defined according to \eqref{eq:phi_1n} and \eqref{eq:phi_2n} respectively. Because the relationship is one-to-one for any given $n$, confidence sets for $(\phi_1,\phi_2)$ can be equivalently represented as confidence sets in terms of $(c,d)$.

After running the regressions of $y_t$ against $y_{t-1}$ and $y_{t-2}$ with different specifications of the deterministic part $D_t$, the BIC can be used to consistently choose the appropriate specification between the constant mean, the deterministic cycle, or the linear time trend. After selecting $D_t$, consider the corresponding Wald statistic $W_n(\Phi_{1,n}(c,d),\Phi_{2,n}(c,d))$. Let $\mathcal{W}_{1-\alpha}(c,d)$ denote the $1-\alpha$ quantile of the asymptotic distribution in \eqref{eq:wald_distribution} with $1-\sigma^2_u/\sigma^2=0$. I.e., $\mathcal{W}_{1-\alpha}(c,d)$ is the $1-\alpha$ quantile of the distribution of
\begin{equation}\label{eq:Wald_distr}
\mathcal W(c,d) \sim
\frac{\bigintsss \bigg\{\widetilde J_{c,d} \cdot \int\widetilde G_{c,d} \dd W -\widetilde G_{c,d}  \cdot \int\widetilde J_{c,d}\dd W\bigg\}^2 }{\int\widetilde J^2_{c,d}\cdot \int\widetilde G^2_{c,d} - \big(\int\widetilde J_{c,d} \widetilde G_{c,d} \big)^2}, 
\end{equation}
where the definitions of $\widetilde J_{c,d}$ and $\widetilde G_{c,d}$ correspond to the selected specification for $D_t$.
The confidence set for $(c,d)$ can now be constructed by test inversion as
\begin{equation*}
CS_{n,1-\alpha}\equiv \bigg\{(c,d): W_n\big(\Phi_{1,n}(c,d),\Phi_{2,n}(c,d)\big) \leq \mathcal{W}_{1-\alpha}(c,d)\bigg\}.
\end{equation*}

The confidence set $CS_{n,1-\alpha}$ is bounded as $\mathcal W_{1-\alpha}(c,d)\to\chi^2_{2,1-\alpha}$ when $c\to-\infty$ or $d\to\infty$. In practice, the confidence set can be approximated by choosing a dense two-dimensional grid of values $c$ and $d$. We use a grid with $\underline{c}\leq c\leq 0$, and $2\pi< d < n\pi$, where the lower bound $\underline{c}$ is chosen by the econometrician and the lower bound of $2\pi$  is imposed to rule out cycles longer than the sample size when measured by $\tau_\theta$.\footnote{In our empirical application in Section \ref{sec:empirical}, the largest considered value of $d$ is 678, and $\underline{c}=-678$. Note also that the grid only needs to cover the stationary region corresponding to complex roots \citep[see e.g.][pages 187--189]{sargent1987macroeconomic}.}

The construction of $CS_{n,1-\alpha}$ is similar to the grid bootstrap procedure of \cite{hansen1999grid}, however, we use the asymptotic critical values instead of their bootstrap approximation. Note that the critical values must be adjusted for every point $(c,d)$ considered. The validity of $CS_{n,1-\alpha}$ is due to the following facts. First, $(c,d)$ is included in the confidence set only if the null hypothesis $H_0:\phi_{1,n}=\Phi_{1,n}(c,d), \phi_{2,n}=\Phi_{2,n}(c,d)$ cannot be rejected by the Wald test with the critical value $\mathcal{W}_{1-\alpha}(c,d)$. Second, the critical values are computed using the same values $(c,d)$ as those specified in $H_0$. Third, the distribution in \eqref{eq:Wald_distr} nests the $\chi^2_2$ distribution, which arises under the fixed $(\phi_1,\phi_2)$  asymptotics, as a limiting case. Note that having the correct size under both drifting and fixed parameter specifications is required for uniform validity \citep{andrews2020generic}.

We construct confidence intervals for $\tau_\theta$ and $\tau_\omega$ from  $CS_{n,1-\alpha}$ by projection:
\begin{alignat}{2}
CI_{n,1-\alpha}^{\tau_\theta} &\equiv \bigg[ \inf_{d:(c,d)\in  CS_{n,1-\alpha}}\frac{2\pi}{d},  \sup_{d:(c,d)\in  CS_{n,1-\alpha}}\frac{2\pi}{d}\bigg],\label{eq:CStau_theta}\\
CI_{n,1-\alpha}^{\tau_\omega} &\equiv \bigg[ \inf_{(c,d)\in  CS_{n,1-\alpha}}\frac{2\pi}{\sqrt{d^2-c^2}},  \sup_{(c,d)\in  CS_{n,1-\alpha}}\frac{2\pi}{\sqrt{d^2-c^2}}\bigg].\label{eq:CStau_omega}
\end{alignat}

The confidence interval for $\tau_\theta$ is bounded as long as the grid of $d$ values used to construct $CS_{n,1-\alpha}$ excludes zero. On the other hand, the confidence interval for $\tau_\omega$ can be unbounded if pairs $(c,d)$ with $c=d$ are included in $CS_{n,1-\alpha}$.


\subsection{Serially correlated $\{u_t\}$}\label{subsec:corr}

In this section we assume that the innovations process $\{ u_t\}$ is generated as AR($p$):
\begin{equation}\label{eq:AR_for_u}
(1-\rho_1 L-\ldots -\rho_p L^p)u_t=\varepsilon_t,
\end{equation}
where $\{\varepsilon_t\}$ are iid $(0,\sigma^2_\varepsilon)$, and the roots of the polynomial $1-\rho_1 L-\ldots -\rho_p L^p$ are real and bounded away from unity.
In this case, $\{y_t\}$ is AR($p+2$), and by running the regressions of $y_t$ against different specifications of the deterministic part $D_t$ and $y_{t-1},y_{t-1},\ldots,y_{t-(m+2)}$ for some $m>p$, one can again use the BIC to consistently estimate the specification for $D_t$ and the number of lags $p$ in \eqref{eq:AR_for_u}.

Let $\widetilde y_t$ denote the residuals from the projection of $y_t$ against the components of $D_t$. Under $H_0:\phi_{1,n}=\phi_{1,0}, \phi_{2,n}=\phi_{2,0}$, the values $\phi_{1,n}$ and $\phi_{2,n}$ are known and can be calculated from the values of $c$ and $d$. Let 
\[
\widetilde u_{t,0}\equiv \big(1-\phi_{1,0}L - \phi_{2,0} L^2\big) \widetilde y_t.
\]
Using the null-restricted residuals $\widetilde u_{t,0}$, one can estimate the autoregressive coefficients $\rho_1,\ldots,\rho_p$. Let $\widehat\rho_{1,0},\ldots,\widehat\rho_{p,0}$ denote their least-squares estimators. Note that under $H_0$, these estimators are consistent. We can now remove the autoregressive part in $u_t$:
\[
\widehat x_{t,0} \equiv(1-\widehat\rho_{1,0} L-\ldots -\widehat\rho_{p,0} L^p) \widetilde y_t.
\] 
Thus, to construct the process $\widehat x_{t,0}$, we have filtered the deterministic part $D_t$ and the serial correlation in $\{u_t\}$. Note that under the null, the population counterpart of $\widehat x_{t,0}$ satisfies:
\begin{align*}
\widetilde x_{t} &\equiv  (1-\rho_{1} L-\ldots -\rho_{p} L^p) \widetilde y_t
=\frac{\widetilde \varepsilon_t}{1-\phi_{1,n}L - \phi_{2,n} L^2},
\end{align*}
where $\widetilde\varepsilon_t$ is the residual from the projection of $\varepsilon_t$ against the components of $D_t$.

 Now we can use $\{\widehat x_{t,0}\}$ for inference on the cyclical properties of $\{y_t\}$, however,  additional adjustments are required to account for the estimation of $\rho_1,\ldots,\rho_p$. The main purpose of the adjustments discussed below is to ensure that the modified Wald statistic has the correct asymptotic null distributions both under the long-cycle asymptotics proposed in the paper and under the standard asymptotics with $\phi_1$ and $\phi_2$ fixed in the stationary range.\footnote{Recall that having correct size under both drifting and fixed parameters specifications is required for the uniform validity \citep{andrews2020generic}.
} 
  
Let $\widehat{\phi}_{1,n}$ and $\widehat{\phi}_{1,n}$ now denote the least-squares estimators of $\phi_1$ and $\phi_2$ respectively from the regression of $\widehat x_{t,0}$ against $\widehat x_{t-1,0}$ and $\widehat x_{t-2,0}$:
\[
\widehat x_{t,0}=\widehat{\phi}_{1,n}\widehat x_{t-1,0}+\widehat{\phi}_{2,n}\widehat x_{t-2,0}+\widehat\varepsilon_{t,0}, 
\]
where $\widehat\varepsilon_{t,0}$ denotes the least-squares residuals.
 The modified Wald statistic takes the form
 \[
 W_{n,p}(\phi_{1,0},\phi_{2,0}) \equiv 
 \frac{1}{\widehat \sigma^2_{\varepsilon,n}}
 \begin{pmatrix}
\widehat{\phi}_{1,n}  - \phi_{1,0}\\
\widehat{\phi}_{2,n} - \phi_{2,0}
\end{pmatrix}^\top    
 M_n{\Sigma}^{-1}_n M_n
\begin{pmatrix}
\widehat{\phi}_{1,n} -\phi_{1,0}\\
\widehat{\phi}_{2,n} - \phi_{2,0}
\end{pmatrix},
 \]
 where $\widehat \sigma^2_{\varepsilon,n} \equiv n^{-1}\sum \widehat \varepsilon_{t,0}^2$, and the matrix $M_n$ is given by
 \[M_n \equiv 
 \begin{pmatrix}
 \sum \widehat x^2_{t-1,0}  & \sum \widehat x_{t-1,0}\widehat x_{t-2,0}\\
\sum\widehat x_{t-1,0}\widehat x_{t-2,0}& \sum \widehat x^2_{t-2,0}
\end{pmatrix}.
 \]
To construct $\Sigma_n$, we first define $\dot x_{t,0}$ and $\ddot x_{t,0}$ as the residual from the least-squares regression of $\widehat x_{t,0}$ and $\widehat x_{t-1,0}$ respectively against $\widetilde u_{t,0},\ldots,\widetilde u_{t-p+1,0}$:
\begin{align}
\widehat x_{t,0} &=\dot \zeta_{1,n} \widetilde u_{t,0}+\ldots+\dot \zeta_{p,n} \widetilde u_{t-p+1,0}+\dot x_{t,0}, \label{eq:dot}\\
\widehat x_{t-1,0} &=\ddot \zeta_{1,n} \widetilde u_{t,0}+\ldots+\ddot\zeta_{p,n} \widetilde u_{t-p+1,0}+\ddot x_{t-1,0}, \notag
\end{align}
where $\dot\zeta_{1,n},\ldots,\dot\zeta_{p,n}$ and $\ddot\zeta_{1,n},\ldots,\ddot\zeta_{p,n}$ are the OLS estimators. The matrix $\Sigma_n$ is given by
 \[\Sigma_n \equiv 
 \begin{pmatrix}
 \sum \dot x^2_{t-1,0}  & \sum \dot x_{t-1,0}\ddot x_{t-2,0}\\
\sum\dot x_{t-1,0}\ddot x_{t-2,0}& \sum \ddot x^2_{t-2,0}
\end{pmatrix}.
 \]
 The next proposition shows that under the conventional stationary asymptotics, the asymptotic null distribution of the Wald statistic is the usual $\chi^2_2$ distribution.
 \begin{proposition}\label{prop:Wald_AR_usual}
  Suppose that $\{y_t\}$ is generated according to $(1-\phi_1 L-\phi_2 L^2) y_t =u_t $ with the coefficients $\phi_1$ and $\phi_2$ fixed in the stationary range, and $\{u_t\}$ satisfying \eqref{eq:AR_for_u} with the coefficients $\rho_1,\ldots,\rho_p$ in the stationary range and $\varepsilon_t \sim \text{iid}(0,\sigma^2_\varepsilon)$. Then,
 \[
 W_{n,p}(\phi_{1},\phi_{2}) \Rightarrow \chi^2_2.
 \]
 \end{proposition} 
 
 In the case of a long-cycle specification, the null asymptotic distribution of the modified Wald statistic is the same as in \eqref{eq:Wald_distr}.
 \begin{proposition}\label{prop:Wald_with_AR}
 Suppose that $\{y_t\}$ is generated according to \eqref{eq:yc+d}, where $\{u_t\}$ satisfies \eqref{eq:AR_for_u} with the coefficients $\rho_1,\ldots,\rho_p$ in the stationary range and $\varepsilon_t \sim \text{i.i.d.}(0,\sigma^2_\varepsilon)$. Suppose further that Assumption \ref{a:roots} holds. Then,
\[
W_{n,p}(\phi_{1,n},\phi_{2,n}) \Rightarrow \mathcal W(c,d),
\]
 where $\mathcal W(c,d)$ is defined in \eqref{eq:Wald_distr} with $\widetilde J_{c,d}$ and $\widetilde G_{c,d}$ defined according to the specification of $D_t$.
 \end{proposition}
 
%

Using the results of Propositions \ref{prop:Wald_AR_usual} and \ref{prop:Wald_with_AR}, one can now construct confidence sets for $(c,d)$ using the modified Wald statistic as
\begin{equation*}
CS_{n,p,1-\alpha}\equiv \bigg\{(c,d): W_{n,p}\big(\Phi_{1,n}(c,d),\Phi_{2,n}(c,d)\big) \leq \mathcal{W}_{1-\alpha}(c,d)\bigg\}.
\end{equation*}
Similarly to the construction in \eqref{eq:CStau_theta} and \eqref{eq:CStau_omega}, the confidence set $CS_{n,p,1-\alpha}$ can be projected to construct confidence intervals for $\tau_\theta$ and $\tau_\omega$.

%% file: Empirical_application.tex

\section{Cyclical properties of macroeconomic and financial variables}\label{sec:empirical}

Recurrent boom-and-bust cycles are a salient feature of economic and financial history. A long-standing interest in understanding these ups and downs in the macro-financial aggregates has led to a vast body of literature on business cycles and a resurgence of research on financial cycles post the financial crisis-induced Great Recession of 2008. Among these strands of work is the empirical characterization of business and financial cycles. The traditional approach to such a characterization is to identify turning points or peaks and troughs in the time series using the dating algorithms of \citet{bry1971front} and \citet{harding2002dissecting}. Based on the turning-point analysis, \citet{DrehmannBorioTsatsaronis2012} highlight the importance of medium-term cycles that last 18 years for credit, 11 years for GDP and 9 years for equity prices. These findings are in line with studies using frequency-based bandpass filters \citep[see][]{Aikman_etal2015EJ,comin2006medium}.


The cyclical properties of the data have also been formally examined in the literature using a variety of methods, including direct and indirect spectrum estimation \citep[e.g.][]{AHearnWoitek2001JME,Strohsal_etal2019JBF} and structural time-series modeling \citep[e.g.][]{Harvey1985,Runstler2018business}. However, they rely on the conventional asymptotic approximations that may produce misleading results with long-cycle data, as we argue in this paper. For example, we show in Appendix \ref{sec:periodogram} that the periodogram-based estimator is asymptotically biased in the case of long cycles.

In this section, we apply our inference procedure to the quarterly series of a set of macroeconomic and financial variables for the U.S. All data are publicly available from FRED, Federal Reserve Bank of St. Louis. A detailed description of the data is summarized in Table \ref{tbl:data description} in Appendix \ref{sec:data_description}. All the series are measured in natural logs except for the credit-to-GDP ratio (for the private non-financial sector), which is in percentage points, and the interest rate spread between Moody's seasoned BAA corporate bond yield and the 10-year treasury constant maturity, which is expressed in levels. For each series, we take the longest and most updated sample ending in 2020. Depending on the series, our samples span periods ranging from 34 to 73 years. 

We use the empirical models in \eqref{eq:yc+d}. Let $y_t$ denote the observed data series such that
\begin{align*}
y_t &= y^c_t + D_t,\\
y^c_t &= \phi_{1,n} y^c_{t-1} + \phi_{2,n} y^c_{t-2} + u_t,
\end{align*}
where $y^c_t$ is the latent cyclical part, the innovations $\{u_t\}$ are potentially serially correlated according to an AR($p$) specification with unknown $p$, and $D_t$ may contain linear deterministic trends and deterministic cycles as discussed in Section \ref{sec:extensions}.  In all specifications, the intercept (constant) is included by default. The raw data for the credit-to-GDP ratio are not seasonally adjusted and, therefore, the specification for $D_t$ also allows for seasonal dummies.

We use the BIC to select the appropriate specification for $D_t$ (i.e. whether to include a linear time trend, deterministic cycles, or seasonal dummies). We also rely on the BIC to select the lag order $p$ 
in the AR($p$)  specification for $u_t$. To this end, suppose that $p\leq M$ for some known positive integer $M$. As discussed in Section \ref{sec:BIC} of the Supplement, one can choose $p\in\{0,1,\ldots,M\}$ that minimizes the BIC for the regression of $y_t$ against $y_{t-1}$, $\Delta y_{t-1}$, the deterministic components, and the second-order differences $\Delta^2 y_{t-1},\ldots,\Delta^2 y_{t-p}$, where $\Delta^2 y_t \equiv \Delta y_t - \Delta y_{t-1}$. 

In our empirical application, most of the time series do not exhibit autocorrelation in $\{u_t\}$, that is $p=0$, except for the credit-to-GDP ratio, as indicated by the BIC. Moreover, the credit-to-GDP ratio is the only series in which we have included seasonal dummies.

\begin{table}
\caption{Length of cycle in quarters for macroeconomic and financial variables: the 95\% confidence intervals for the two measures of the cycle length ($\tau_\omega$ and $\tau_\theta$), two localization parameters ($c$ and $d$), and BIC selected specifications}

\begin{adjustbox}{width=\textwidth}

\begin{threeparttable}[b]

\begingroup
\setlength{\tabcolsep}{15pt} 
\renewcommand{\arraystretch}{1.0} 
\begin{tabular}{@{} ccccccccc @{}}
\toprule
\addlinespace
& \multirow{2}{*}{$n\tau_{\theta}$}  & \multirow{2}{*}{$n\tau_{\omega}$} & \multirow{2}{*}{$c$}& \multirow{2}{*}{$d$}& \multirow{2}{*}{$n$} & Linear  & Deterministic  & Autocorr.\\
& & & & & & time trend &cycles& $u_t$\\
\addlinespace
\midrule
\addlinespace
\addlinespace
\multicolumn{9}{@{}l@{}}{\textbf{Macroeconomic variables}}\\
\addlinespace
\addlinespace
Real GDP per capita &\makecell{--- \\ (23, 264)} &\makecell{---\\ (25, 512)} & (-108, 0) & (7, 81) & 294 & Yes & No & No\\[4ex]
Unemployment rate & \makecell{52 \\ (22, 260)}&\makecell{---\\ (27, 185)} &(-104, -43)& (7, 84)& 290 & No & No& No\\[4ex]
Hours per capita & \makecell{22\\ (18, 64)} &\makecell{25\\ (18, 198)} &(-73, -10)& (29, 105)& 294 & Yes& $k=1$& No\\[4ex]

\addlinespace
\multicolumn{9}{@{}l@{}}{\textbf{Financial variables}} \\
\addlinespace
\addlinespace
\makecell{VXO\\ S\&P 100 volatility index} &\makecell{---\\ $\emptyset$}  & \makecell{---\\$\emptyset$ }&$\emptyset$&$\emptyset$& 139& No& $k=3$ & No \\[4ex]
\makecell{Credit risk premium \\ BAA to 10Y }  & \makecell{---\\ $\emptyset$ }  & \makecell{---\\ $\emptyset$ }&$\emptyset$&$\emptyset$& 269& Yes & No& No \\[4ex]
Equity price index &  \makecell{---\\ $\emptyset$}   & \makecell{---\\ $\emptyset$}  &$\emptyset$&$\emptyset$& 197 &Yes & No & No\\[4ex]

\makecell{ Private non-financial  sector \\ credit \% GDP  }  & \makecell{---\\ (48, 245) }  & \makecell{---\\ (55, 308)}&(-41, -13)&(7, 36)& 273&Yes & No & AR(1)\\[4ex]
Home price index &\makecell{63\\ (42, 120)} & \makecell{77\\(43, 234)} &(-16, -1)&(7, 20)& 134 & Yes & No & No\\
\addlinespace
\addlinespace
\bottomrule
\end{tabular}

\begin{tablenotes}
\item[1]All data series are sampled at quarterly frequencies with the sample size of each series given by $n$. 
\item[2] Columns 1 and 2 indicate, respectively, the duration of the cycles measured by the angular frequency $n\tau_{\theta}$ and the spectrum-maximizing frequency $n\tau_{\omega}$.
\item[3] In columns 1 and 2, the numbers on the top indicate the point estimates of the cycle length. A dashed line ``---" is used when the point estimate corresponds to an acyclical process and when the point estimate is unavailable in the case of autocorrelation.  Enclosed in parentheses are the minimum and maximum cycle lengths implied by the $95\%$ confidence intervals of $\tau_{\theta}$ and $\tau_{\omega}$. When the interval is empty, it is indicated by "$\emptyset$". All numbers are in quarters.
\item[4]The intercept is included in all specifications.
\item[5]Credit to the private non-financial sector (\% GDP) are seasonally adjusted by including seasonal dummies in the regression. 
\end{tablenotes}
\endgroup
\end{threeparttable} \label{tbl:empirical}

\end{adjustbox}
\end{table}

Table \ref{tbl:empirical} presents our results. 
Note that the last three columns of the table describe the specifications selected by the BIC for $D_t$ and the order of autocorrelation for $\{u_t\}$. For example, the hours per capita contains a linear time trend and deterministic cycles of cosine and sine waves with $k=1$, which corresponds to the periodicity of
$n/k = n$.
According to the BIC, the errors $\{u_t\}$ are serially uncorrelated. 

Columns 1 and 2 of the table report, respectively, the angular frequency-based measure $n\tau_{\theta}$ and the spectrum-maximizing frequency-based measure $n\tau_{\omega}$ for the cycle length. The point estimates for $n\tau_{\theta}$ and $n\tau_{\omega}$ are constructed by backing out the corresponding values of $c,d$ from the OLS estimates of $\phi_{1,n}$ and $\phi_{2,n}$, which can then be used to compute $\tau_{\theta}$ and $\tau_{\omega}$. The point estimates are indicated as ``---" when the autoregressive coefficient estimates of $\phi_{1,n}$ and $\phi_{2,n}$ correspond to acyclical processes,\footnote{While the OLS estimates of the autoregressive coefficients may correspond to an acyclical process (no complex roots or no spectrum peak) our confidence sets for  $\phi_{1,n}$ and $\phi_{2,n}$ may nevertheless include cyclical values.} or when they are not available as in the case of autocorrelation. The minimum and maximum cycle lengths implied by the 95\% confidence intervals of $n\tau_{\theta}$ and $n\tau_{\omega}$ are given in parentheses. For completeness, the table also reports the 95\% confidence intervals for $c$ and $d$.

The two alternative measures of cycle length generally produce similar lower bound estimates. Based on the 95\% confidence intervals, we are unable to reject the null that macroeconomic variables, such as the real GDP per capita, the unemployment rate, and the hours per capita, contain stochastic cycles with periodicity of at least 5-6 years. Partly due to the projection-based construction of the $n\tau_{\theta}$ and $n\tau_{\omega}$ confidence intervals, the implied range of the cycle length is typically wide. The upper bound confidence estimates usually are large and differ considerably between the two measures. Nevertheless, our results point to the presence of cyclicality among macroeconomic variables, conforming to the view of endogenous business cycles \citep{beaudry2020aer}. 
On the financial side, we find that credit to the private non-financial sector as a percent of GDP and the home prices exhibit cycles of at least 10 years in duration, twice as long as the minimum detected cycle length in the macroeconomic variables. 

The most striking finding of this section is that for asset market variables (the volatility index, credit risk premium, and equity prices), our procedure returns empty confidence sets. 
This suggests that the underlying mechanism for asset market fluctuations is different from that of macro variables and financial variables such as credit and home prices. Our results are in favour of the dichotomy between the asset market and the real economy. Moreover, the results do not support the view that recessions are driven by risk perception, risk premiums, and risk-bearing capacity suggested in the macro-finance literature \citep[see][]{cochrane2017macro}. Note that the S\&P 100 Volatility Index, a measure of market uncertainty, has a deterministic cycle of approximately 46 quarters in length according to the BIC. However, it is different in nature from the stochastic cycles detected in the other variables.

To better visualize the cyclical dynamics consistent with the data, for each series ${y_t}$, we plot in Figure \ref{fig:irf} the impulse responses to a one-standard-deviation shock to the innovation $u_{0}$ of all cyclical specifications in the 95\% confidence sets $CS_{n,1-\alpha}$.\footnote{For the credit to the private non-financial sector, the standard deviation of the innovation is computed assuming no serial correlation.} The dynamics shown in the figure resonate with the results in Table \ref{tbl:empirical}. Financial variables such as the credit to the private non-financial sector and the home prices exhibit much longer cycles than the business cycle variables. The duration from peaks to troughs is at least 25-30 quarters in financial cycles and at least 15 quarters in business cycles. Furthermore, financial cycles are also more pronounced: for a one-standard deviation shock, the initial amplitude of the cyclical response is approximately 3 to 7 times the standard deviation for the credit and the home prices, and about 1.5 to 3 times for the unemployment rate and the hours per capita. 

For real GDP per capita, the impulse responses are split into two parts. On the left, the axis corresponds to the set of impulse responses similar to those observed in the unemployment rate and the hours per capita. On the right, the axis maps to the set of cyclical impulse responses with large amplitudes and high persistences. Note that the scale of the axis on the right has increased by 10-fold. Although the possibility of having much longer and highly persistent stochastic cycles cannot be rejected, the real GDP per capita shares a similar dynamic to the unemployment rate and the hours per capita.\footnote{Note also that the hours per capita is much more persistent than the unemployment rate and the real GDP per capita. }

\begin{figure} 
 \bigskip
\begin{center}
\begin{subfigure}{0.35\textwidth}
\includegraphics[width=\textwidth]{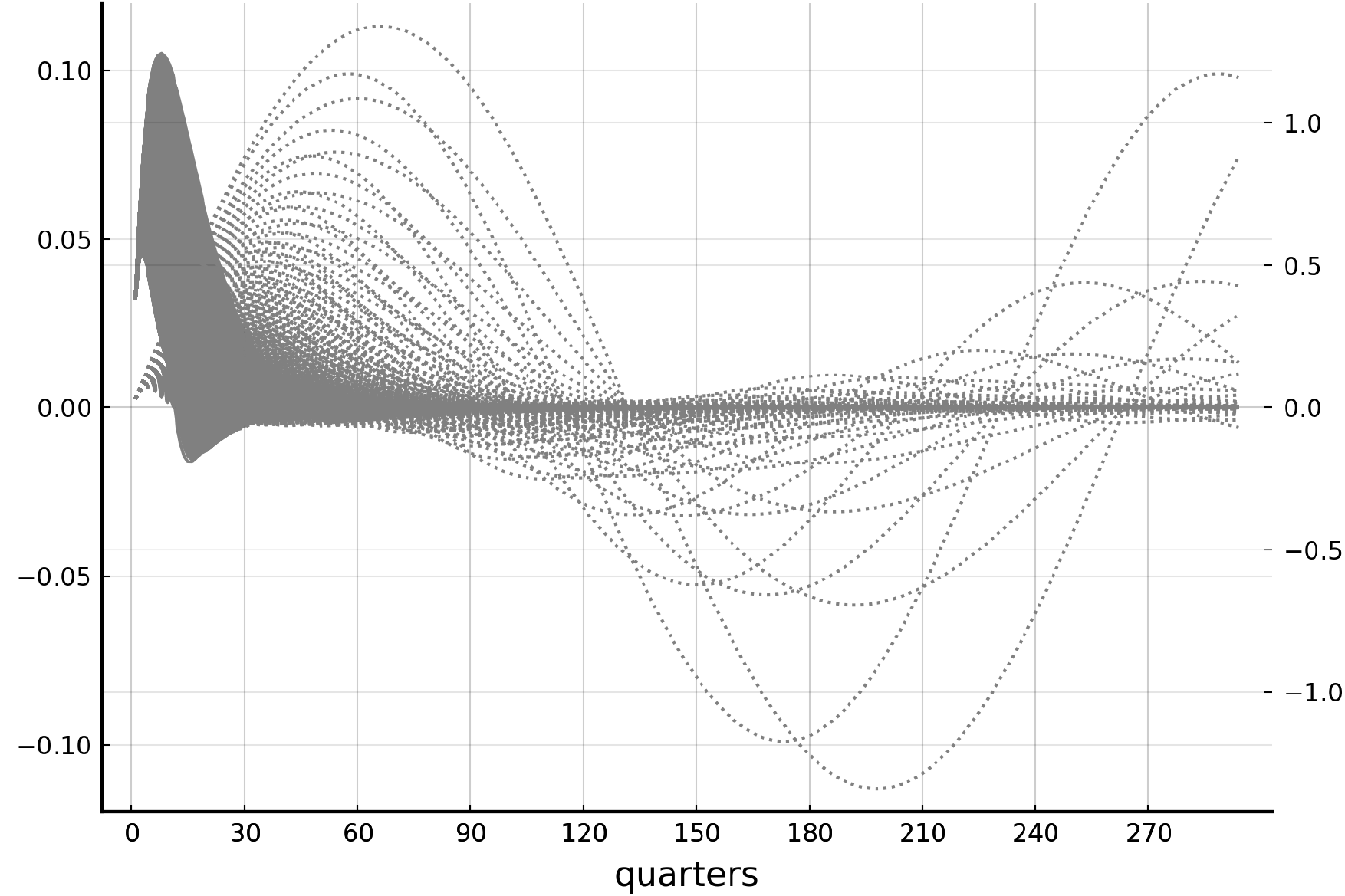}
\caption{Real GDP per capita (the dotted lines are on the right axis\textsuperscript{1}) }
\end{subfigure}
\hspace{1ex}
\begin{subfigure}{0.35\textwidth}
\includegraphics[width=\textwidth]{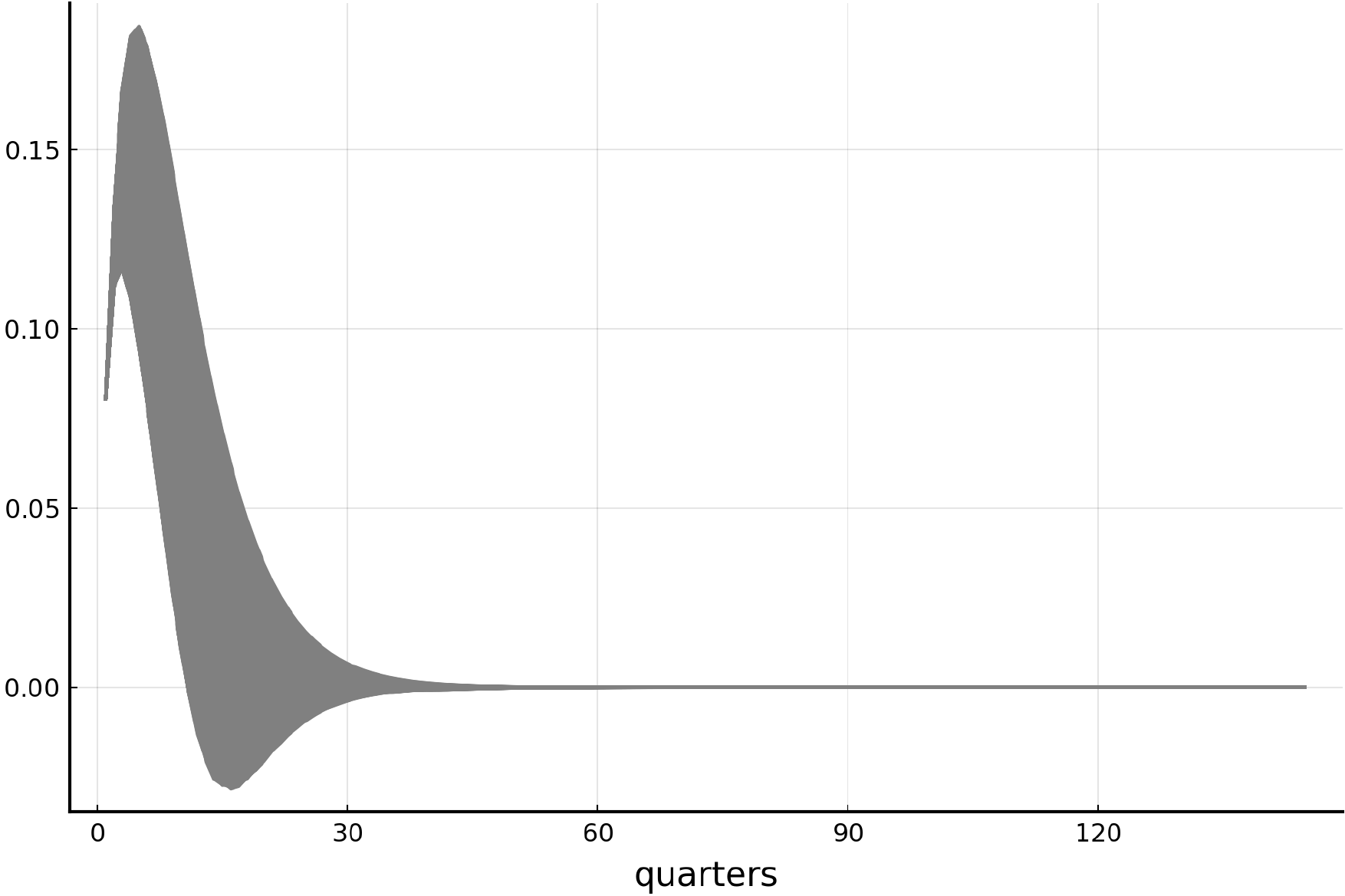}
\caption{Unemployment rate\\ \hspace{1pt}}
\end{subfigure}

\medskip 

\begin{subfigure}{0.35\textwidth}
\includegraphics[width=\textwidth]{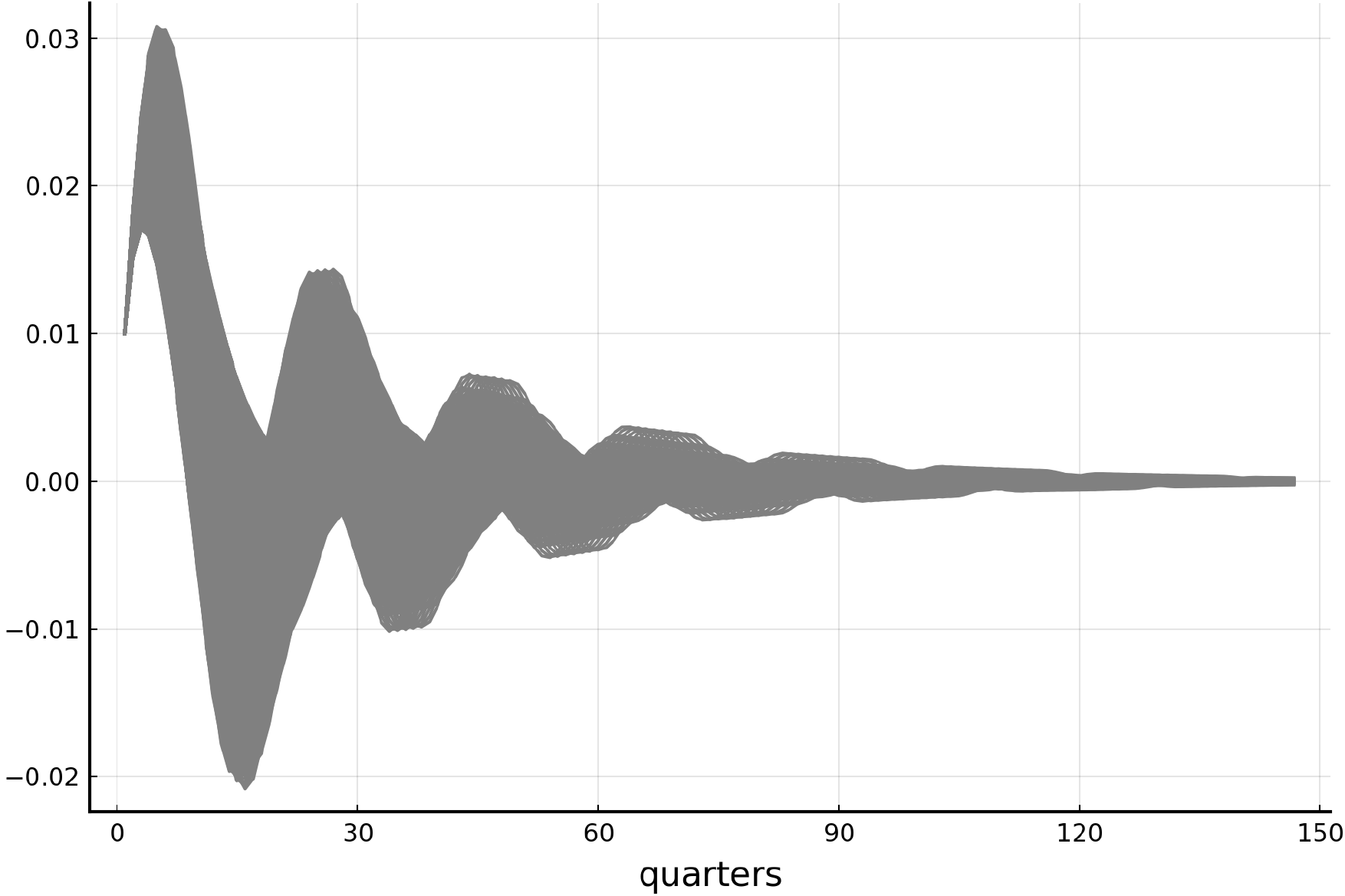}
\caption{Hours per capita \\ \hspace{1pt}}
\end{subfigure}
\hspace{1ex}
\begin{subfigure}{0.35\textwidth}
\includegraphics[width=\textwidth]{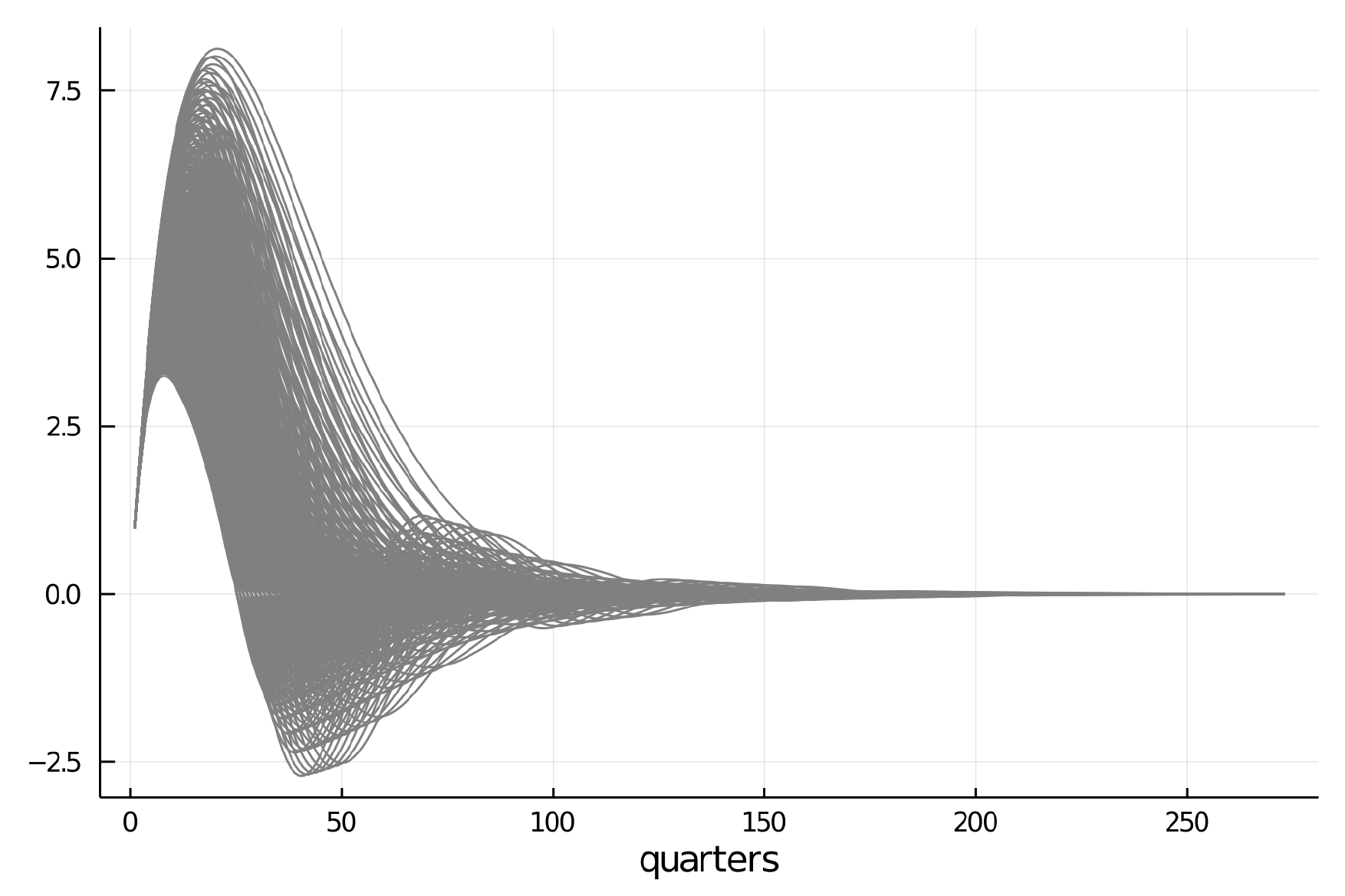}
\caption{Credit to private non-financial sector (\%  GDP)}
\end{subfigure}

\medskip

\begin{subfigure}{0.35\textwidth}
\includegraphics[width=\textwidth]{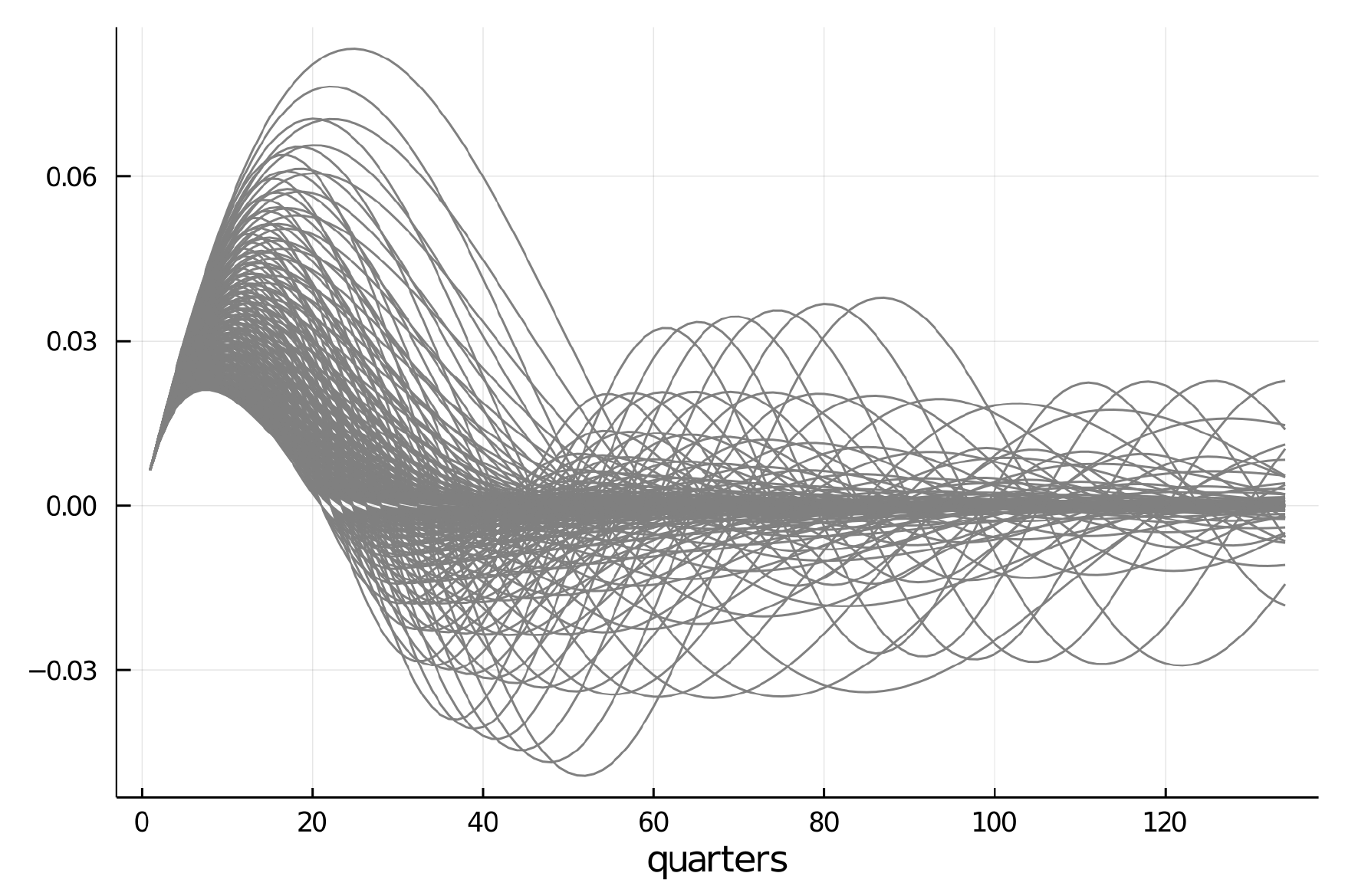}
\caption{Home price index \\ \hfill}
\end{subfigure}
\end{center}


 \bigskip
 \caption{Impulse responses to a one-standard-deviation shock to innovations. Each panel shows the impulse response functions for every cyclical specification within the 95\% confidence sets of $(\phi_1, \phi_2)$ \label{fig:irf} \vspace{1ex}\\
  \footnotesize 
  Notes: \textsuperscript{1}In panel A, the confidence set contains cyclical processes with high persistence exhibiting large amplitudes. For better visualization,  the impulse responses are split into two parts and drawn on two vertical axes of different scales; the right vertical axis is for the dotted lines.
 }
  
\end{figure}

In sum, our results suggest that business cycles (as marked by the expansions and contractions of aggregate economic activity) are not just recurrent but periodic, with an average duration of at least 5-6 years. Furthermore, financial cycles as characterized by the booms and busts in credit and home prices are much longer than business cycles: at least 10 years in duration. Additionally, these financial cycles have more prominent oscillations with amplitudes much larger than those of business cycles. Moreover, we find that equity prices, though commonly included in the characterization of financial cycles, do not exhibit stochastic cycles, and therefore merit separate consideration from credit and home prices. Lastly, our results suggest that asset market fluctuations are a different phenomenon from changes in real economic activities.


%% file: Periodogram.tex
\section{The periodogram of long-cycle processes}\label{sec:periodogram}
\subsection{Asymptotic properties}

Periodogram-based nonparametric estimators are commonly used for inference about the cyclical behavior of time series. In this section, we derive the asymptotic properties of the periodogram in the case of long-cycle processes. 
For $-\pi\leq \omega \leq \pi$, the periodogram of $\{y_t\}$ is defined as
\begin{equation}
I_n(\omega)\equiv\frac{1}{2\pi n} \left |\sum_{t=1}^n y_{t} e^{-i\omega t}\right |^2,
\end{equation}
see, for example, equation (6.1.24) in \citet{priestley1981spectral}. In the case of covariance stationary processes with continuous spectral densities, it is well known that the periodogram is an asymptotically unbiased estimator of the spectral density at $\omega$ \citep[see  equation (6.2.12) in][]{priestley1981spectral}.

Given the results of Proposition \ref{p:period}, in the case of long-cycle processes, we are interested in the spectrum near the origin at frequencies of the form $\omega_n=h/n$ for a constant $h\in\mathbb{R}$. Suppose that $\{y_{n,t}\}$ is generated according to the DGP in equation \eqref{eq:DGP_n} with the roots as in Assumption \ref{a:roots}. Assume also that $\{u_t\}$ are serially uncorrelated with a zero mean and variance $\sigma^2_u$.  In this case, the spectral density of $\{y_t\}$, denoted $f_n(\omega)$, satisfies\footnote{See the proof of Proposition \ref{p:EI}.}
\begin{align*}
n^{-4} f_n(h/n) &=\frac{\sigma^2_u}{2\pi n^4}\frac{1}{\vert 1-\lambda_{1,n}e^{ih/n}\vert^2\vert 1-\lambda_{2,n}e^{ih/n}\vert^2}\\
&\to \frac{\sigma^2_u}{2\pi} \frac{1}{(c^2+(d+h)^2)(c^2+(d-h)^2)}.
\end{align*}
We show in the following that, in the case of long-cycle processes and near the origin frequencies, the periodogram is a biased estimator.
\begin{proposition}\label{p:EI}
Suppose that $\{y_{n,t}\}$ is generated according to equation \eqref{eq:DGP_n} and Assumption \ref{a:roots}, and 
$\{u_t\}$ are serially uncorrelated with a zero mean and the variance $\sigma^2_u>0$.
Then, for a constant $h$, the periodogram of $\{y_{n,t}\}$ satisfies
\begin{equation}\label{eq:EI}
\lim_{n\to\infty} n^{-4} E \left [I_n\left(\frac{h}{n}\right) \right ]= \frac{\sigma^2_u}{\pi}
\int_{-\infty}^{\infty}\frac{1-\cos(h-x)}{(h-x)^2}\frac{1}{(c^2+(d+x)^2)(c^2+(d-x)^2)}dx.
\end{equation}
\end{proposition}
The result can be extended to allow strictly stationary and serially correlated $\{u_t\}$, when the spectral density $\varphi(\omega)$ of $\{u_t\}$ is bounded, bounded away from zero, and continuously differentiable with the derivative satisfying $\sup_{x\in[-\pi n,\pi n]} |\varphi'(x/n)| =O(n^{-1})$.  For example, the condition holds when $\{u_t\}$ is an MA($p$) process.
In that case, $\sigma^2_u$ in equation \eqref{eq:EI} should be replaced with $\varphi(0)$. 

The result in Proposition \ref{p:EI} can be used to assess the magnitude of the bias implied by the periodogram, as we illustrate below. When the cyclical properties of a process are assessed using its spectrum, the appropriate measure of the cycle length is $\tau_\omega$. The solid line in Figure \ref{fig:EI} shows the limiting expression for the expected values of the periodogram of a long-cycle process at near-origin frequencies. Its maximizing frequency is shown by the solid vertical line. The dashed line displays the limit of the true spectral density.  The vertical dashed line indicates the true frequency maximizing the spectrum $\sqrt{d^2-c^2}$ derived in Proposition \ref{p:period}. To construct the plot, we use the following values of the localization parameters: $c=4$ and $d=10$. 

The numerical results displayed in the figure demonstrate that the periodogram may underestimate the spectrum maximizing frequency and, as a result, overestimate the length of the cycle. According to the true spectrum, the cycle length relative to the sample size is $\tau_\omega=0.69$, while according to the periodogram it is $\tau_\omega=0.73$. For quarterly data and a sample size $n=200$, this corresponds to the upward bias of 8 quarters for the length of a cycle.
\begin{figure}
  \includegraphics[width=0.6\textwidth]{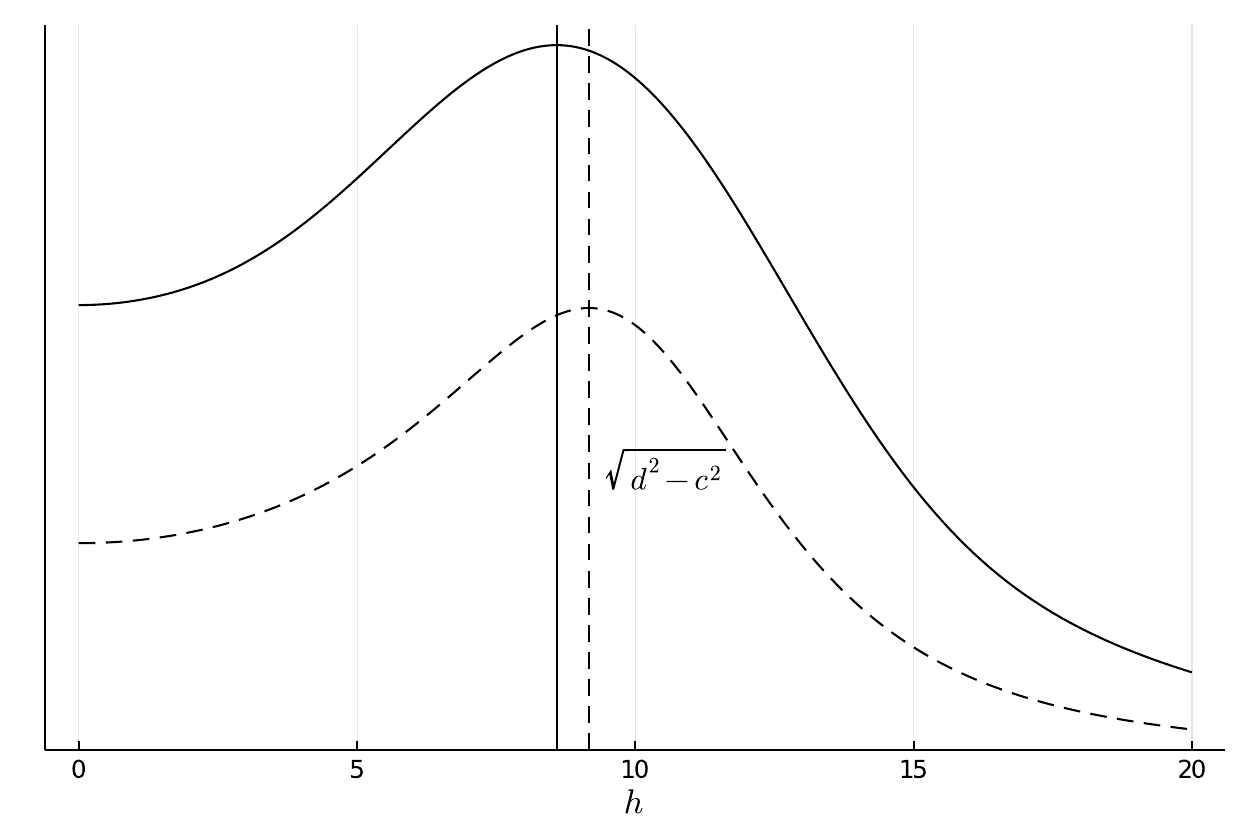}
  \caption{The limits of the expected value of the periodogram (solid line) and the true spectrum (dashed line) for $c=4$ and $d=10$. The corresponding vertical lines indicate the maximizing frequencies in terms of $h$, where $h$ is determined by $\omega_n=h/n$, and $\omega_n$ denotes frequencies }
  \label{fig:EI}
\end{figure}

Proposition \ref{prop:convergence_J} can be used to describe the asymptotic distribution of the periodogram of a long-cycle process at near-the-origin frequencies. The next result shows that asymptotic distribution of the periodogram depends on a continuous-time Fourier transform of the asymptotic approximation of the long-cycle process.
\begin{corollary}\label{cor:EI}
Suppose that $\{y_{n,t}\}$ is generated according to equation \eqref{eq:DGP_n}, and Assumptions \ref{a:roots} and  \ref{a:FCLT} hold. Then,
\begin{equation*}
n^{-4} I_n\left(\frac{h}{n} \right)\Rightarrow \frac{1}{2 \pi} \left | \int_0^1 J_{c,d}(r) e^{-i h r} dr \right |^2,
\end{equation*}
where the process $J_{c,d}(r)$ is defined in equation \eqref{eq:limit_process}.
\end{corollary}

\subsection{Proofs of the asymptotic properties of the periodogram}

\begin{proof}[Proof of Proposition \ref{p:EI}]
The spectral density of $\{y_{n,t}\}$ is given by
\begin{equation}\label{eq:true_spectrum}
f_n(\omega)=\frac{\sigma^2_u}{2\pi}\frac{1}{|1-\lambda_{n,1}e^{i\omega}|^2 |1-\lambda_{n,2}e^{i\omega}|^2}.
\end{equation}
By the results in \cite{priestley1981spectral}, equations (6.2.10)--(6.2.11),
\begin{equation}\label{eq:EI_int}
EI_n\left(\frac{h}{n}\right)=\int_{-\pi}^{\pi} f_n(x) F_n\left(x-\frac{h}{n}\right)dx 
=\frac{1}{ n}\int_{-\pi n}^{\pi n} f_n\left(\frac{x}{n}\right) F_n\left(\frac{x-h}{n}\right)dx,
\end{equation}
where the second result holds by the change of variable, and
\begin{equation*}
F_n(x)=\frac{\sin^2(nx/2)}{n\sin^2(x/2)}=\frac{1-\cos(nx)}{n(1-\cos(x))}.
\end{equation*}
Applying a series expansion $\cos((h-x)/n)=1-0.5 ((h-x)/n)^2+O((h-x)/n)^4$, we obtain
\begin{equation}\label{eq:F_exp}
F_n\left(\frac{h-x}{n}\right) =\frac{2n(1-\cos(h-x))}{(h-x)^2\left(1+O\left(\frac{h-x}{n}\right)^2\right)}.
\end{equation}

Next, we consider an expansion of the elements of $f_n(x/n)$.
\begin{eqnarray*}
&& 1-\lambda_{n,1}e^{ix/n} \\
&=& 1-e^{c/n}\left(\cos\left(\frac{d+x}{n}\right)+i\sin\left(\frac{d+x}{n}\right)\right) \\
&=& 1-\left(1+\frac{c}{n}+O\left(\frac{1}{n^2}\right)\right) \\
&& \quad \times \left( 
1-\frac{1}{2} \left(\frac{d+x}{n}\right)^2+O \left(\frac{d+x}{n}\right)^4
+i\left(\frac{d+x}{n}+O \left(\frac{d+x}{n}\right)^3\right)
\right)\\
&=& -\frac{c}{n}+O\left( \left(\frac{d+x}{n}\right)^2+
 \frac{(d+x)^2}{n^3} +\left(\frac{d+x}{n} \right)^4 +\frac{1}{n^2}\right) \\
&& \quad + i\left( \frac{d+x}{n} +
 O\left(\frac{d+x}{n^2}+\frac{(d+x)^3}{n^4}\right)
 \right).
\end{eqnarray*}
Hence,
\begin{equation}\label{eq:1_minus_1}
 \left |1-\lambda_{n,1}e^{ix/n} \right |^2
= \frac{ c^2  +(d+x)^2}{n^2}+
O\left(\frac{d+x}{n}\right)^4.
\end{equation}
Similarly,
\begin{equation}\label{eq:1_minus_2}
 \left |1-\lambda_{n,2}e^{ix/n} \right |^2
= \frac{ c^2  +(d-x)^2}{n^2}+
O\left(\frac{d-x}{n}\right)^4.
\end{equation}
By \eqref{eq:true_spectrum} and \eqref{eq:1_minus_1}--\eqref{eq:1_minus_2},
\begin{equation}
{\frac{1}{n} f_n\left(\frac{x}{n}\right) =} 
  \frac{n^3\sigma^2}{2\pi\left(c^2+(d+x)^2\left(1+O\left(\frac{d+x}{n}\right)^2\right)\right)\left(c^2+(d-x)^2\left(1+O\left(\frac{d-x}{n}\right)^2\right)\right)}. \label{eq:f_exp}
\end{equation}
The result of the proposition follows from \eqref{eq:EI_int}, \eqref{eq:F_exp}, and \eqref{eq:f_exp}.

\end{proof}

\begin{proof}[Proof of Corollary \ref{cor:EI}]
Since
\begin{equation*}
\int_{t/n}^{(t+1)/n} e^{-ihs}ds = \frac{e^{-iht/n} }{n} \left( 1 + O(n^{-1}) \right),
\end{equation*}
we have
\begin{eqnarray}
 n^{-5/2}\sum_{t=1}^n y_{t} e^{-i h t/n} 
&=& \frac{1}{1+O(n^{-1})}\sum_{t=1}^n {n^{-3/2}}{y_{t}} \int_{t/n}^{(t+1)/n} e^{-i h s}\dd s \notag \\
&=& \frac{1}{1+O(n^{-1})}\sum_{t=1}^n  \int_{t/n}^{(t+1)/n} {n^{-3/2}}{y_{\floor{ns}}} e^{-i h s}\dd s \notag \\
&=& \frac{1}{1+O(n^{-1})}  \int_{0}^{1} {n^{-3/2}}{y_{\floor{ns}}} e^{-i h s}\dd s + O_p(n^{-1}) \notag \\
&\Rightarrow & \int_0^1 J_{c,d} e^{-i h s}\dd s, \label{eq:Int_e}
\end{eqnarray}
where the equality in the third line holds because $y_{\floor{ns}}=y_t$ for $t/n\leq s <(t+1)/n$, and the equality in the last line holds by the Continuous Mapping Theorem (CMT) and Proposition \ref{prop:convergence_J}.
The result of the corollary follows by \eqref{eq:Int_e} and the CMT.
\end{proof}

%% file: data_description.tex

\section{Description of the data in Section \ref{sec:empirical}}\label{sec:data_description}

Table \ref{tbl:data description} in this appendix provides a description of the variables used in the empirical application in Section \ref{sec:empirical}. The description includes the source with exact identifiers for each variable, any transformations applied to the raw data, and the sample periods.

\begin{table}[h!]
\caption{Data description}
\begin{adjustbox}{width=\textwidth}
\begin{threeparttable}[b]
\begingroup
\setlength{\tabcolsep}{10pt} 
\renewcommand{\arraystretch}{2} 
\begin{tabular}{@{} ccccc @{}}
\toprule
& Source & Identifier  &  Construction & Sample \\
\midrule
Real GDP per capita & FRED & A939RX0Q048SBEA & Natural logarithm & 1947Q1 to 2020Q2\\
Unemployment rate & FRED & UNRATE & Natural logarithm & 1948Q1 to 2020Q2\\
Hours per capita & FRED &  \makecell{HOANBS\\B230RC0Q173SBEA} &  \makecell{ Ratio of non-farm business hours  \\to population, Natural logarithm } & 1947Q1 to 2020Q2 \\
\makecell{VXO\\ S\&P 100 volatility index}  & FRED & VXOCLS &  Natural logarithm & 1986Q1 to 2020Q3\\
\makecell{Credit risk premium \\ BAA to 10Y } & FRED & BAA10YM & --- & 1953Q2 to 2020Q2\\
\makecell{Credit to non-financial sector\\ \% GDP}& FRED & QUSPAM770A &  --- &  1952Q1 to 2020Q1\\

Home price index & FRED & CSUSHPISA & \makecell{S\&P/Case-Shiller U.S. National \\ Home Price Index, Natural logarithm } & 1987Q1 to 2020Q2 \\
Equity price index  & FRED & \makecell{WILL5000IND\\CPALTT01USQ661S} &\makecell{Wilshire 5000 Total Market Index \\ divided by CPI,  Natural logarithm} & 1971Q2 to 2020Q2\\
\bottomrule
\end{tabular}

\begin{tablenotes}
\item[] In the case where aggregation is needed, the end-of-period values are used.
\end{tablenotes}

\endgroup
\end{threeparttable} \label{tbl:data description}
\end{adjustbox}
\end{table}

%% file: Proofs.tex

\section{Proofs of the main results}

\begin{proof}[Proof of Proposition \ref{p:period}]
By Assumption \ref{a:roots} and \eqref{eq:phi_1n}--\eqref{eq:phi_2n},
\begin{align*}
 -\frac{\phi_{1,n}(1-\phi_{2,n})}{4\phi_{2,n}}&=0.5\cos(dn^{-1})(\exp(cn^{-1})+\exp(-cn^{-1})) \\
 &=(1-0.5d^2n^{-2}+O(n^{-4}))(1+0.5c^2n^{-2}+O(n^{-4})) \\
 &=1-0.5(d^2-c^2)n^{-2}+O(n^{-4}).
\end{align*}
Since the argument of $\cos^{-1}$ converges to one, it follows that
\begin{equation}
\omega^*_n=\cos^{-1}\left( 1-0.5(d^2-c^2)n^{-2}+O(n^{-4})\right)=o(1).
\end{equation}
Consider $\cos^{-1}(1-s)=t$ or $1-s=\cos(t)$, where $s$ and $t$ are small. Expanding $\cos(t)$ around $t=0$, we obtain $s=t^2/2 +O(t^4)$. Hence, $2s=t^2(1+O(t^2))$, and it follows that
\begin{eqnarray*}
t &=& \sqrt{2s}(1 + O(t^2)) \\
 &=& \sqrt{2s}(1 + O(2s(1+O(t^2))) \\
 &=& \sqrt{2s} + O(s^{3/2}).
\end{eqnarray*}
Therefore,
\begin{equation*}
n\omega_n^{*} = \sqrt{d^2 - c^2} + O(n^{-2}).
\end{equation*}
\end{proof}


Lemmas \ref{alem:phi_taylor} and \ref{alem:stoch.proc} below present  auxiliary results needed for the proof of Proposition \ref{prop:convergence_G} and  Lemma \ref{lem:convergence_moments}. In particular, Lemma \ref{alem:stoch.proc}  establishes the properties of the diffusion processes that appear in the limiting expressions for the estimators and the test statistics.
  
\begin{lemma}\label{alem:phi_taylor}
Suppose that Assumption \ref{a:roots} holds. The following approximation holds for the long-cycle autoregressive coefficients in \eqref{eq:phi_1n} and \eqref{eq:phi_2n}:
\begin{enumerate}[(a)]
\item $\phi_{1,n} = 2 + \frac{2c}{n} + \frac{c^2-d^2}{n^2} + O(n^{-3})$.
\item $-\phi_{2,n} = 1 + \frac{2c}{n} + \frac{2c^2}{n^2} + O(n^{-3})$.
\item $\frac{\phi_{1,n}}{1-\phi_{2,n}} = 1 - \frac{d^2+c^2}{2n^2} + O(n^{-3})$.
\item $\frac{2}{1-\phi_{2,n}} = 1-\frac{c}{n} - \frac{c^2}{n^2} +O(n^{-3})$.
\item $\phi_{1,n} + \phi_{2,n} = 1 - \frac{c^2 + d^2}{n^2} + O(n^{-3})$.
\item $-(\phi_{1,n} + \phi_{2,n})\phi_{2,n} = 1 + \frac{2c}{n} + \frac{c^2-d^2}{n^2} + O(n^{-3})$.
\item $1 - (\phi_{1,n} + \phi_{2,n})^2 = \frac{2(c^2 + d^2)}{n^2} + O(n^{-3})$.
\item $\phi_{2,n}^2 = 1 + \frac{4c}{n} + \frac{8c^2}{n^2} + O(n^{-3})$. 
\end{enumerate}
\end{lemma}


\begin{lemma} \label{alem:stoch.proc}
The diffusion processes $J_{c,d}(\cdot)$, $K_{c,d}(\cdot)$, and $G_{c,d}(\cdot)$ have the following properties:
\begin{enumerate}[(a)]
\item $\dd J_{c,d}(r) = c\cdot J_{c,d}(r)\dd r + d \cdot K_{c,d}(r) \dd r = G_{c,d}(r) \dd r$.
\item $\dd K_{c,d}(r) = c \cdot K_{c,d}(r) \dd r - d\cdot  J_{c,d}(r) \dd r + \frac{1}{d} \dd W(r)$.
\item $\int_0^r e^{2c(r-s)} J_{c,d}(r)\dd s = \frac{1}{c^2+d^2} \Big\{\int_0^r e^{2c(r-s)} \dd W(s) - \Big(c \cdot J_{c,d}(r) + d \cdot K_{c,d}(r)\Big)\Big\}$.
\item $\dd \big(J_{c,d}(r) \cdot K_{c,d}(r)\big) = 2c\cdot J_{c,d}(r)\cdot K_{c,d}(r) \dd r + d\cdot(K^2_{c,d}(r)- J^2_{c,d}(r)) \dd r + \frac{1}{d} J_{c,d}(r) \dd W(r) $.
\item $\int_0^1 G^2_{c,d}(r) \dd r =(c^2+d^2)\int_0^1 J^2_{c,d}(r) \dd r + J_{c,d}(1)G_{c,d}(1) -\int_0^1 J_{c,d}(r) \dd W(r) - c\cdot J_{c,d}^2(1)$
\item $J^2(1) = 2\int_0^1 J_{c,d}(r) G_{c,d}(r) \dd r$
\item $(G_{c,d}^2(1)-1)/2 = c \int_0^1 G^2_{c,d}(r) \dd r + cd \int_0^1 K_{c,d}(r) G_{c,d}(r) \dd r - d^2 \int_0^1 J_{c,d}(r) G_{c,d}(r) \dd r + \int_0^1 G_{c,d}(r) \dd W(r)$.
\end{enumerate}
\end{lemma}

\begin{proof}[Proof of Lemma \ref{alem:stoch.proc}]
To prove part (a) and (b), note that by applying trigonometric identities, we have
\begin{eqnarray*}
J_{c,d}(r) &=& \frac{1}{d} \int_0^{r} e^{c(r-s)} \left\{\sin(dr)\cos(ds) - \cos(dr)\sin(ds)\right \}\dd W(s),\\
K_{c,d}(r) &=& \frac{1}{d} \int_0^{r} e^{c(r-s)}\left\{\cos(dr)\cos(ds) + \sin(dr)\sin(ds) \right\} \dd W(s).
\end{eqnarray*}
By applying stochastic differentiation of $J_{c,d}(r)$ and $K_{c,d}(r)$, 
\begin{eqnarray*}
d \cdot \dd J_{c,d}(r) &=& (c \cdot e^{cr} \sin(dr) + d \cdot e^{cr} \cos(dr)) \int_0^r e^{-cs} \cos(ds) \dd W(s) \cdot \dd r\\
&& + e^{cr} \sin(dr) e^{-cr} \cos(dr) \dd W(r) \\
&& - (c \cdot e^{cr} \cos(dr) - d \cdot e^{cr} \sin(dr)) \int_0^r e^{-cs} \sin(ds) \dd W(s) \cdot \dd r\\
&& - e^{cr} \sin(dr) e^{-cr} \cos(dr) \dd W(r) \\
&=& c \int_0^r e^{c(r-s)} \left \{\sin(dr)\cos(ds) - \cos(dr)\sin(ds) \right \} \dd W(s) \cdot \dd r\\
&& + d \int_0^r e^{c(r-s)} \left \{\cos(dr)\cos(ds) + \cos(dr)\sin(ds)\right \} \dd W(s) \cdot \dd r, \\
d \cdot \dd K_{c,d}(r) &=& (c \cdot e^{cr} \cos(dr) - d \cdot e^{cr} \sin(dr))  \int_0^r e^{-cs} \cos(ds) \dd W(s) \cdot \dd r \\
&& + e^{cr} \cos(dr) e^{-cr} \cos(dr) \dd W(r) \\
&& + (c \cdot e^{cr} \sin(dr) + d \cdot e^{cr} \cos(dr)) \int_0^r e^{-cs} \sin(ds) \dd W(s) \cdot \dd r \\
&& + e^{cr} \sin(dr) e^{-cr} \sin(dr) \dd W(r) \\
&=& c \int_0^r e^{c(r-s)} \left \{\cos(dr)\cos(ds) + \sin(dr)\sin(ds) \right \} \dd W(s) \cdot \dd r \\
&& + d \int_0^r e^{c(r-s)} \left \{\sin(dr)\cos(ds) - \cos(dr)\sin(ds)\right \} \dd W(s) \cdot \dd r \\
&& + \dd W(r).
\end{eqnarray*}
Parts (a) and (b) now follow from the trigonometric identities. To prove part (c), use the results from (a) and (b) and evaluate the following integrals using integration by parts:
\begin{eqnarray*}
\int_{0}^{r} e^{2c(r-s)} J_{c,d}(s)\dd s &=& \frac{d}{c} \int^{r}_0 e^{2c(r-s)} K_{c,d}(s)\dd s - \frac{1}{c} J_{c,d}(r), \\
\int_{0}^{r} e^{2c(r-s)} K_{c,d}(s)\dd s &=& \frac{1}{cd} \int^r_0  e^{2c(r-s)} \dd W(s) - \frac{d}{c} \int^r_{c} e^{2c(r-s)} J_{c,d}(s)\dd s - \frac{1}{c} K_{c,d}(r). 
\end{eqnarray*}
With some algebraic manipulations, we obtain part (c).

By Ito's lemma, 
\begin{eqnarray*}
\dd \big(J_{c,d}(r)\cdot K_{c,d}(r)\big) &=& \dd J_{c,d}(r) \cdot  K_{c,d}(r) + J_{c,d}(r) \cdot  \dd K_{c,d}(r).
\end{eqnarray*}
Note that the quadratic covariation is negligible in this case. Using (a) and (b), part (d) follows immediately.

Next, we proceed to prove (e). From (d), it follows that
\begin{eqnarray*}
d \cdot J_{c,d}(1)K_{c,d}(1) = 2cd \int_0^1 J_{c,d}(r) K_{c,d}(r) \dd r + d^2 \int_0^1 (K^2_{c,d}(r)- J^2_{c,d}(r)) \dd r +  \int_0^1 J_{c,d}(r) \dd W(r),
\end{eqnarray*}
By the definition of $G_{c,d}(\cdot)$,
\begin{eqnarray*}
J_{c,d}(1)G_{c,d}(1) = c\cdot J^2_{c,d}(1) + d \cdot  J_{c,d}(1)K_{c,d}(1).
\end{eqnarray*}
By applying the two results from above, we obtain the result in (e):
\begin{eqnarray*}
\int_0^1 G^2_{c,d}(r) \dd r &=& c^2  \int_0^1 J^2_{c,d}(r)\dd r + 2cd  \int_0^1  J_{c,d}(r)  K_{c,d}(r) \dd r + d^2 \int_0^1 K^2_{c,d}(r)\dd r \\
&=& (c^2 + d^2) \int_0^1 J^2_{c,d}(r)\dd r + d \cdot J_{c,d}(1)K_{c,d}(1) - \int_0^1 J_{c,d}(r) \dd W(r)\\
&=& (c^2 + d^2) \int_0^1 J^2_{c,d}(r)\dd r + \cdot J_{c,d}(1)G_{c,d}(1) - c\cdot J_{c,d}^2(1) - \int_0^1 J_{c,d}(r) \dd W(r).
\end{eqnarray*} 

To prove (f) and (g), we use stochastic differentiation of $J_{c,d}^2(r)$ and  $G_{c,d}^2(r)$, respectively:
\begin{eqnarray*}
\dd J_{c,d}^2(r) &=& 2J_{c,d}(r)\dd J_{c,d}(r) =2 J_{c,d}(r) G_{c,d}(r) \dd r, \\
\dd G_{c,d}^2(r) &=& 2G_{c,d}(r)\dd G_{c,d}(r)  + (\dd G_{c,d}(r))^2 \\
&=& 2G_{c,d}(r) (c\cdot \dd J_{c,d}(r) + d\cdot \dd K_{c,d}(r)) + \dd r\\
&=& 2c\cdot G_{c,d}(r)G_{c,d}(r)\dd r + 2cd\cdot G_{c,d}(r) K_{c,d}(r) \dd r - 2d^2 \cdot G_{c,d}(r) J_{c,d}(r) \\&&+ 2G_{c,d}(r)\dd W(r)+ \dd r.
\end{eqnarray*}
The results in (f) and (g) follow by integrating both sides of the stochastic differential equations above with respect to $r$ over $[0, 1]$.

\end{proof}

\begin{proof}[Proof of Proposition \ref{prop:convergence_G}]
By Lemma \ref{alem:phi_taylor}(a) and (b),
\begin{eqnarray}
y_t &=& \bigg(2 + \frac{2c}{n} + \frac{c^2-d^2}{n^2} + O(n^{-3})\bigg)y_{t-1} - \bigg(1 + \frac{2c}{n} + \frac{2c^2}{n^2} + O(n^{-3})\bigg)y_{t-2} + u_t,\; \text{and} \notag \\
\Delta y_t &=& \bigg(1+\frac{2c}{n}\bigg)\Delta y_{t-1} + \bigg(\frac{c^2-d^2}{n^2} + O(n^{-3})\bigg)
y_{t-1} - \bigg( \frac{2c^2}{n^2} + O(n^{-3})\bigg)y_{t-2} + u_{t} \notag \\
&=& \sum_{j=0}^t \bigg(1+\frac{2c}{n}\bigg)^{t-j} u_j + \bigg(\frac{c^2-d^2}{n^2} + O(n^{-3})\bigg)\sum_{j=0}^t \bigg(1+\frac{2c}{n}\bigg)^{t-j} y_{j-1} \notag \\ &&- \bigg( \frac{2c^2}{n^2} + O(n^{-3})\bigg)\bigg)\sum_{j=0}^t \bigg(1+\frac{2c}{n}\bigg)^{t-j} y_{j-2}. \label{eq:Delta_expansion}
\end{eqnarray}
Define $S_n(r)\equiv \sum_{t=1}^{\floor{nr}}u_t$. We have:
\begin{eqnarray*}
\Delta y_{\floor{nr}} &=& \sum_{j=0}^{\floor{nr}} \bigg(1+\frac{2c}{n}\bigg)^{\floor{nr}-j} \int^{\frac{j}{n}}_{\frac{j-1}{n}}\dd S_n(s) \\ 
&&+ \bigg(\frac{c^2-d^2}{n} + O(n^{-2})\bigg) \sum_{j=0}^{\floor{nr}}\int^{\frac{j}{n}}_{\frac{j-1}{n}}\bigg(1+\frac{2c}{n}\bigg)^{\floor{nr}-j} y_{\floor{n\frac{j-1}{n}}}\dd s \\ 
&&- \bigg( \frac{2c^2}{n} + O(n^{-2})  \bigg)\sum_{j=0}^{\floor{nr}} \int^{\frac{j}{n}}_{\frac{j-1}{n}}\bigg(1+\frac{2c}{n}\bigg)^{{\floor{nr}}-j} y_{\floor{n\frac{j-2}{n}}} \dd s.
\end{eqnarray*}
By the CMT and Proposition \ref{prop:convergence_J},
\begin{align*}
n^{-1/2} \Delta y_{\floor{nr}} &\Rightarrow \sigma\int_{0}^{r} e^{2c(r-s)}\dd W(s) - \sigma (c^2 + d^2)  \int_{0}^{r} e^{2c(r-s)} J_{c,d}(s) \dd s,\\
& = \sigma (c\cdot J_{c,d}(r) + d \cdot K_{c,d}(r)),
\end{align*}
where the result in the last line follows by Lemma \ref{alem:stoch.proc}(c). The result of the proposition now follows by the definition of $G_{c,d}(r)$ in \eqref{eq:G}.
\end{proof}

\begin{proof}[Proof of Lemma \ref{lem:convergence_moments}]

Parts (a)--(c) follow immediately from Propositions \ref{prop:convergence_J} and \ref{prop:convergence_G} by the CMT.
To prove the result in part (d), 
by squaring both sides of equation \eqref{eq:DGP_n} and summing over $t$, we obtain:
\begin{align*}
\sum y_{t}^2 &= (\phi_{1,n} + \phi_{2,n})^2 \sum y_{t-1}^2 + \phi_{2,n}^2 \sum (\Delta y_{t-1})^2 + \sum u_t^2 \\
&\quad- 2(\phi_{1,n} + \phi_{2,n})\phi_{2,n} \sum y_{t-1}\Delta y_{t-1} + 2(\phi_{1,n} + \phi_{2,n}) \sum y_{t-1} u_{t} - 2\phi_{2,n} \sum \Delta y_{t-1}u_t.
\end{align*}
After rearranging and applying the results of Lemmas \ref{alem:phi_taylor} and above identities, we have:
\begin{align*}
\sum y_{t-1}u_t = \frac{c^2 + d^2}{n^2} \sum y_{t-1}^2 + y_n \Delta y_{n} - \sum (\Delta y_{t-1})^2 - \frac{2c}{n} \sum y_{t-1} \Delta y_{t-1} + O_p(n).
\end{align*}
By the results in parts (a)--(c) of the lemma, and using the shortened notation as explained on page \pageref{notation},
\begin{align*}
n^{-2}\sum y_{t-1}u_t & \Rightarrow \sigma^2 \bigg((c^2+d^2) \int J^2_{c,d} + J_{c,d}(1)G_{c,d}(1) -  \int G_{c,d}^2 -2c \int J_{c,d}G_{c,d} \bigg)\\
&= \sigma^2 \int J_{c,d} \dd W,
\end{align*}
where the result in the last line is by Lemma \ref{alem:stoch.proc}(e) and (f).

To prove part (e), we follow the same steps as in part (d) using \eqref{eq:ARn_transformed} to obtain
\begin{align*}
\frac{1}{n} \sum \Delta y_{t-1}u_t &=
 \frac{c^2 +d^2}{n^3} \sum y_{t-1} \Delta y_{t-1} - \frac{2c}{n^2} \sum (\Delta y_{t-1})^2 - \frac{1}{2n} \sum u_t^2 + \frac{1}{2n}(\Delta y_{n})^2 
 + O(n^{-1}) \\
 &\Rightarrow \sigma^2 (c^2 + d^2)\int J_{c,d}G_{c,d}  - 2c \sigma^2 \int G_{c,d}^2- \frac{1}{2}\sigma^2_u +\frac{1}{2}\sigma^2 G_{c,d}^2(1)\\
 &=\sigma^2 \int G_{c,d}\dd W + \frac{1}{2} (\sigma^2 - \sigma^2_u),
\end{align*}
where the equality in the last line is by part (g) of Lemma \ref{alem:stoch.proc} and the definition of $G_{c,d}$. 
\end{proof}

\begin{proof}[Proof of Proposition \ref{prop:LSphi}]
By  \eqref{eq:phi_12},
\begin{eqnarray*}
\label{aeq: estimator}
\begin{pmatrix}
\widehat{\phi}_{1,n} +\widehat{\phi}_{2,n}- \phi_{1,n}-\phi_{2,n}\\
\widehat{\phi}_{2,n} - \phi_{2,n}
\end{pmatrix} &=& \frac{1}{\sum y_{t-1}^2\sum (\Delta y_{t-1})^2 - (\sum y_{t-1} \Delta y_{t-1})^2} \\
&&\times
{\begin{pmatrix}
\sum (\Delta y_{t-1})^2 & \sum y_{t-1} \Delta y_{t-1}\\
\sum y_{t-1} \Delta y_{t-1} & \sum y_{t-1}^2
\end{pmatrix} \begin{pmatrix}
\sum y_{t-1}u_t\\
-\sum \Delta y_{t-1} u_t
\end{pmatrix}}.
\end{eqnarray*}
The result in part (a) and the result in part (b) for $\widehat \phi_{2,n}$ follow immediately by Lemma \ref{lem:convergence_moments} and the CMT. The result in part (b) for $\widehat \phi_{1,n}$ follows since 
\begin{align*}
n(\widehat{\phi}_{1,n} - \phi_{1,n}) &=
n(\widehat{\phi}_{1,n} +\widehat{\phi}_{2,n}- \phi_{1,n}-\phi_{2,n})
-n(\widehat{\phi}_{2,n}- \phi_{2,n})\\
&=O_p(n^{-1})-n(\widehat{\phi}_{2,n}- \phi_{2,n}),
\end{align*}
where the second equality holds by the result in part (a).

\end{proof}

\begin{proof}[Proof of Proposition \ref{prop:wald}] 
The result follows from Lemma \ref{lem:convergence_moments}(a)-(c) and Proposition \ref{prop:LSphi}, provided that $\hat\sigma^2_n\to_p\sigma^2$. The long-run variance estimator $\widehat\sigma^2_n$ is given by
\begin{align*}
\widehat \sigma^2_n &=\hat \sigma^2_{u,n}+ 2\sum_{h=1}^{m_n} w_n(h) n^{-1}\sum_{t=h+1}^n \hat u_t \hat u_{t-h}, \; \text{where} \;
\widehat \sigma^2_{u,n} = n^{-1}\sum_{t=1}^n \hat u_t^2.
\end{align*}
Denote $\phi_{12,n}\equiv\phi_{1,n}+\phi_{2,n}$ and $\widehat\phi_{12,n}\equiv\widehat\phi_{1,n}+\widehat\phi_{2,n}$. We have:
\begin{eqnarray*}
\widehat{\sigma}^2_{u,n} 
&=& \frac{1}{n} \sum u^2_t - (\widehat{\phi}_{12,n} - \phi_{12,n})\frac{2}{n} \sum y_{t-1}  u_t +(\widehat{\phi}_{2,n} - \phi_{2,n})\frac{2}{n} \sum\Delta y_{t-1}u_t \\ 
&&+ \frac{1}{n} \sum \bigg((\widehat{\phi}_{12,n} - \phi_{12,n})y_{t-1} + (\widehat{\phi}_{2,n} - \phi_{2,n})\Delta y_{t-1}\bigg)^2\\
&=&  \frac{1}{n} \sum u^2_t + O_p(n^{-1})\\
&\to_p& \sigma^2_u,
\end{eqnarray*}
where the equality in the line before the last holds by Lemma Lemma \ref{lem:convergence_moments}(d),(e) and Proposition \ref{prop:LSphi}, and the result in the last line holds by Assumption \ref{a:var}. By the same arguments and since the weight function $w_n(\cdot)$ is bounded, 
\begin{align*}
n^{-1}\sum_{t=h+1}^n \hat u_t \hat u_{t-h}
&= n^{-1}\sum_{t=h+1}^n  u_t  u_{t-h} +O_p(n^{-1}).
\end{align*}
Hence,
\begin{equation*}
\widehat\sigma^2_n =\tilde\sigma^2_n+O_p(m_n/n),
\end{equation*}
and the result follows by Assumption \ref{a:lrvar}.
\end{proof}

\begin{proof}[Proof of Lemma \ref{lem:convergence_moments_mean}]
By the results of Propositions \ref{prop:convergence_J}, \ref{prop:convergence_G}, and the CMT,
\begin{alignat*}{3}
n^{-3/2} \bar y_n/\sigma &\Rightarrow \int_0^1 J_{c,d}(s) \dd s ,\quad
n^{-1/2} \overline{\Delta y}_n /\sigma &&\Rightarrow \int_0^1 G_{c,d}(s) \dd s.
\end{alignat*}
Hence, 
\begin{align*}
n^{-3/2} (y_{\floor{nr}} -\bar y_n)/\sigma &\Rightarrow  J_{c,d}(s) -\int_0^1 J_{c,d}(r)\dd s = \widetilde J_{c,d}(r),\\
n^{-1/2} (\Delta y_{\floor{nr}}-\overline{\Delta y}_n)/\sigma &\Rightarrow G_{c,d}(r)-\int_0^1 G_{c,d}(s) =\widetilde G_{c,d}(r)\dd s.
\end{align*}
The results of the lemma now follow by the CMT using the same arguments as those in the proof of Lemma \ref{lem:convergence_moments}


\end{proof}

\begin{proof}[Proof of Lemma \ref{lem:convergence_moments_cycles}]
The results of the lemma follow by the same arguments as those in the proofs of Lemma \ref{lem:convergence_moments} and \ref{lem:convergence_moments_mean} after observing that $\int_0^1\cos^2(2 \pi k s) \dd s=\int_0^1\sin^2(2 \pi k s) \dd s=1/2$.
\end{proof}

\begin{proof}[Proof of Lemma \ref{lem:convergence_moments_ltt}]
The results of the lemma follow by the same arguments as those in the proofs of Lemma \ref{lem:convergence_moments} and \ref{lem:convergence_moments_mean} after observing that 
\[
\begin{pmatrix}
1 & \int_0^1 s \dd s\\
\int_0^1 s \dd s & \int_0^1 s^2 \dd s
\end{pmatrix}^{-1}
= 
\begin{pmatrix}
4 & -6\\
-6 & 12
\end{pmatrix}.
\]
\end{proof}

\begin{proof}[Proof of Proposition \ref{prop:Wald_AR_usual}] To simplify the presentation, we prove the result for $p=1$. For the general case, the proof is similar but requires more a complicated notation. Under $H_0$,
$\widetilde u_{t,0}=\widetilde u_t$, where $\{\widetilde u_t\}$  are the residuals from the projection of $\{u_t\}$ against the components of $D_t$. Since
\begin{align*}
(1-\phi_1 L - \phi_2 L^2)\widehat x_{t,0} &= \widetilde \varepsilon_t -(\widehat\rho_{1,0}-\rho_1) \widetilde u_{t-1},\\
\widehat\rho_{1,0}-\rho_1&=\frac{\sum  \widetilde u_{t-1} \varepsilon_t}{\sum  \widetilde u_{t-1}^2},
\end{align*}
the estimators of $\phi_1$ and $\phi_2$ satisfy:
\begin{eqnarray*}
\begin{pmatrix}
\widehat{\phi}_{1,n} - \phi_{1}\\
\widehat{\phi}_{2,n} - \phi_{2}
\end{pmatrix} = 
\begin{pmatrix}
\sum \widehat x_{t-1,0}^2 & \sum \widehat x_{t-1,0} \widehat x_{t-2,0}\\
\sum \widehat x_{t-1,0} \widehat x_{t-2,0} &  \sum \widehat x_{t-2,0}^2 
\end{pmatrix}^{-1} \begin{pmatrix}
\sum \widehat x_{t-1,0}(\widetilde \varepsilon_t -(\widehat\rho_{1,0}-\rho_1) \widetilde u_{t-1})\\
\sum \widehat x_{t-2,0}(\widetilde \varepsilon_t -(\widehat\rho_{1,0}-\rho_1) \widetilde u_{t-1})
\end{pmatrix},
\end{eqnarray*}
with
\begin{align}
\begin{pmatrix}
\sum \widehat x_{t-1,0}(\widetilde \varepsilon_t -(\widehat\rho_{1,0}-\rho_1)\widetilde u_{t-1})\\
\sum \widehat x_{t-2,0}(\widetilde \varepsilon_t -(\widehat\rho_{1,0}-\rho_1) \widetilde u_{t-1})
\end{pmatrix}
&=
\begin{pmatrix}
\sum \bigg(\widehat x_{t-1,0} - \frac{\sum\widehat x_{s-1,0}\widetilde u_{s-1}}{\sum \widetilde u_{s-1}^2}\widetilde u_{t-1}\bigg)\widetilde\varepsilon_t \\
\sum \bigg(\widehat x_{t-2,0} - \frac{\sum\widehat x_{s-2,0}\widetilde u_{s-1}}{\sum \widetilde u_{s-1}^2}\widetilde u_{t-1}\bigg)\widetilde\varepsilon_t 
\end{pmatrix}. \label{eq:xe}
\end{align}
The result follows since under the null, $\widehat \rho_{1,0}-\rho_1=O_p(n^{-1/2})$ and $\widehat x_{t,0}=\widetilde x_t-(\widehat \rho_{1,0}-\rho_1) \widetilde y_{t-1}$.
\end{proof}

\begin{proof}[Proof of Proposition \ref{prop:Wald_with_AR}] Similarly to the proof of Proposition  \ref{prop:Wald_AR_usual}, we prove the result for $p=1$. For the general case, the proof is analoguous, but requires more a complicated notation. 
Consider $\dot \zeta_{1,n}$ in \eqref{eq:dot}:
\[
\dot \zeta_{1,n}=\frac{\sum \widetilde u_t (\widetilde x_t-(\widehat \rho_{1,0}-\rho_1) \widetilde y_{t-1})}{\sum \widetilde u_t^2}
=O_p(n),
\]
where the second equality holds by the lemmas in Section \ref{sec:extensions}. Next, consider the elements of the matrix $\Sigma_n$:
\begin{align*}
n^{-4}  \sum \dot x^2_{t-1,0} &= n^{-4}  \sum (\widehat x_{t,0}-\dot \zeta_{1,n} \widetilde u_{t,0})^2 \\
&= n^{-4}  \sum \widetilde x_{t}^2 +o_p(1)\\
&\Rightarrow \sigma^2_\varepsilon \int \widetilde J_{c,d}^2,
\end{align*}
where the results in the second and third lines hold again by the lemmas in Section \ref{sec:extensions}.
After applying the same arguments to the other elements in $\Sigma_n$, the elements of $M_n$, and the expressions on the right-hand side of \eqref{eq:xe}, the result of the proposition follows by the CMT.

\end{proof}

%% file: unit_root_tests.tex
\section{Long cycles and the local-to-unity model}

In this appendix, we illustrate that one can discriminate between long cycles and the local-to-unity DGP, and in sufficiently large samples the confidence sets for $c,d$ and the local-to-unity parameter can not be non-empty at the same time.\footnote{We thank the associate editor for bringing up this issue.} 

For the purpose of the illustration, we assume below that $\{u_t\}$ is serially uncorrelated,  $\sigma=1$ and is known.
Suppose $\{x_t\}$ is a local-to-unity process as in Section \ref{sec:Asy} and the result in \eqref{eq:near-unit} holds.  Consider the $t$-statistic $t_n(a_n)=(\hat a_n - a_n)(\sum_t x_{t-1})^{1/2}$, where $\hat a_n=\sum_t x_t x_{t-1}/\sum_t x_{t-1}$. In this case, $t_n(a_n)=O_p(1)$ as  $t_n(a_n) \Rightarrow \int_0^1 J_c(s) \dd W(s)/(\int_0^1 J_c^2(s)\dd s)^{1/2}$ \citep{phillips1987towards}.

Next, consider the above $t$-statistic computed using a long-cycle process $\{y_{t}\}$ satisfying the assumptions of Proposition \ref{prop:convergence_J}: 
\begin{align*}
t_n(e^{b/n}) & = \left(\sum_t y^2_{t-1}\right)^{1/2} \left(\frac{\sum_t y_ty_{t-1}}{\sum_t y^2_{t-1}} -e^{b/n}\right)\\
&= \left(\sum_t y^2_{t-1}\right)^{1/2} \left(\phi_{1,n}+\phi_{2,n}-e^{b/n}\right)  +e^{2c/n}\frac{\sum_t y_{t-1} \Delta y_{t-1}}{\left(\sum_t y^2_{t-1}\right)^{1/2}}+\frac{\sum_t y_{t-1} u_t}{\left(\sum_t y^2_{t-1}\right)^{1/2}},
\end{align*}
where $\phi_{1,n}+\phi_{2,n}- e^{b/n}= 2 e^{c/n} \cos(d/n) - e^{2c/n} - e^{b/n}=-b/n+o(1/n)$.
By Lemma \ref{lem:convergence_moments}(a),(b), and (d),
\begin{equation*}
n^{-1}t_n(e^{b/n}) \Rightarrow -b \left(\int J_{c,d}^2\right)^{1/2} + \frac{\int J_{c,d} G_{c,d}}{ \left(\int J_{c,d}^2\right)^{1/2}}.
\end{equation*}

Thus, if data are generated from a local-to-unity model, the $t$-statistic is $O_p(1)$. However, if  data are generated as a long-cycle process, the $t$-statistic is $O_p(n)$. Hence, in the case of a long-cycle process, a grid-based confidence set for a local-to-unity parameter such as in \citet{hansen1999grid} would be empty in sufficiently large samples. Note that in the case of local-to-unity DGPs, our confidence set for $c,d$  would be empty in sufficiently large samples as discussed above. 

%% file: Size_distortion.tex
\section{Size distortions from conventional critical values}\label{sec:distortions}
In the AR(2) model with complex roots, the cycle frequency can be directly inferred from the autoregressive coefficients. In this section, we discuss size distortions one may see when using conventional $\chi^2_2$ critical values in place of the quantiles of the distributions derived in equation \eqref{eq:wald_distribution} in Section \ref{sec:extensions}. For the purpose of this exercise, we assume that there is no serial correlation in $\{u_t\}$ and, as a result, the noncentrality term $0.5(1-\sigma^2_u/\sigma^2)$ is zero. Note that if $\{u_t\}$ are serially correlated, one can expect more substantial size distortions due to the presence of the noncentrality term in the asymptotic distribution.


Let $F_{c,d}(\cdot)$ denote the CDF of the asymptotic null distribution in \eqref{eq:wald_distribution}. Note that the CDF depends on the unknown localization parameters $c$ and $d$. Consider a test that rejects the null hypothesis when the Wald statistic exceeds the conventional critical value $\chi^2_{2, 1-\alpha}$, where $\chi^2_{2, 1-\alpha}$ is the $1-\alpha$  quantile of the $\chi^2_2$ distribution.  The asymptotic size of this test is 
 $1 - F_{c,d}(\chi^2_{2, 1-\alpha})$, 
and size distortion are given by the difference between the asymptotic size 
and the nominal size $\alpha$. Next, we examine the extent of the size distortions for different values of $c$ and $d$ in the case of the three specifications for the deterministic component $D_t$ in Section \ref{sec:extensions}.

Table \ref{tbl:AStbl_0.05} reports the asymptotic size 
for $\alpha = 0.05$ and different values of $c$ and $d$ for each of the three specifications of $D_t$.  The CDF $F_{c,d}(\cdot)$ is computed by Monte Carlo simulation with 100,000 replications and $J_{c,d}$ and $G_{c,d}$ processes generated using the Euler-Maruyama method with a time step $\Delta t = 0.01$.  The table also reports the length of the cycle as a fraction of the sample size measured by $\tau_{\theta}= 2\pi/d$. The smaller the value of $d$, the lower the oscillation frequency and the longer the cycle length relative to the sample size. 

In the case of the three specifications for $D_t$, the table shows similar patterns: the asymptotic size deviates from the nominal values of $0.05$ for the values of $c$ and $d$ closer to zero. However, the asymptotic size approaches the nominal value as $c$ becomes more negative and $d$ becomes more positive.

For example, in the model with a constant mean, the asymptotic size at $c = -1$ and $d = 5$ is 0.116, which means that the Wald test based on the conventional critical value $\chi_{2,1-\alpha}^2$ over-rejects the null by 0.066.
As we move down the rows and across the columns of Table \ref{tbl:AStbl_0.05}, the process becomes less persistent and with a shorter cycle period, and as a result the size distortions become negligible. Note, however,  that the relationship can be non-monotone. 

Although in the case of the constant mean model, the size distortions are relatively minor, they are much more prominent in the case of the  specifications with deterministic cycles and linear trends. In particular, the usage of conventional $\chi^2_2$ critical values may result in severe size distortions in the case of deterministic cycles. For example, when $c=-1$ and $d=5$, the null rejection probability is approximately 75\% instead of 5\%. It is approximately 32\% in the case of the linear time trend specification. Although $d=5$ corresponds to very long cycles as measured by $\tau_\theta$, the size distortions remain substantial even for shorter cycles. For example, in the specifications with deterministic cycles, the size of the conventional test is approximately 19\% for $c=-10, d=15$. These values correspond to $\tau_\theta=0.42$ and $\tau_\omega=0.56$. 

Note again that the size distortions can be non-monotone across the rows/columns. However, for large negative values of $c$ or large values of $d$, size distortions disappear. This is consistent with the results in \cite{phillips1987towards}, who shows that in the   local-to-unity  model, the null distribution of the $t$-statistic for the autoregressive coefficient converges to the standard normal as $c\to -\infty$.

In conclusion, depending on the values of $c$ and $d$, the expression on the right-hand side of \eqref{eq:wald_distribution}  can generate a wide range of different asymptotic distributions. The distributions can deviate substantially from the $\chi^2_2$ distribution for the values of $c,d$ sufficiently close to zero. Such specifications correspond to longer cycles. For the values of $c,d$ sufficiently far from zero, which correspond to shorter cycles, the distributions converge to the $\chi^2_2$ distribution. In particular, across all specifications of the deterministic component, the size distortions from using $\chi^2_2$ critical values become negligible for $\tau_\theta \leq 0.14$. However, when the length of the cycle as measured by $\tau_\theta$ exceeds 14\% of the sample size, the use of $\chi^2_2$ critical values  leads to size distortions. The distortions are typically more pronounced for longer cycles.

%
%

\begin{table}
	\caption{Asymptotic size of the conventional  Wald test with $\chi^2_{2, 1-\alpha}$ critical values for $\alpha = 0.05$ for different values of the localization parameters $c,d$ and different specifications of the deterministic component.}
	\begin{adjustbox}{width=0.75\textwidth}
		\begin{threeparttable}[b]
			\tiny
			\addtolength{\tabcolsep}{10pt} 
			\input{"AStbl_0.05.tex"}
			\addtolength{\tabcolsep}{-10pt} 
			\begin{tablenotes}
				\item[]$\tau_{\theta} = 2\pi/d$: cycle length as a fraction of the sample size.
			\end{tablenotes}
		\end{threeparttable} \label{tbl:AStbl_0.05}
	\end{adjustbox}
\end{table}

%% file: AStbl_0.05.tex
\begin{tabular}{@{} ccccccc @{}}
 \toprule 
 \footnotesize
& \multicolumn{6}{@{}c@{}}{$d$} \\ 
\cmidrule[0.01pt](lr){2-7}
$c$ &5 &15 &25 &35 &45 &55 \\ 
 \midrule 
&&&&&&\\ 
\multicolumn{7}{@{}c@{}}{\underline{constant mean}} \\ 
&&&&&&\\ 
 -1 &.116 &.059 &.070 &.054 &.051 &.050\\ 
 -5 &.110 &.060 &.052 &.062 &.054 &.050\\ 
 -10 &.089 &.064 &.053 &.049 &.056 &.051\\ 
 -15 &.077 &.064 &.055 &.051 &.049 &.061\\ 
 -20 &.070 &.062 &.056 &.051 &.049 &.049\\ 
 -70 &.051 &.051 &.050 &.050 &.049 &.048\\ 
&&&&&&\\ 
\multicolumn{7}{@{}c@{}}{\underline{deterministic cycle: $k=1$}} \\ 
&&&&&&\\ 
 -1 &.746 &.103 &.071 &.054 &.051 &.050\\ 
 -5 &.641 &.161 &.070 &.063 &.054 &.050\\ 
 -10 &.441 &.192 &.090 &.059 &.056 &.051\\ 
 -15 &.317 &.186 &.104 &.069 &.054 &.059\\ 
 -20 &.242 &.170 &.109 &.075 &.058 &.050\\ 
 -130 &.052 &.052 &.051 &.051 &.050 &.048\\ 
&&&&&&\\ 
\multicolumn{7}{@{}c@{}}{\underline{linear time trend}} \\ 
&&&&&&\\ 
 -1 &.317 &.071 &.072 &.053 &.051 &.050\\ 
 -5 &.257 &.089 &.058 &.062 &.054 &.051\\ 
 -10 &.188 &.101 &.066 &.053 &.056 &.051\\ 
 -15 &.146 &.101 &.072 &.057 &.051 &.060\\ 
 -20 &.121 &.096 &.074 &.060 &.052 &.049\\ 
 -100 &.053 &.052 &.052 &.052 &.050 &.049\\ 
&&&&&&\\ 
\cmidrule[0.01pt](lr){2-7}
$ \tau_\theta $&1.26&0.42&0.25&0.18&0.14&0.11\\ \bottomrule 
 \end{tabular}

%% file: Supplement.tex
\section{Consistency of the BIC under long-cycles specifications\label{sec:BIC}}

There is a rich literature in econometrics and statistics on the use of information criteria (IC) for the specification of time series models. For example, \cite{hannan1979determination} show the consistency of the BIC for the selection of lags in I($0$) autoregressive models; \cite{ng2001lag} propose an IC procedure for the selection of an autoregressive lag truncation parameter in unit root tests; \cite{marcellino2006comparison} consider the AIC and BIC for data-dependent lag order choices in a large-scale forecasting study with macroeconomic time series data.  In this section, we show that the BIC can be used to specify long-cycle models consistently.

First, consider a model with serially correlated $\{u_t\}$ but without deterministic components. That is, data are generated according to \eqref{eq:DGP_n} and \eqref{eq:AR_for_u}:
\begin{equation}\label{eq:S.AR(p+2)}
(1-\phi_{1,n} L-\phi_{2,n} L^2) y_{t} =u_t=(1-\rho_1 L-\ldots -\rho_p L^p)^{-1}\varepsilon_t,  
\end{equation}
where $\{\varepsilon_t\}$ are i.i.d. $(0,\sigma^2_\varepsilon)$, and the roots of the polynomial $1-\rho_1 L-\ldots -\rho_p L^p$ are real and bounded away from unity. We assume that the number of lags $p$  is unknown but bounded from above by a known integer constant $M>0$.  Under \eqref{eq:S.AR(p+2)}, $\{y_t\}$ is an AR($p+2$) process:
\begin{equation*}
    y_t=\zeta_{1,n}y_{t-1}+\ldots+\zeta_{p+2,n}y_{t-p-2}+\varepsilon_t,
\end{equation*}
where assuming $p>0$, $\zeta_{1,n}\equiv\phi_{1,n}+\rho_1$, $\zeta_{2,n}=\phi_{2,n}+\rho_2-\phi_{1,n}\rho_1$, $\zeta_{p+1,n}\equiv-\phi_{1,n}\rho_p- \phi_{2,n}\rho_{p-1} $, $\zeta_{p+2,n}\equiv-\phi_{2,n}\rho_p$, and  for $3\leq j\leq p$,  $\zeta_{j,n}\equiv\rho_j-\phi_{1,n}\rho_{j-1}- \phi_{2,n}\rho_{j-2}$. Similarly to \eqref{eq:ARn_transformed}, we transform the model into first and second differences to avoid singularities in the asymptotic distributions. Recall that $\Delta y_t\equiv y_t-y_{t-1}$, and let $\Delta^2 y_t\equiv\Delta y_t-\Delta y_{t-1}$. We have:
\begin{align}
    y_t & = \zeta_{1,n} y_{t-1}+\zeta_{2,n}y_{t-2}+\ldots+ (\zeta_{p+1,n}+\zeta_{p+2,n})y_{t-p-1}-\zeta_{p+2,n}\Delta y_{t-p-1}+\varepsilon_t\notag\\
    &=(\sum_{j=1}^{p+2}\zeta_{j,n})y_{t-1}-(\sum_{j=2}^{p+2}\zeta_{j,n})\Delta y_{t-1} -\ldots- (\zeta_{p+1,n}+\zeta_{p+2,n})\Delta y_{t-p}-\zeta_{p+2,n}\Delta y_{t-p-1}+\varepsilon_t\notag\\
    &=(\sum_{j=1}^{p+2}\zeta_{j,n})y_{t-1}-(\sum_{j=2}^{p+2}\zeta_{j,n})\Delta y_{t-1} -\ldots-  (\zeta_{p+1,n}+2\zeta_{p+2,n})\Delta y_{t-p}+ \zeta_{p+2,n}\Delta^2 y_{t-p}+\varepsilon_t\notag\\
    &=\varrho_{1,n} y_{t-1}+\varrho_{2,n} \Delta y_{t-1} + \nu_{1,n}\Delta^2 y_{t-1}+\ldots +\nu_{p,n}\Delta^2 y_{t-p}+\varepsilon_t,\label{eq:S.lags}
\end{align}
where $\varrho_{1,n}\equiv\sum_{j=1}^{p+2}\zeta_{j,n}$, $\varrho_{2,n}\equiv-\sum_{j=2}^{p+2}(j-1)\zeta_{j,n}$, $\nu_{1,n}\equiv\sum_{j=1}^{p} j\zeta_{j+2,n}$,  $\nu_{2,n}\equiv\sum_{j=2}^{p} (j-1)\zeta_{j+2,n}$, etc. Thus,  for $1\leq l\leq p$, $\nu_{l,n}\equiv\sum_{j=l}^{p}(j-l+1)\zeta_{j+2,n}$. 

The results of Lemma \ref{lem:convergence_moments} are extended below to cover expressions with $\Delta^2 y_{t-l}$  terms.
\begin{lemma}\label{lem:S.second_differences}
Suppose that $\{y_t\}$  are generated according to Assumption \ref{a:roots} and equation  \eqref{eq:S.AR(p+2)}, $\{\varepsilon_t\}$  are i.i.d. and have a zero mean and a finite variance,  and the roots of the polynomial $1-\rho_1 L-\ldots -\rho_p L^p$ are real and bounded away from one. Then,
\begin{enumerate}[(a)]
    \item $\Delta^2 y_t = u_t + O_p(n^{-1/2})$. \label{en:S.Delta^2}
    \end{enumerate}
    Moreover, for $l,s\geq 1$,
    \begin{enumerate}[(a)]
    \setcounter{enumi}{1}
    \item $\sum y_{t-1} \Delta^2 y_{t-l} = O_p(n^2)$. 
    \item $\sum \Delta y_{t-1} \Delta^2 y_{t-l} = O_p(n)$.
    \item $\sum \Delta ^2 y_{t-l} \Delta^2 y_{t-s} = O_p(n)$.
    \item $\sum \Delta ^2 y_{t-l} \varepsilon_t = O_p(n^{1/2})$.
\end{enumerate}

\end{lemma}

Combining \eqref{eq:S.lags} with \eqref{eq:yc+d}, we obtain the following specification:
\begin{align}
    y_{t-1}&= \alpha_n +\beta_n \cdot (t/n)+ \sum_{k\in \kappa} \left\{ \gamma_{1k,n} \cos\left(\frac{2\pi k t}{n}\right) + \gamma_{2k,n} \sin\left(\frac{2\pi k t}{n}\right)\right\}\notag \\
      &\qquad 
   + \varrho_{1,n} y_{t-1}+\varrho_{2,n} \Delta y_{t-1} + \nu_{1,n}\Delta^2 y_{t-1}+\ldots +\nu_{p,n}\Delta^2 y_{t-p}
    +\varepsilon_t\notag\\
    &=\alpha_n+\beta_n\cdot (t/n)+C_\kappa(t/n)'\gamma_n+\varrho_{1,n} y_{t-1}+\varrho_{2,n} \Delta y_{t-1}+z_{t-1,p}'\nu_n+\varepsilon_t,\label{eq:S.full_spec}
\end{align}
where $\kappa$ is a set of positive integers that determine the deterministic cyclical components,  $C_\kappa(x)=(\cos(2\pi k x),\sin(2\pi k x):k\in\kappa)'$ and $\gamma_n=(\gamma_{1k,n},\gamma_{2k,n}:k\in \kappa)'$ is the vector of corresponding coefficients, $z_{t-1,p}=(\Delta^2 y_{t-1},\ldots,\Delta^2 y_{t-p})'$ and $\nu_n=(\nu_{1,n},\ldots,\nu_{p,n})'$ is the vector of corresponding coefficients. We use the BIC to determine $\kappa$, $p$, and whether the linear trend should be included in the regression.  Note that while $\kappa$ is unknown, we assume that its elements are selected from a small known set of positive integers. For example, we consider $k=1,2,3$ in the main paper application section and allow for $\kappa=\emptyset$.

To determine $\kappa$, $p$, and whether the linear trend should be included, the econometrician computes the BIC:
\begin{align*}
    BIC_n(T,\kappa,p)\equiv n\log(n^{-1}SSR_n(T,\kappa,p))+(T+|\kappa|+p)\log(n),
\end{align*}
where $T\in\{0,1\}$ with zero indicating no linear trend and one indicating that the linear trend is included , $\kappa$  is a set of integers for $C_\kappa(t/n)$ considered in a specification, $|\kappa|$ is the number of elements in $\kappa$,  $p\in\{0,1,\ldots,M\}$ indicates the lag length in a specification: $z_{t-1,p}=(\Delta^2 y_{t-1},\ldots,\Delta^2 y_{t-p})'$,
\begin{align*}
    SSR_n(1,\kappa,p)&\equiv 
        \sum\left(y_{t}-\hat\alpha_n-\hat\beta_n\cdot (t/n)-C_\kappa(t/n)'\hat\gamma_n -\hat\varrho_{1,n} y_{t-1}-\hat\varrho_{2,n} \Delta y_{t-1}-z_{t-1,p}'\hat\nu_n\right)^2,\\
         SSR_n(0, \kappa,p)&\equiv 
        \sum\left(y_{t} -\hat\alpha_n-C_\kappa(t/n)'\hat\gamma_n-\hat\varrho_{1,n} y_{t-1}-\hat\varrho_{2,n} \Delta y_{t-1}-z_{t-1,p}'\hat\nu_n\right)^2,
\end{align*}
where $\hat\beta_n,\hat\gamma_n,\hat\nu_n,\hat\varrho_{1,n},\hat\varrho_{2,n}$ denote the OLS estimators of the corresponding coefficients in a specification. Let $\mathcal K$ denote the set of possible values for $\kappa$. The BIC estimators of $\kappa,p$, and the choice for the linear trend are given by
\begin{equation*}
    (\hat T_n, \hat \kappa_n,\hat p_n)\equiv \argmin_{T\in\{0,1\},\kappa\in \mathcal{K},p\leq M} BIC_n(T,\kappa,p).
\end{equation*}

The consistency of $\hat T_n, \hat \kappa_n,\hat p_n$ follows by standard arguments from the following lemma. Let $T_0,\kappa_0,p_0$ denote the true values of $T,\kappa, p$ respectively.
\begin{proposition} \label{prop:S.BIC}Suppose that $\{y_t\}$ is generated according to \eqref{eq:S.full_spec} and Assumption \ref{a:roots}, and $\{\varepsilon_t\}$ are i.i.d. with a zero mean and a finite variance $\sigma^2_\varepsilon$.
\begin{enumerate}[(a)]
    \item If $T_0\leq T$, $\kappa_0\subseteq\kappa$, and $p_0\leq p$, then $SSR_n(T,\kappa,p)=\sum\varepsilon^2_t+O_p(1)$. 
    \item If $T_0>T$, or $p_0>p$, or $\kappa_0\cap\kappa\ne\kappa_0$,  then $n^{-1}SSR_n(T,\kappa,p)\to_d \sigma^2_\varepsilon+B$, where the random variable $B$ is positive with probability one.
\end{enumerate}
    
\end{proposition}

Suppose that the model is over-specified: $T_0\leq T$, $\kappa_0\subseteq\kappa$, and $p_0\leq p$ with at least one relation held strictly. In that case, 
\begin{align*}
   & P\Bigg(BIC_n(T_0,\kappa_0,p_0)>BIC_n(T,\kappa,p)\Bigg) \\
   & = P\Bigg(n\log\left(\frac{n^{-1}SSR_n(T_0,\kappa_0,p_0)}{n^{-1}SSR_n(T,\kappa,p)}\right)>(T-T_0+|\kappa|-|\kappa_0|+p-p_0)\log(n)\Bigg)\\
    & = P\Bigg(n\log\left(\frac{\sigma^2_\varepsilon +O_p(n^{-1})}{\sigma^2_\varepsilon +O_p(n^{-1})}\right)>(T-T_0+|\kappa|-|\kappa_0|+p-p_0)\log(n)\Bigg)\\
   &= P\Bigg(n\log\left(1+O_p(n^{-1})\right)>(T-T_0+|\kappa|-|\kappa_0|+p-p_0)\log(n)\Bigg)\\
   &= P\Bigg(O_p(1)>(T-T_0+|\kappa|-|\kappa_0|+p-p_0)\log(n)\Bigg)\\
   &\to 0,
\end{align*}
where the second equality holds by Proposition \ref{prop:S.BIC}(a).

If the model is under-specified, i.e. there are missing terms, by Proposition \ref{prop:S.BIC}(b),
\begin{align*}
   P\Bigg(n^{-1}BIC_n(T_0,\kappa_0,p_0)>n^{-1}BIC_n(T,\kappa,p)\Bigg) \to_d P\Bigg(\log\left(\frac{\sigma^2_\varepsilon}{\sigma^2_\varepsilon+B}\right)>0\Bigg)=0.
\end{align*}
This establishes the consistency of the BIC selection procedure for the deterministic components and the lag length $p$ in the AR($p$) specification for $\{u_t\}$.

\section{Monte Carlo simulations}\label{sec:MC}

In this section, we use Monte Carlo simulations to study the effect of using the BIC to select a specification, that is, the time trend, deterministic cycles, and autoregressive lags for $\{u_t\}$,  on the size and power of our inferential procedure. For this purpose, we generate data from models as in \eqref{eq:yc+d} and \eqref{eq:AR_for_u}:
\begin{align*}
    y_t=(1-\phi_{1,n}L-\phi_{2,n}L^2)(1-\rho_1L)\Big(y_t-\mu-\beta (t/n) - \eta_1\cos(2\pi k t/n)-\eta_2 \sin(2\pi k t/n)\Big) =\varepsilon_t,
\end{align*}
where $\{\varepsilon_t\}$ are i.i.d. $N(0,0.1^2)$.

The coefficients $\phi_{1,n}$ and $\phi_{2,n}$ are generated according to \eqref{eq:phi_1n} and \eqref{eq:phi_2n} respectively. We consider two sets of values for $c$ and $d$: i) $c=-10$, $d=26$, and ii) $c=-150$, $d=80$. The first set corresponds to $\tau_\omega=0.26$ and $\tau_\theta=0.24$. The latter corresponds to a shorter cycle with $\tau_\theta=0.08$ ($\tau_\omega$ is not defined in that case).

For the autoregressive specifications for $\{u_t\}$, we consider $\rho=0$ (no serial correlation) and $\rho=0.5$ (serial correlation with an AR$(1)$ specification).

For the deterministic component, we set $\mu=0$, however, the regressions include the intercept, and the critical values are corrected accordingly. For the linear trend, we use $\beta=0$ (no trend) and $\beta=\sigma$ (trend), where $\sigma$ is the square root of the long-run variance of the corresponding specification. For the deterministic cyclical component, we use $k=1$, and $\eta_1=\eta_2=0$ (no deterministic cycles)  or $\eta_1=\eta_2=\sigma$ (deterministic cycles).

We generate $n=500$ observations and test the hypothesis $H_0:c=c_0,d=d_0$ vs. $H_1: c\ne c_0$ or $d\ne d_0$ over a grid with $c_0=-200,-199,\ldots,0$ and $d_0 = 10,12,\ldots,100$. The deterministic components and the serial correlation in $\{u_t\}$ are removed according to the selected specification as described in Section \ref{sec:inference}, and the corresponding critical values are used.

When selecting a model using the BIC, we search across the specifications with or without linear trend, no deterministic cycles or $k=1,2$, AR$(p)$ with $p=0,1,2$, and their combinations.

\begin{table}[]
    \centering
    \begin{tabular}{lcccccc}
  \toprule
  \addlinespace
    &   & \multicolumn{2}{c}{$c = -10$, $d = 26$} & &\multicolumn{2}{c}{$c = -150$, $d = 80$} \\ 
\cmidrule(lr){3-4} \cmidrule(lr){6-7} 
  &   & true model & BIC-selected &  & true model &  BIC-selected \\ \addlinespace\midrule
  \addlinespace[2ex]
  no trend &   & \multirow{3}{*}{0.046} & \multirow{3}{*}{0.054} &   & \multirow{3}{*}{0.048} & \multirow{3}{*}{0.049} \\
  no deterministic cycles &&&&&\\
 no serial correlation  &  &  &  &  &  &  \\\addlinespace[2ex]
  trend &   & \multirow{3}{*}{0.044} & \multirow{3}{*}{0.058} &   & \multirow{3}{*}{0.071} & \multirow{3}{*}{0.057} \\
  no deterministic cycles &&&&&&\\
  no serial correlation  &  &  &  &  &  &  \\\addlinespace[2ex]
 no trend&   & \multirow{3}{*}{0.040} & \multirow{3}{*}{0.052} &   & \multirow{3}{*}{0.053} & \multirow{3}{*}{0.055} \\
  deterministic cycles &&&&&&\\
  no serial correlation  &  &  &  &  &  &  \\\addlinespace[2ex]
  no trend  &   & \multirow{3}{*}{0.031} & \multirow{3}{*}{0.016} &   & \multirow{3}{*}{0.032} & \multirow{3}{*}{0.015} \\
   no deterministic cycles &&&&&&\\
  serial correlation &  &  &  &  &  &  \\\addlinespace[2ex]
  trend &   & \multirow{3}{*}{0.042} & \multirow{3}{*}{0.025} &   & \multirow{3}{*}{0.069} & \multirow{3}{*}{0.027} \\
   deterministic cycles &&&&&&\\
  serial correlation  &  &  &  &  &  &  \\
  \addlinespace[2ex]\bottomrule
\end{tabular}
    \caption{Simulated size of the nominal $0.05$-size test of $H_0:c=c_0,d=d_0$ vs. $H_1: c\ne c_0$ or $d\ne d_0$ for different choices of $c$ and $d$, different specifications of the deterministic component (trend, cycle) and serial correlation in $\{u_t\}$, using the true or BIC-selected specifications.}
    \label{tab:size}
\end{table}

Table \ref{tab:size} reports the simulated rejection probabilities based on $1,000$ Monte Carlo repetitions for each specification. Columns 1 and 3 are computed using the true specifications of the deterministic component and the serial correlation in $\{u_t\}$, and the tests in columns 2 and 4 are performed using the specifications selected by the BIC. According to our results, for the DGPs with $c=-10,d=25$, our procedure provides reliable control of the size with all deviations above the nominal $0.05$ within the expected simulation errors. Only minor size distortions are observed in two cases for the DGPs with $c=-150,d=80$. The distortions disappear with larger sample sizes (e.g., $n=1,000$).\footnote{The results for larger sample sizes are not reported here as they are similar to those with $n=500$ except for the two discussed cases with $c=-150,d=80$.} We conclude that the BIC selection procedure performs well and as expected, given the results of Section \ref{sec:BIC}.

\begin{figure}[]
    \centering
    \begin{subfigure}[b]{.49\textwidth}
    \centering
    \includegraphics[width=\textwidth]{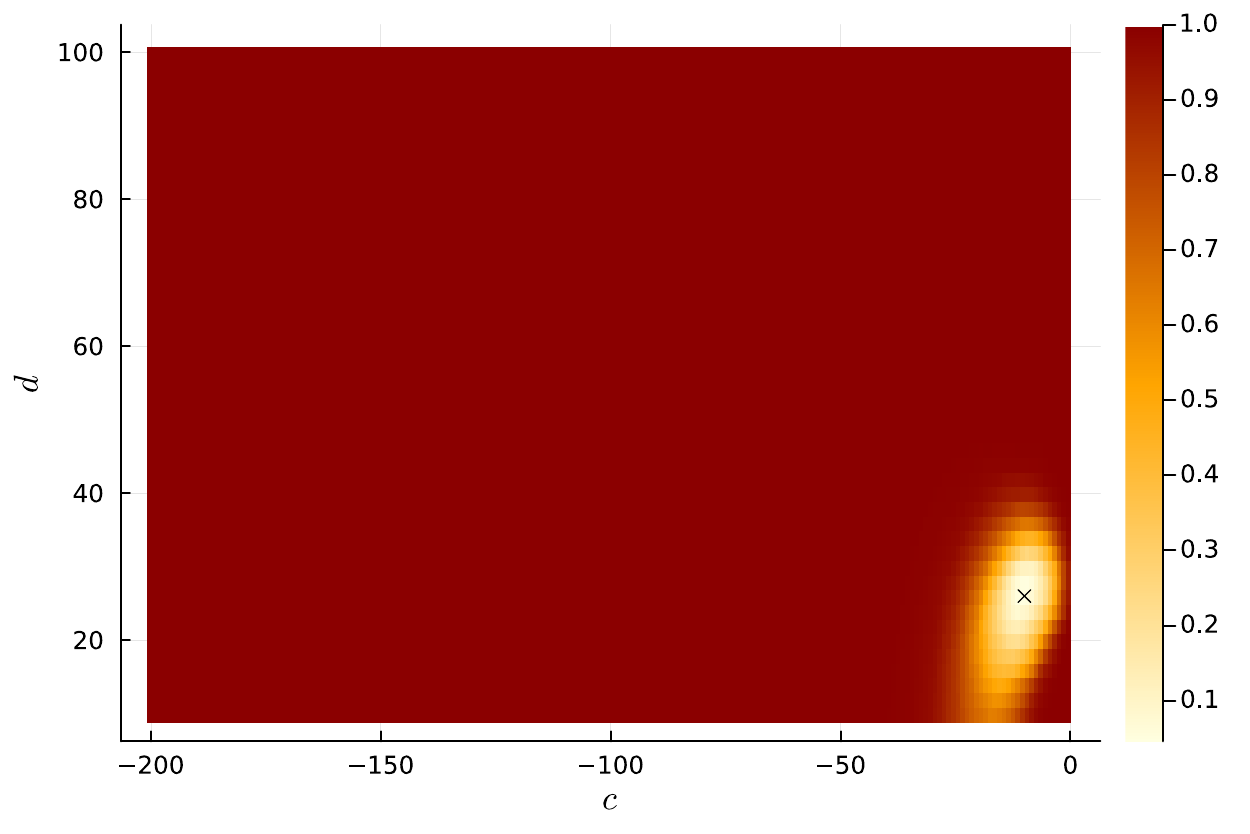}
    \caption{no trend, no cycle, no serial correlation}
    \end{subfigure}
    \hfill
    \begin{subfigure}[b]{.49\textwidth}
    \centering
    \includegraphics[width=\textwidth]{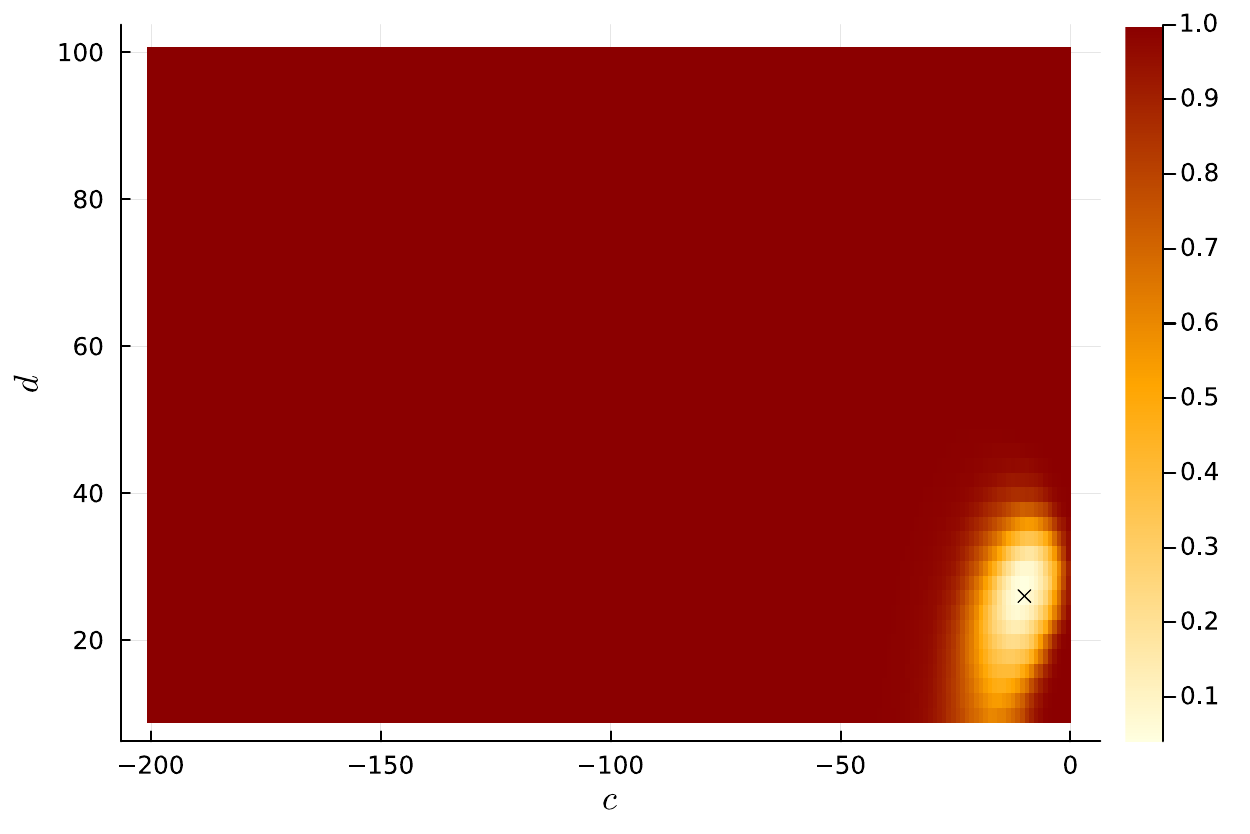}
    \caption{trend, no cycle, no serial correlation}
    \end{subfigure}
    \hfill
    \begin{subfigure}[b]{.49\textwidth}
    \centering
    \includegraphics[width=\textwidth]{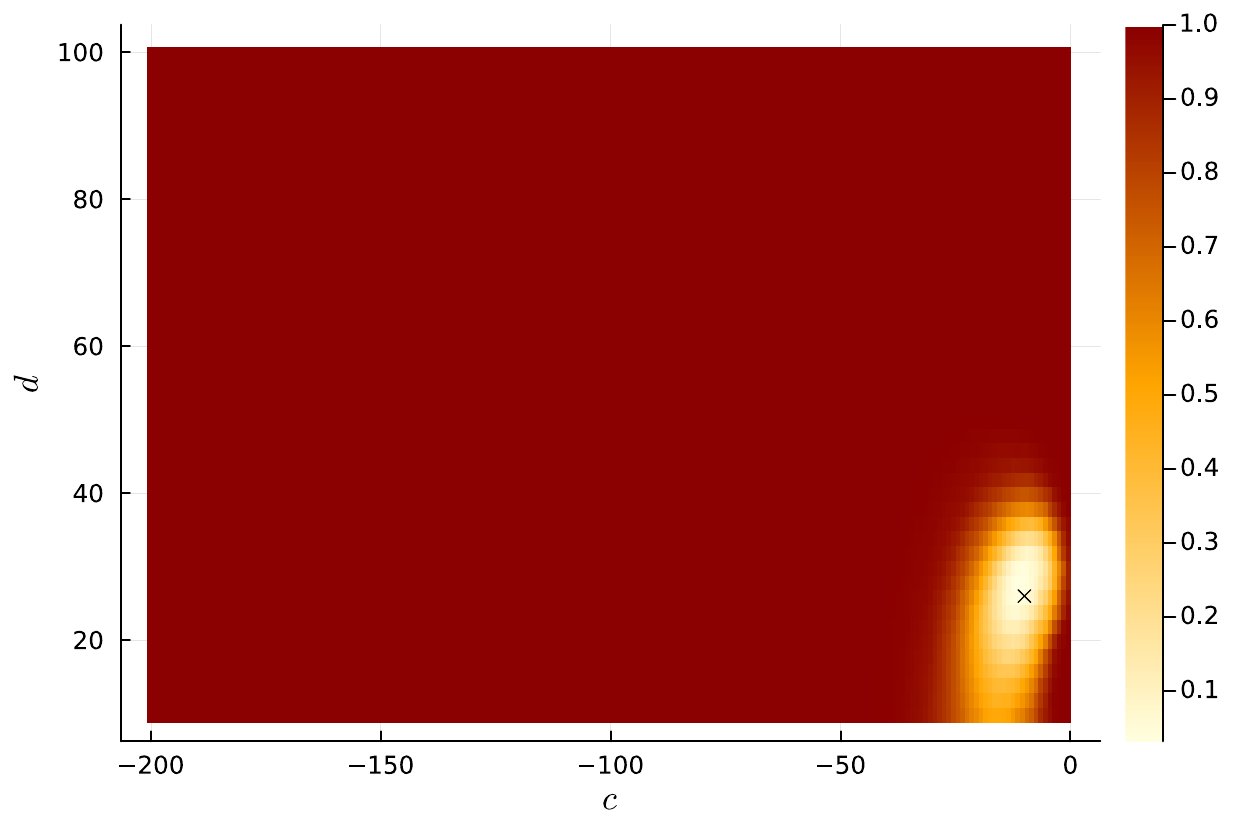}
    \caption{no trend, cycle, no serial correlation}
    \end{subfigure}
    \begin{subfigure}[b]{.49\textwidth}
    \centering
    \includegraphics[width=\textwidth]{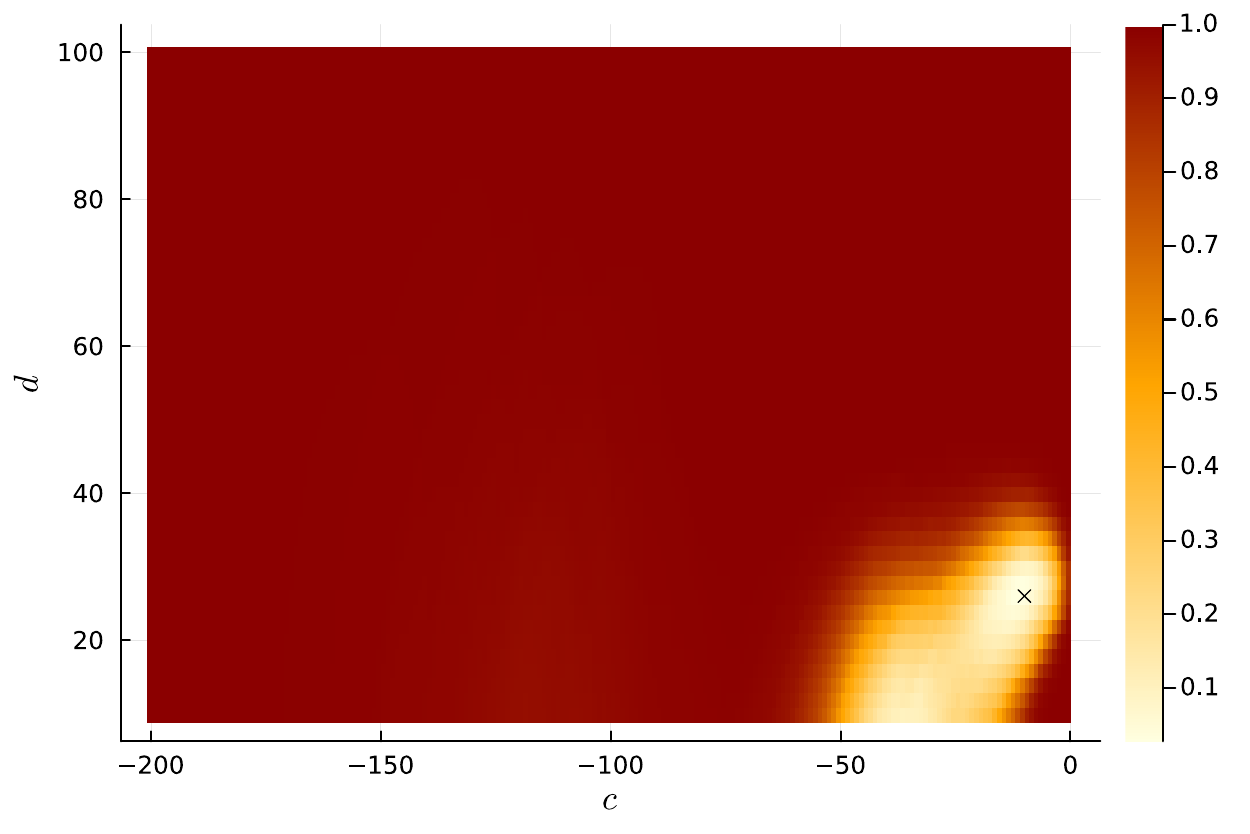}
    \caption{no trend, no cycle, serial correlation}
    \end{subfigure}
    \begin{subfigure}[b]{.49\textwidth}
    \centering
    \includegraphics[width=\textwidth]{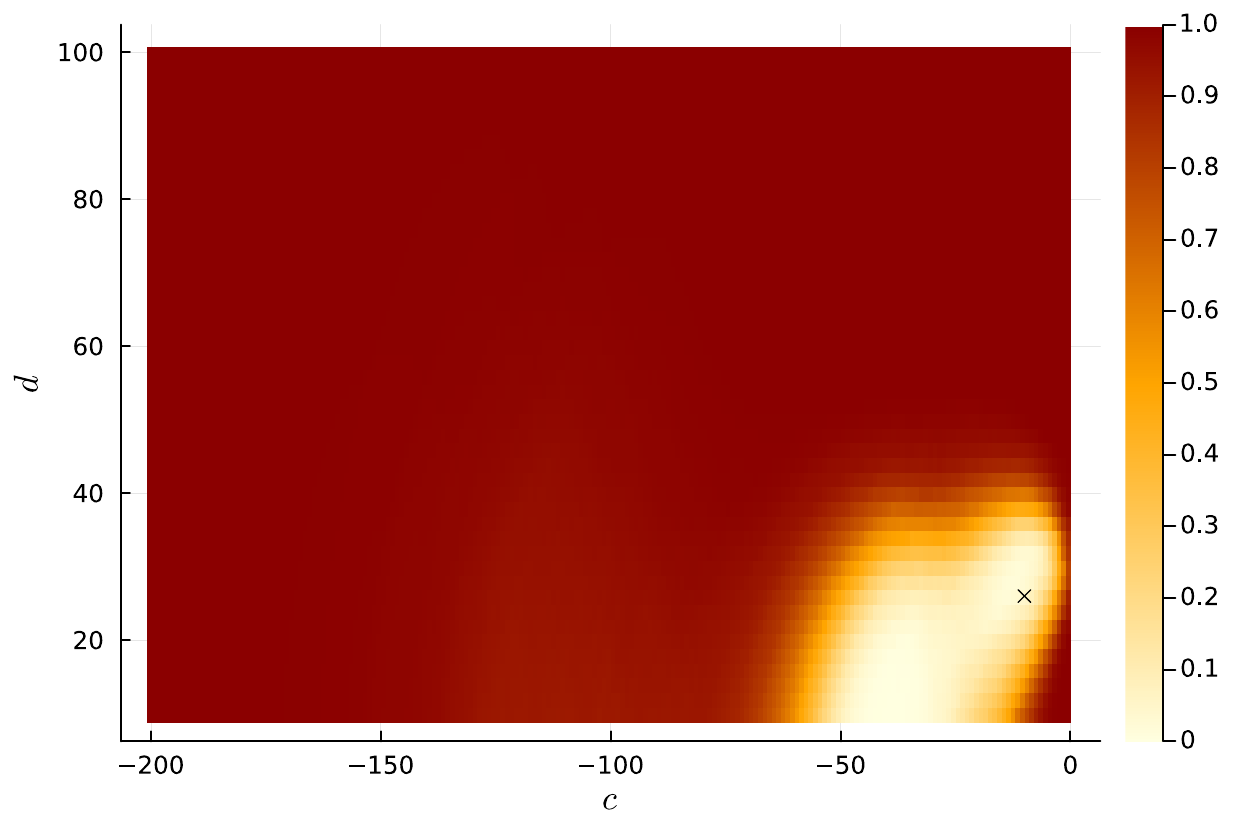}
    \caption{trend, cycle, serial correlation}
    \end{subfigure}
    \caption{Heatmaps of the simulated rejection probabilities of the nominal $0.05$-size test of $H_0:c=c_0,d=d_0$ vs. $H_1: c\ne c_0$ or $d\ne d_0$ for different specifications of the deterministic component (trend, cycle) and serial correlation in $\{u_t\}$, and different values of $c_0,d_0$. Data are generated with $c=-10,d=26$ (marked by \texttimes\ in the graphs). The tests are performed using the true specifications.   }
    \label{fig:Rej_10_26}
\end{figure}

Figure \ref{fig:Rej_10_26} reports the simulated rejection probabilities at the points of the $c_0,d_0$-grid and various specifications when data are generated with $c=-10,d=26$. The tests in the figure are conducted according to the true specifications. The power quickly reaches a rejection probability of one for sufficiently distant from the truth $c_0,d_0$ combinations.  Compared to the DGPs without serial correlation in $\{u_t\}$, there is some loss of power when $\{u_t\}$ are serially correlated. 

Figure \ref{fig:BICRej_10_26} reports the power of the tests performed using the specifications chosen by the BIC. The results are extremely close qualitatively and numerically to those where the true specifications were used. 

\begin{figure}[]
    \centering
    \begin{subfigure}[b]{.49\textwidth}
    \centering
    \includegraphics[width=\textwidth]{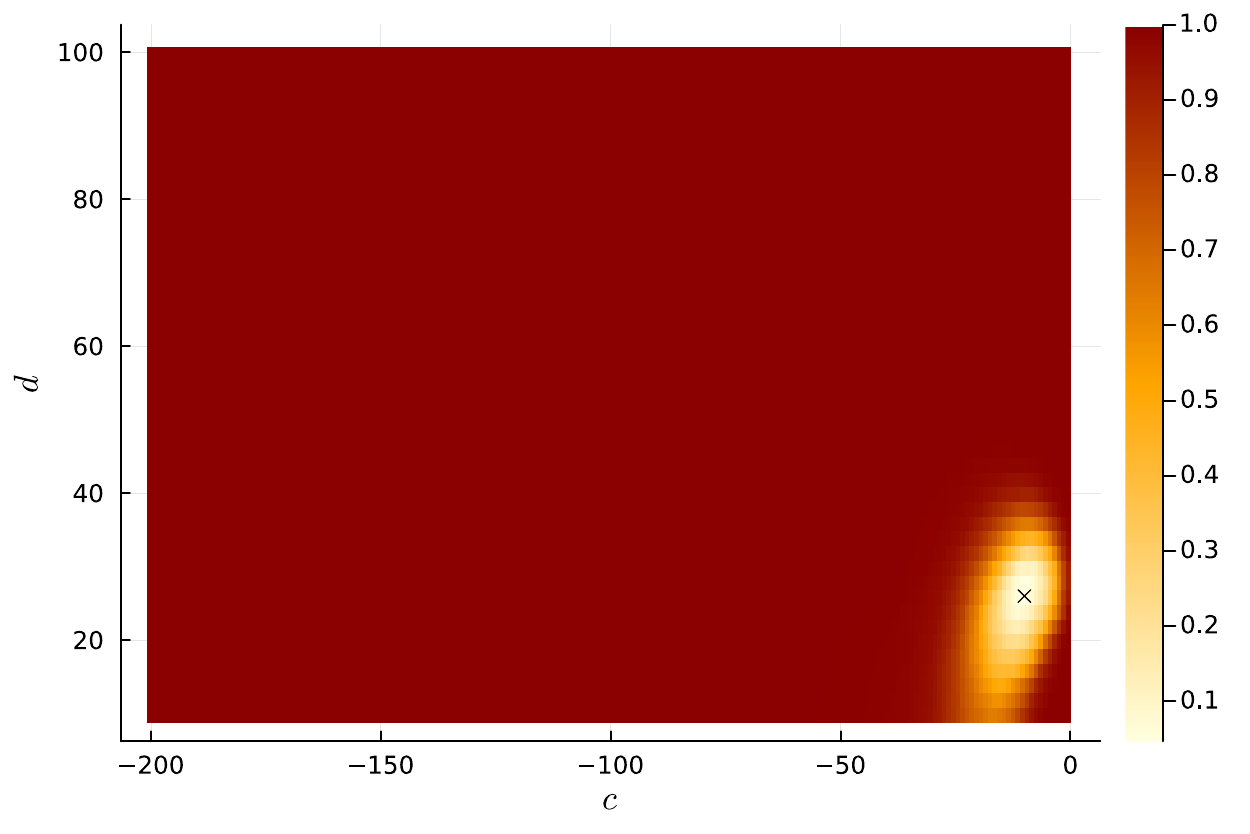}
    \caption{no trend, no cycle, no serial correlation}
    \end{subfigure}
    \hfill
    \begin{subfigure}[b]{.49\textwidth}
    \centering
    \includegraphics[width=\textwidth]{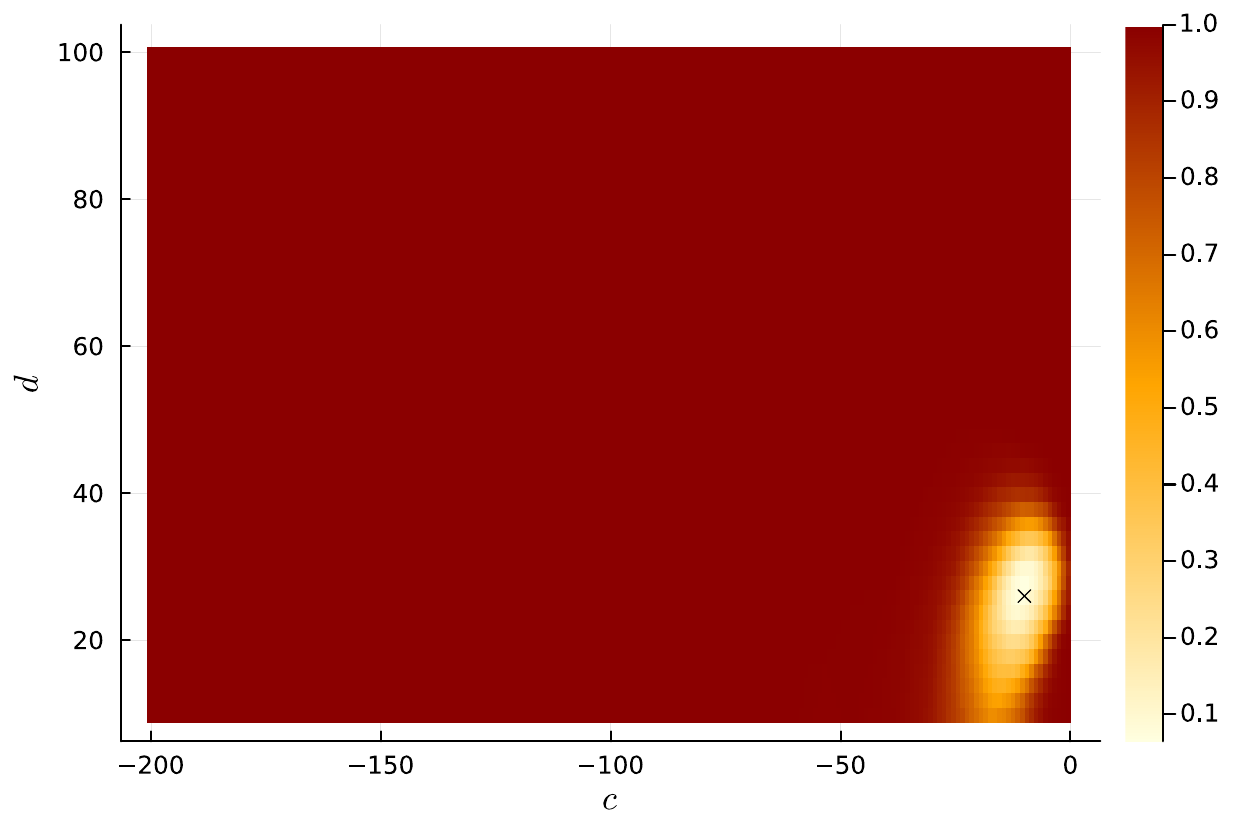}
    \caption{trend, no cycle, no serial correlation}
    \end{subfigure}
    \hfill
    \begin{subfigure}[b]{.49\textwidth}
    \centering
    \includegraphics[width=\textwidth]{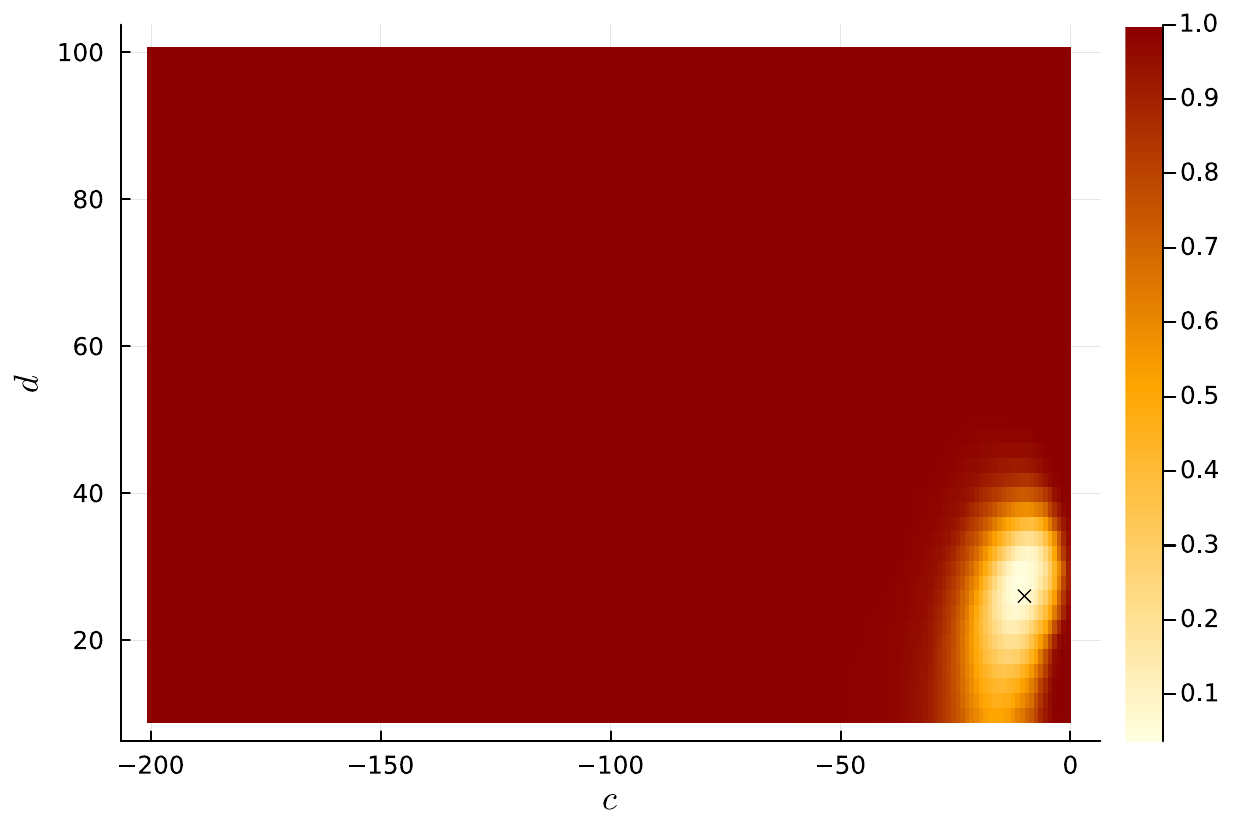}
    \caption{no trend, cycle, no serial correlation}
    \end{subfigure}
    \begin{subfigure}[b]{.49\textwidth}
    \centering
    \includegraphics[width=\textwidth]{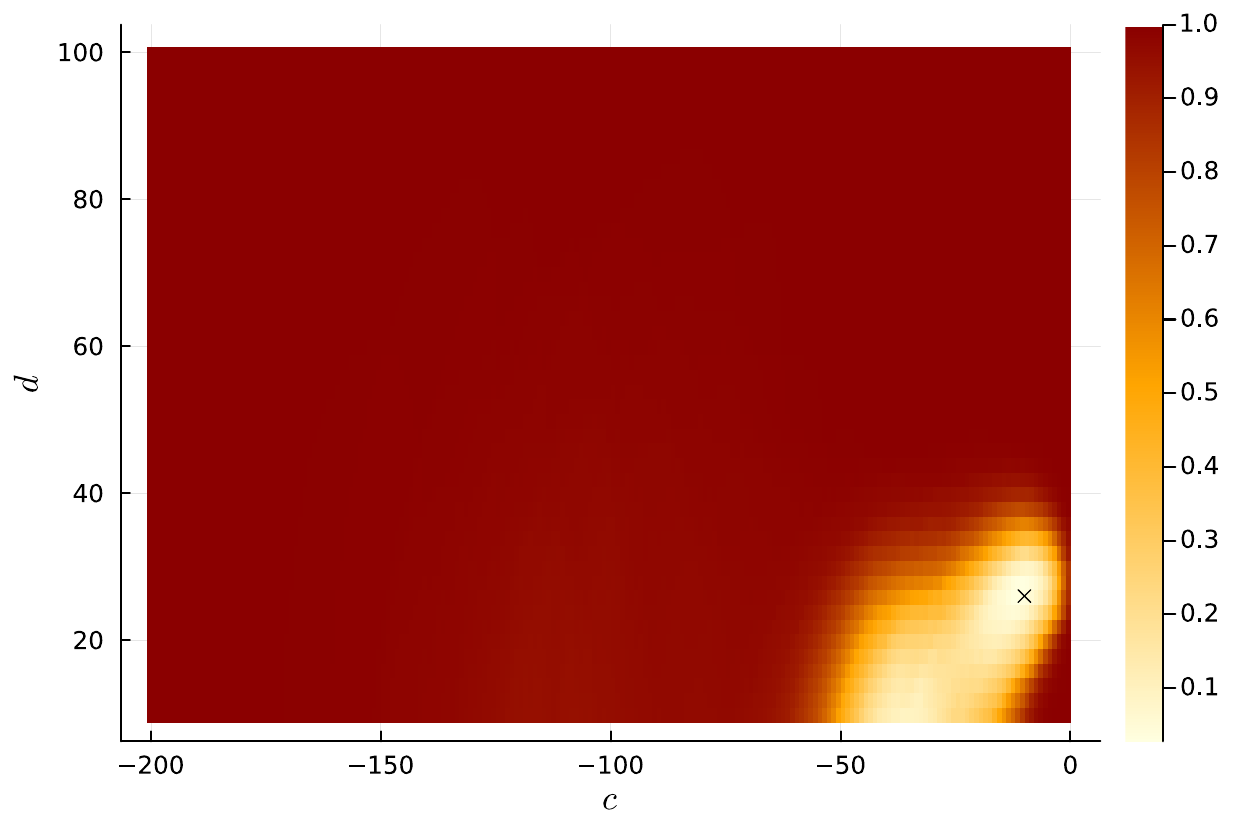}
    \caption{no trend, no cycle, serial correlation}
    \end{subfigure}
    \begin{subfigure}[b]{.49\textwidth}
    \centering
    \includegraphics[width=\textwidth]{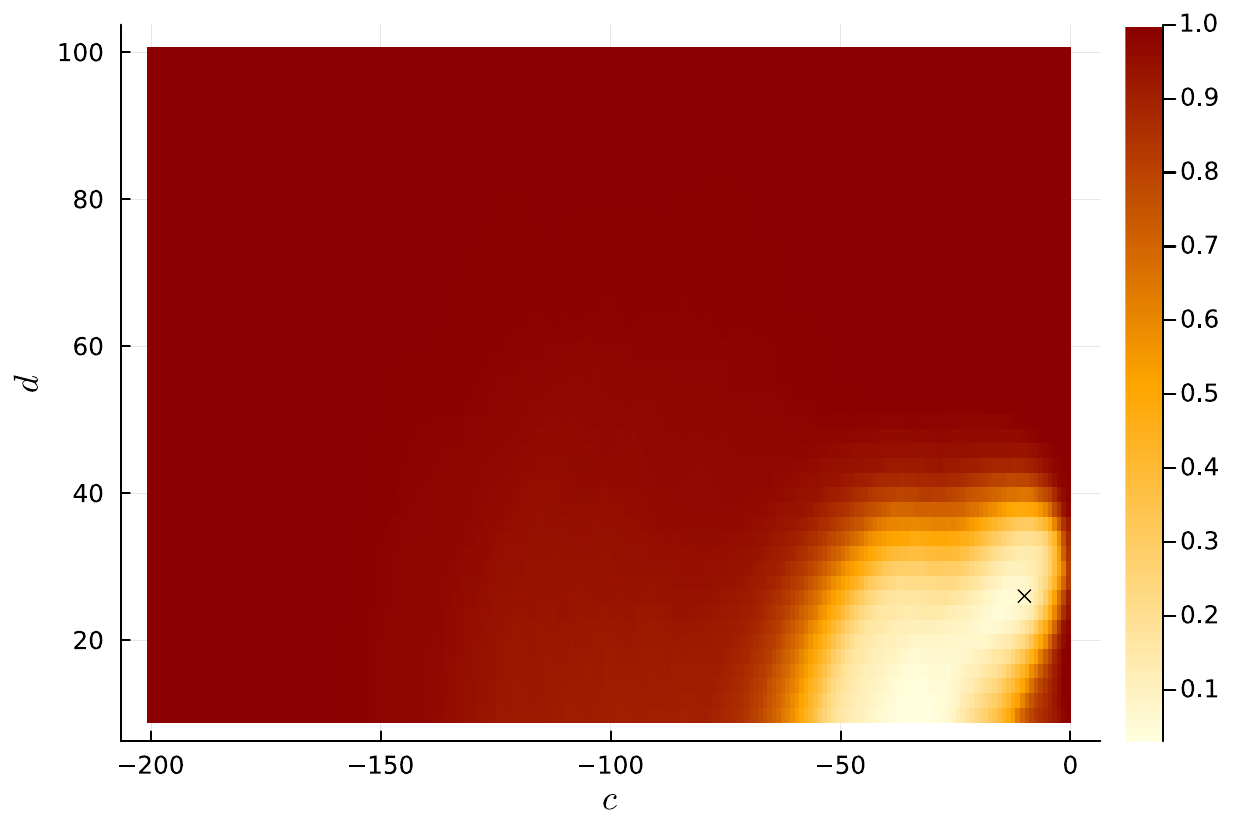}
    \caption{trend, cycle, serial correlation}
    \end{subfigure}
    \caption{Heatmaps of the simulated rejection probabilities of the nominal $0.05$-size test of $H_0:c=c_0,d=d_0$ vs. $H_1: c\ne c_0$ or $d\ne d_0$ for different specifications of the deterministic component (trend, cycle) and serial correlation in $\{u_t\}$, and different values of $c_0,d_0$. Data are generated with $c=-10,d=26$ (marked by \texttimes\ in the graphs). The tests are performed using the BIC-selected specifications.}
    \label{fig:BICRej_10_26}
\end{figure}

\begin{figure}[]
    \centering
    \begin{subfigure}[b]{.49\textwidth}
    \centering
    \includegraphics[width=\textwidth]{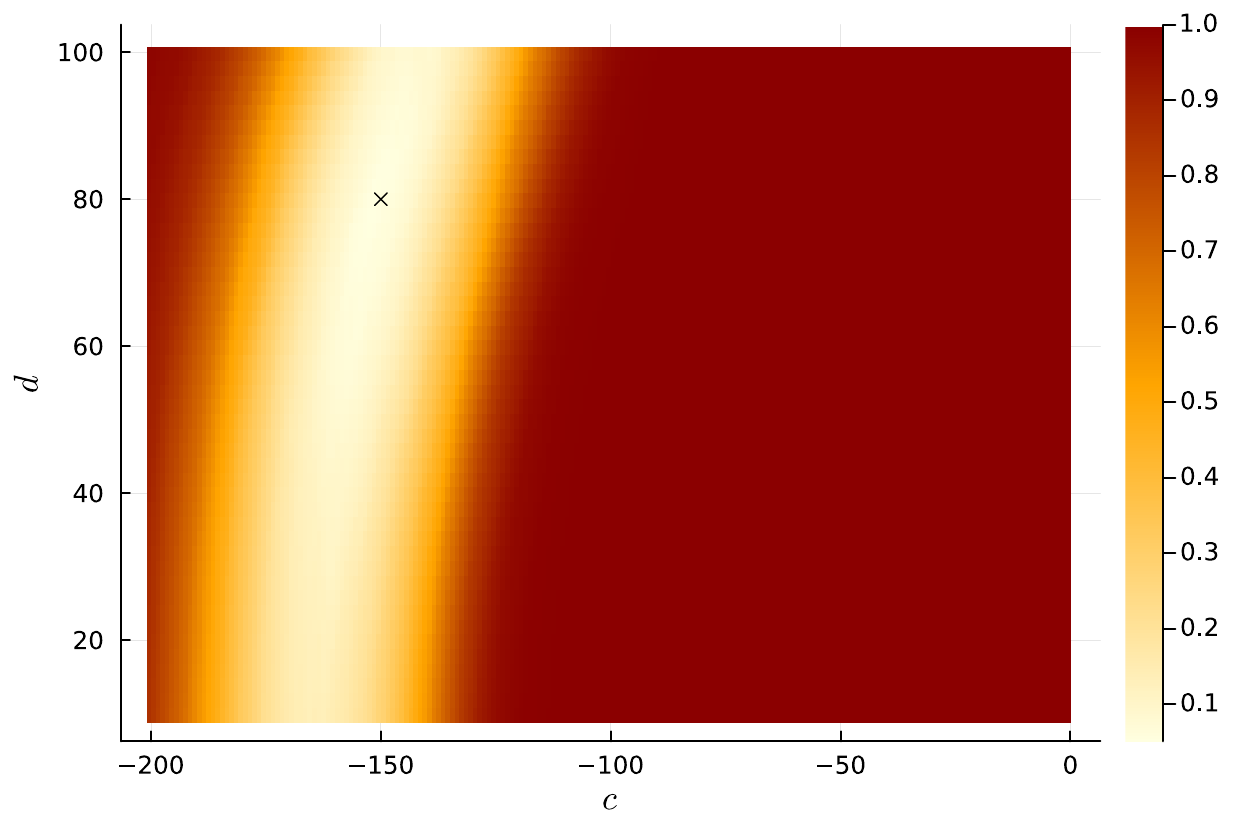}
    \caption{no trend, no cycle, no serial correlation}
    \end{subfigure}
    \hfill
    \begin{subfigure}[b]{.49\textwidth}
    \centering
    \includegraphics[width=\textwidth]{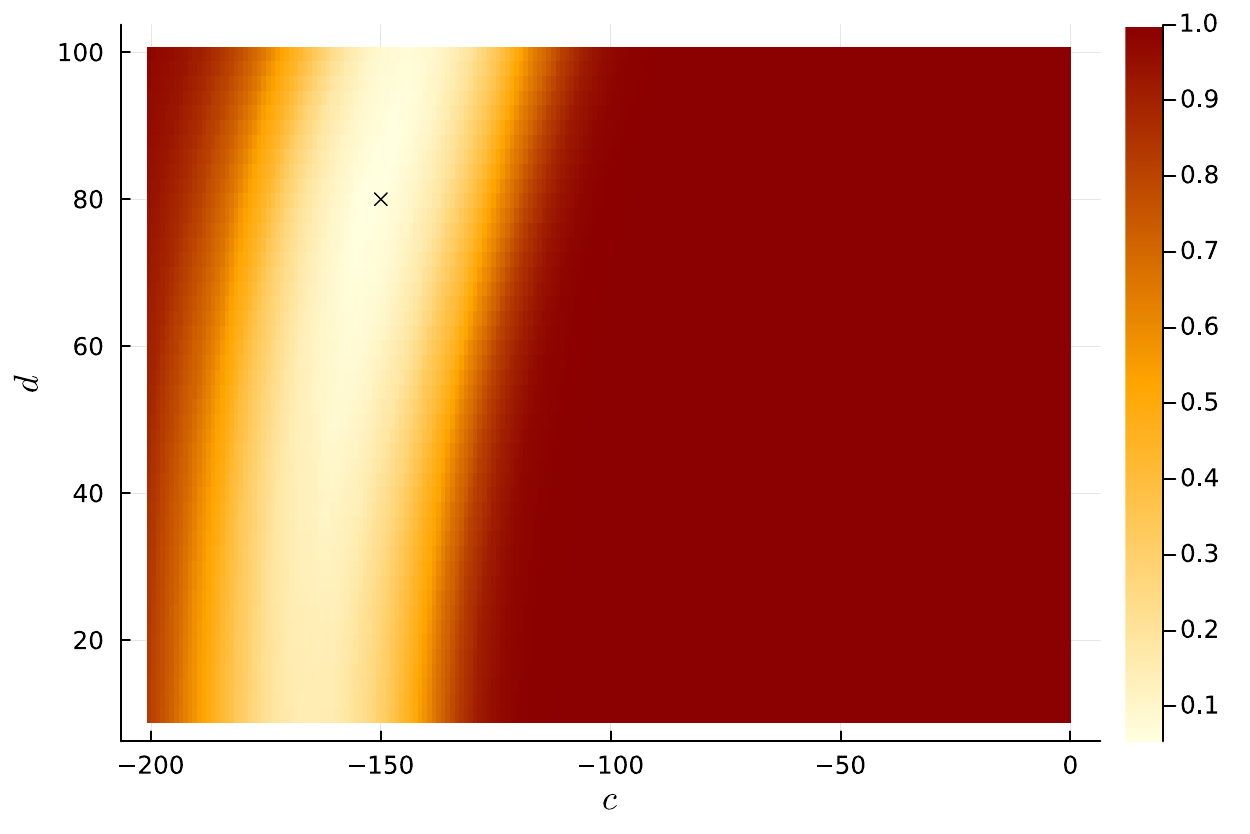}
    \caption{trend, no cycle, no serial correlation}
    \end{subfigure}
    \hfill
    \begin{subfigure}[b]{.49\textwidth}
    \centering
    \includegraphics[width=\textwidth]{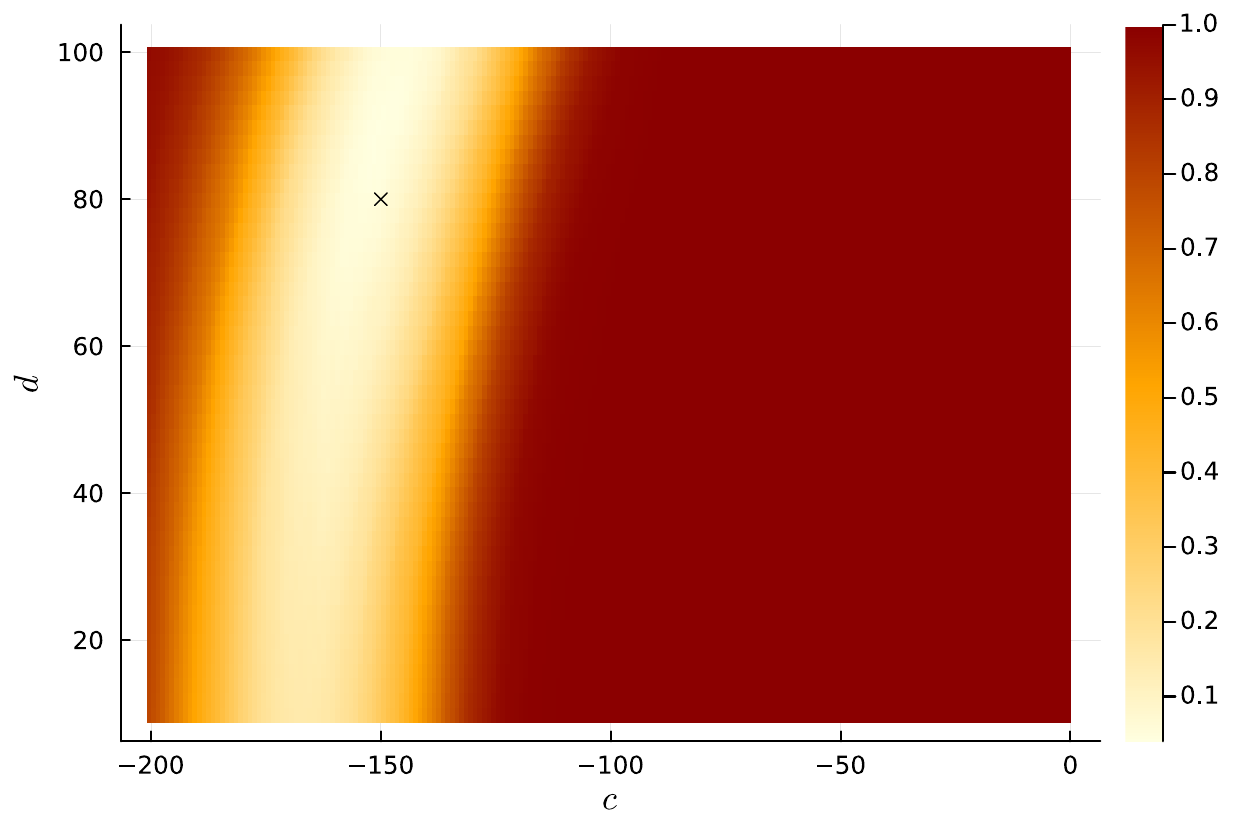}
    \caption{no trend, cycle, no serial correlation}
    \end{subfigure}
    \begin{subfigure}[b]{.49\textwidth}
    \centering
    \includegraphics[width=\textwidth]{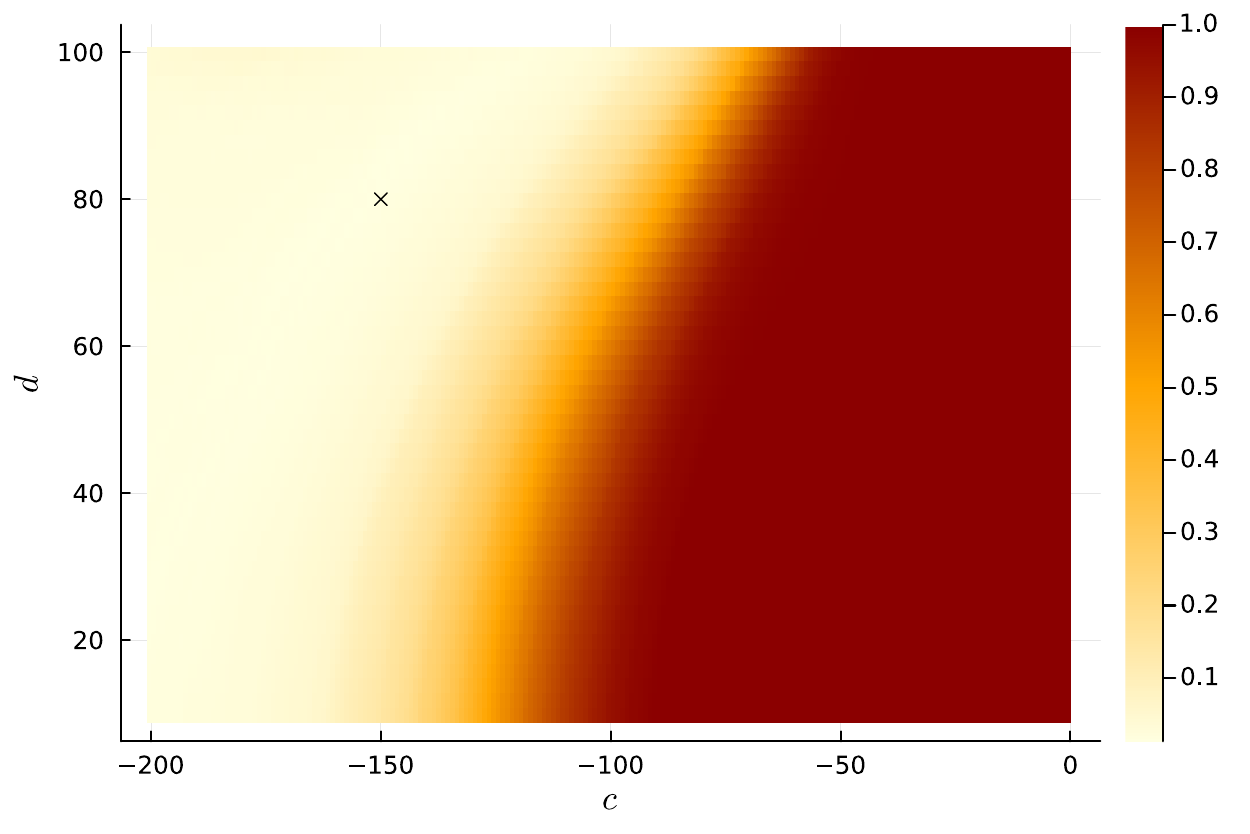}
    \caption{no trend, no cycle, serial correlation}
    \end{subfigure}
    \begin{subfigure}[b]{.49\textwidth}
    \centering
    \includegraphics[width=\textwidth]{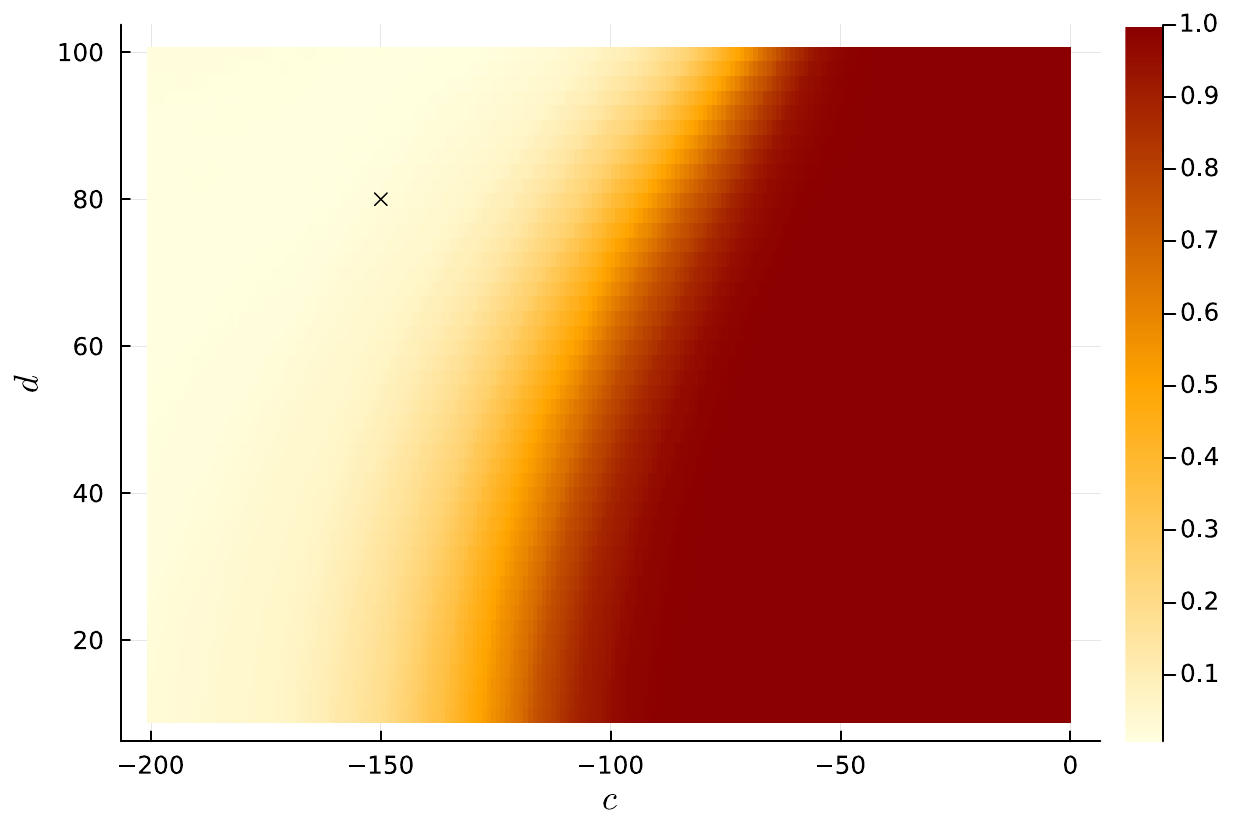}
    \caption{trend, cycle, serial correlation}
    \end{subfigure}
    \caption{Heatmaps of the simulated rejection probabilities of the nominal $0.05$-size test of $H_0:c=c_0,d=d_0$ vs. $H_1: c\ne c_0$ or $d\ne d_0$ for different specifications of the deterministic component (trend, cycle) and serial correlation in $\{u_t\}$, and different values of $c_0,d_0$. Data are generated with $c=-150,d=80$ (marked by \texttimes\ in the graphs). The tests are performed using the true specifications.}
    \label{fig:Rej_150_80}
\end{figure}

\begin{figure}[]
    \centering
    \begin{subfigure}[b]{.49\textwidth}
    \centering
    \includegraphics[width=\textwidth]{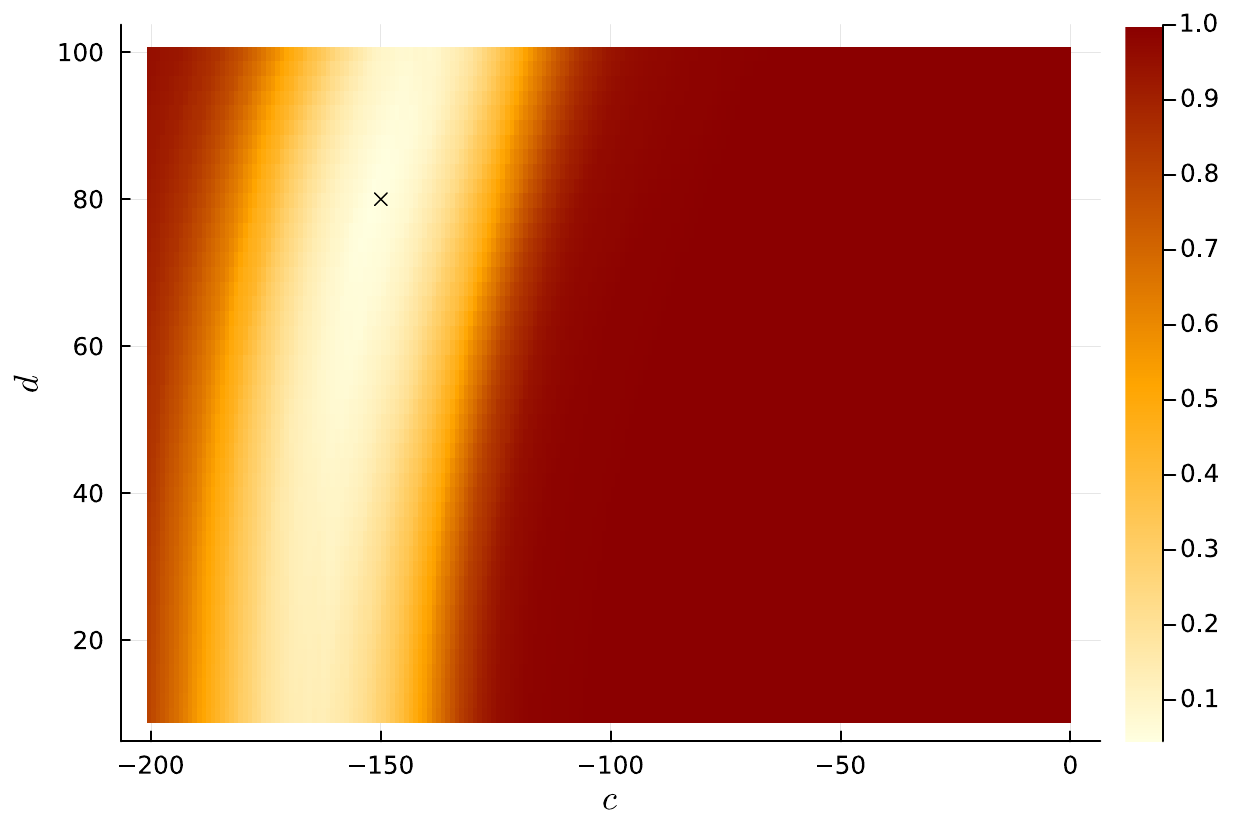}
    \caption{no trend, no cycle, no serial correlation}
    \end{subfigure}
    \hfill
    \begin{subfigure}[b]{.49\textwidth}
    \centering
    \includegraphics[width=\textwidth]{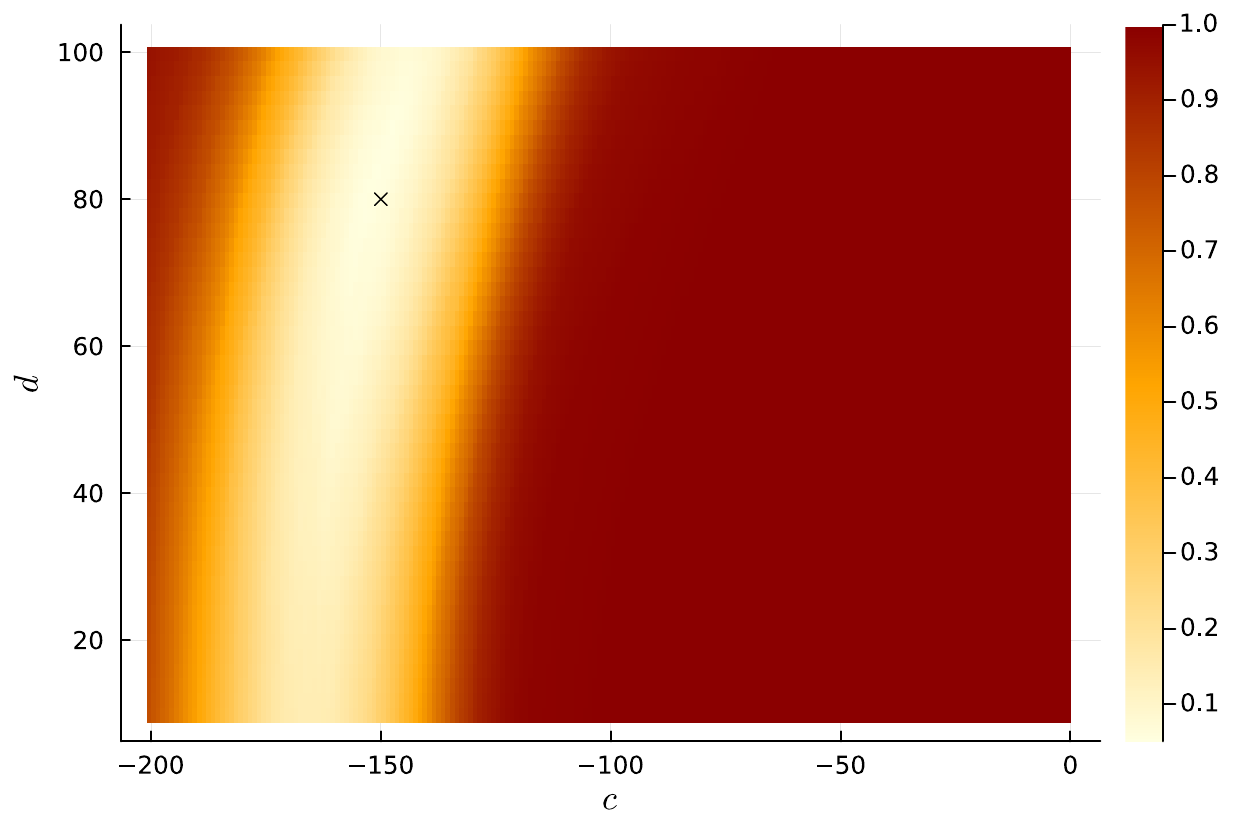}
    \caption{trend, no cycle, no serial correlation}
    \end{subfigure}
    \hfill
    \begin{subfigure}[b]{.49\textwidth}
    \centering
    \includegraphics[width=\textwidth]{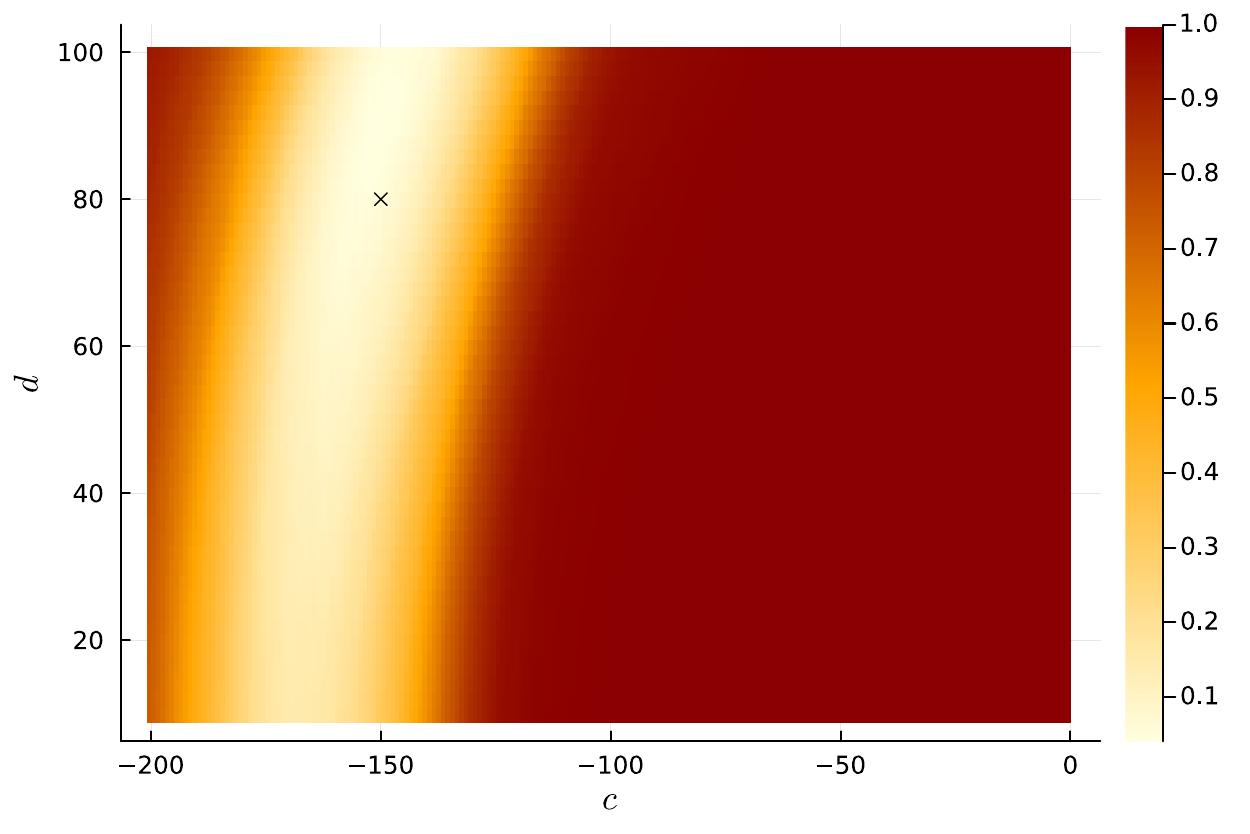}
    \caption{no trend, cycle, no serial correlation}
    \end{subfigure}
    \begin{subfigure}[b]{.49\textwidth}
    \centering
    \includegraphics[width=\textwidth]{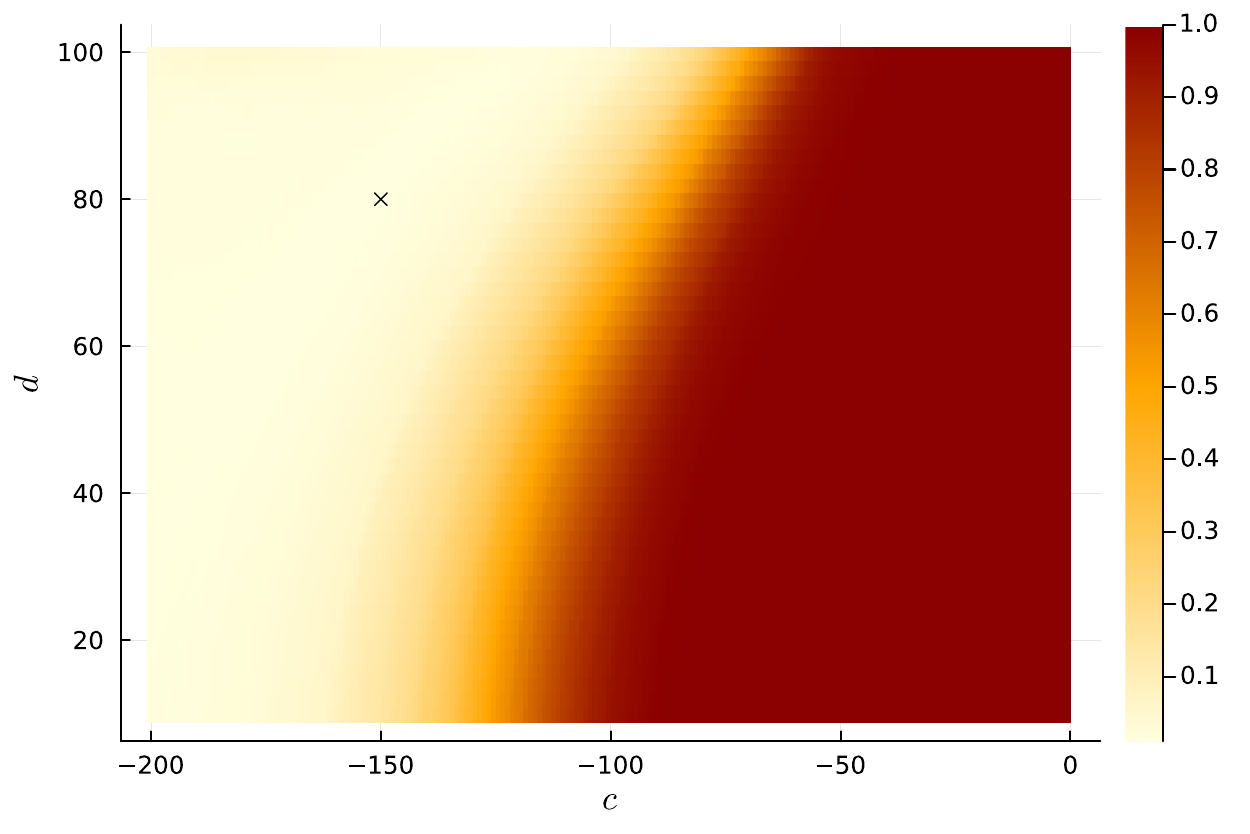}
    \caption{no trend, no cycle, serial correlation}
    \end{subfigure}
    \begin{subfigure}[b]{.49\textwidth}
    \centering
    \includegraphics[width=\textwidth]{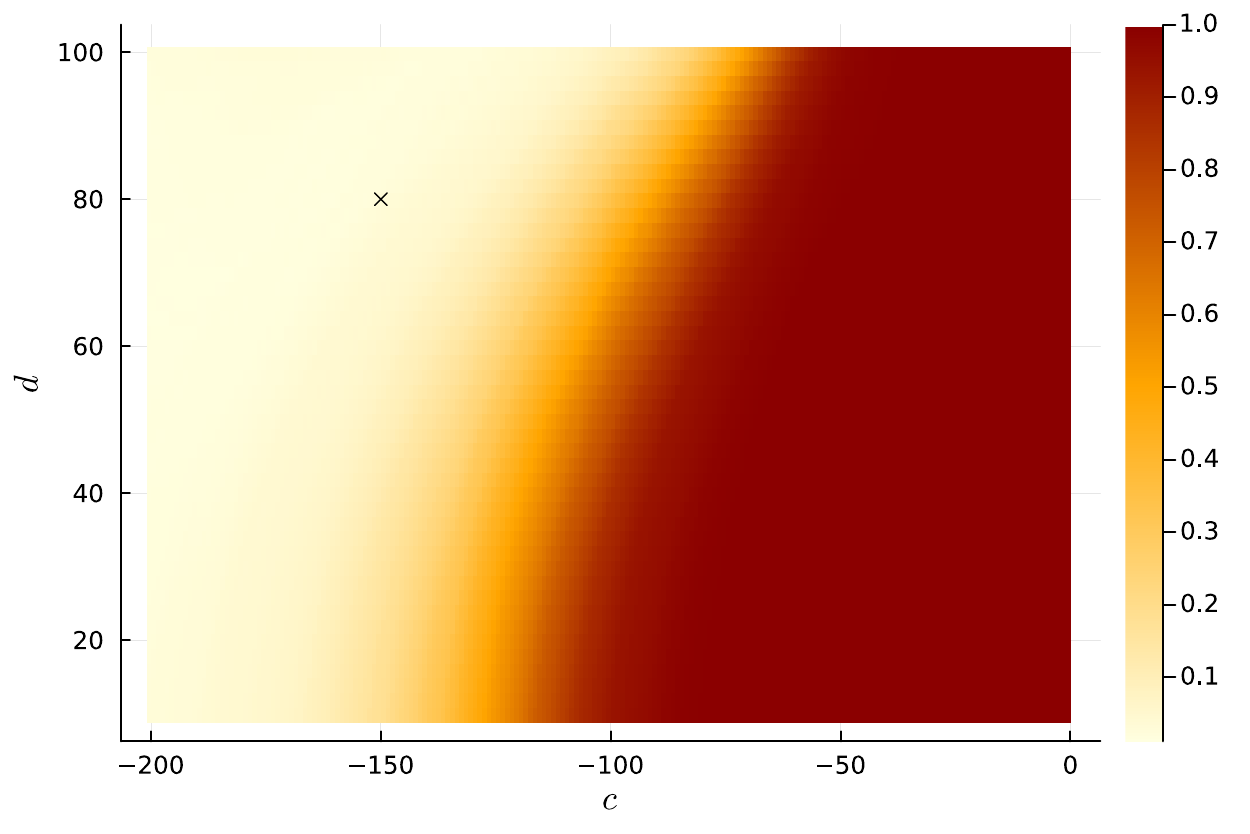}
    \caption{trend, cycle, serial correlation}
    \end{subfigure}
    \caption{Heatmaps of the simulated rejection probabilities of the nominal $0.05$-size test of $H_0:c=c_0,d=d_0$ vs. $H_1: c\ne c_0$ or $d\ne d_0$ for different specifications of the deterministic component (trend, cycle) and serial correlation in $\{u_t\}$, and different values of $c_0,d_0$. Data are generated with $c=-150,d=80$ (marked by \texttimes\ in the graphs). The tests are performed using the BIC-selected specifications.}
    \label{fig:BICRej_150_80}
\end{figure}

Figure \ref{fig:Rej_150_80} reports the corresponding results for DGPs with $c=-150,d=80$, where the tests are performed using the true specifications. The tests are substantially less powerful than those when $c=-10,d=26$. This can be explained by the fact that under $c=-150,d=80$, the DGPs are near I$(0)$ and the estimators of the $\phi$'s lose their superconsistency properties. One can see that the deviations of $c_0$ from the true $c$ matter more for power than the deviations of $d_0$ from $d$: the test has little power to detect incorrect $d$'s when the true $c$ is far from zero. However, note that, as we discuss in Section \ref{section:model}, for such values of $c$, data are unlikely to display any oscillating behavior. The power increases faster with distance between $c_0$ and $c$ and reaches a rejection probability of one for sufficiently distant values of $c_0$. Similarly to the more persistent specifications, serial correlation in $\{u_t\}$ reduces the power of the test. 

Figure \ref{fig:BICRej_150_80} reports the corresponding results for the tests with the BIC-selected specifications for the deterministic component and the serial correlation in $\{u_t\}$. Again, one can see that the results are numerically and qualitatively close to those under the true specifications. 

From Table \ref{tab:size}, and the comparisons between Figures \ref{fig:Rej_10_26} and \ref{fig:BICRej_10_26}, and \ref{fig:Rej_150_80} and \ref{fig:BICRej_150_80}, we conclude that the BIC is capable of consistently selecting the correct specifications without affecting the size or power properties of our tests.


\section{Proofs of the results in Appendix \ref{sec:BIC}}

\begin{proof}[Proof of Lemma \ref{lem:S.second_differences}] For part \ref{en:S.Delta^2}, using the expansion in \eqref{eq:Delta_expansion} in Proposition \ref{prop:convergence_G},
\begin{align*}
    \Delta^2 y_t & = \Delta y_t-\Delta y_{t-1} \\
    & =  u_t+ \frac{2c}{n}\sum_{j=0}^{t-1}\bigg(1+\frac{2c}{n}\bigg)^{t-1-j}u_j +\bigg(\frac{c^2-d^2}{n^2} + O(n^{-3})\bigg)\bigg( y_{t-1} + \frac{2c}{n}\sum_{j=0}^{t-1}\bigg(1+\frac{2c}{n}\bigg)^{t-1-j}y_{j-1}\bigg) \\
    &\quad - \bigg( \frac{2c^2}{n^2} + O(n^{-3})\bigg) \bigg( y_{t-2} +\frac{2c}{n}\sum_{j=0}^{t-1}\bigg(1+\frac{2c}{n}\bigg)^{t-1-j}y_{j-2}\bigg)\\
    &= u_t+\bigg(\frac{c^2-d^2}{n^2} + O(n^{-3})\bigg) y_{t-1} - \bigg( \frac{2c^2}{n^2} + O(n^{-3})\bigg)  y_{t-2}+\frac{2c}{n}\Delta y_{t-1} \\
    &= u_t + O_p(n^{-1/2}),
\end{align*}
where the second and third equalities hold by the expansion in \eqref{eq:Delta_expansion} , and the equality in the last line holds by Propositions \ref{prop:convergence_J} and \ref{prop:convergence_G}. 
The remaining lemma results follow those in Part \ref{en:S.Delta^2}, Lemma \ref{lem:convergence_moments}, and because $\{u_t\}$ is an I($0$) and stationary AR($p$) process.

\end{proof}

\begin{proof}[Proof of Proposition \ref{prop:S.BIC}]
In part (a), suppose that $\kappa\ne\emptyset$ and $p>0$.
\[BIC_n(1,\kappa,p)-\sum\varepsilon_t^2 = - v_n' A_n^{-1} v_n\] 
where
\begin{align*}
   v_n &\equiv \left(
        \sum \varepsilon_t ,
        \sum \varepsilon_t \cdot (t/n) ,
        \sum \varepsilon_t C_\kappa(t/n)' ,
        \sum \varepsilon_t y_{t-1} ,
        \sum \varepsilon_t \Delta y_{t-1},
          \sum \varepsilon_t z_{t-1,p}' 
   \right)',
\end{align*}
and the $A_n$ matrix is given by
\[
    \begin{pmatrix}
        n & \sum t/n & \sum C_\kappa(t/n)' & \sum y_{t-1} & \sum \Delta y_{t-1} & \sum z_{t-1,p}' \\
        & \sum (t/n)^2 & \sum  (t/n)C_\kappa (t/n)'   & \sum (t/n) y_{t-1} & \sum (t/n)\Delta y_{t-1} & \sum (t/n) z_{t-1,p} '\\ 
        & & \sum C_\kappa (t/n)C_\kappa (t/n)'  & \sum C_\kappa (t/n) y_{t-1} & \sum C_\kappa (t/n) \Delta y_{t-1} & \sum C_\kappa (t/n) z_{t-1,p}'\\
        &  & &  \sum y_{t-1}^2 & \sum y_{t-1}\Delta y_{t-1} &  \sum y_{t-1} z_{t-1,p}'  \\
        && & &  \sum (\Delta y_{t-1})^2 & \sum \Delta y_{t-1} z_{t-1,p}'   \\
        & &&  & & \sum z_{t-1,p} z_{t-1,p}'   
    \end{pmatrix}.
\] 
Define a diagonal scaling matrix
\begin{align*}
    H_n = \begin{pmatrix}
        n^{-1/2} & 0 &&& \ldots & 0 \\
        & n^{-1/2} & 0 &&\ldots & 0\\
        & & n^{-1/2} I_{|\kappa|} &0&\ldots&  0\\
        & & &  n^{-2} &0  &0\\
        & & & &  n^{-1}  &0\\
        & & & & &  n^{-1}I_p
    \end{pmatrix}.
\end{align*}
By Lemmas \ref{lem:convergence_moments} and \ref{lem:S.second_differences},
\begin{align*}
    H_n v_n &=\Big(
       n^{-1/2} \sum \varepsilon_t,
        n^{-1/2} \sum \varepsilon_t (t/n),
       n^{-1/2}  \sum \varepsilon_t C_\kappa(t/n) ' ,\\
      &\qquad n^{-2} \sum \varepsilon_t y_{t-1},
       n^{-1} \sum \varepsilon_t \Delta y_{t-1},
       n^{-1/2}\sum \varepsilon_t z_{t-1,p}'
    \Bigg)'\\
    &=O_p(1).
\end{align*}
By the same arguments, $H_n A_n H_n $ is given by 
\begin{align*}
    &\begin{pmatrix}
        1 & \frac{\sum t/n}{n} & \frac{\sum C_\kappa(t/n)'}{n}& \frac{\sum y_{t-1}}{n^{5/2}} & \frac{\sum \Delta y_{t-1}}{n^{3/2}} & \frac{\sum z_{t-1,p}'}{n}  \\
        & \frac{\sum (t/n)^2}{n} & \frac{\sum  (t/n)C_\kappa (t/n)'}{n} & \frac{\sum (t/n) y_{t-1}}{n^{5/2}} & \frac{\sum (t/n)\Delta y_{t-1}}{n^{3/2}}  & \frac{\sum (t/n) z_{t-1,p} '}{n}  \\ 
        & & \frac{\sum C_\kappa (t/n)C_\kappa (t/n)'}{n} & \frac{\sum C_\kappa (t/n) y_{t-1}}{n^{5/2}} & \frac{\sum C_\kappa (t/n) \Delta y_{t-1}}{n^{3/2}} & \frac{\sum C_\kappa (t/n) z_{t-1,p}'}{n} \\
        &  & &  \frac{\sum y_{t-1}^2}{n^4} & \frac{\sum y_{t-1}\Delta y_{t-1}}{n^3} &  \frac{\sum y_{t-1} z_{t-1,p}'}{n^{5/2}}  \\
        && & &  \frac{\sum (\Delta y_{t-1})^2}{n^{2}} & \frac{\sum \Delta y_{t-1} z_{t-1,p}'}{n^{3/2}}   \\
        & &&  & & \frac{\sum z_{t-1,p} z_{t-1,p}'}{n}\\
    \end{pmatrix}\\
    &\to_d\begin{pmatrix}
        1 & \int_0^1 s \dd s & \int_0^1 C_\kappa(s)' \dd s& \sigma \int J_{c,d} & \sigma \int G_{c,d}  & 0 \\
        & \int_0^1 s^2 \dd s & \int_0^1 s C_\kappa(s)' \dd s & \sigma \int_0^1 s J_{c,d}(s) \dd s & \sigma \int_0^1 s G_{c,d}(s) \dd s & 0 \\
        & & \int_0^1 C_\kappa(s)C_\kappa(s)' \dd s &  \sigma\int_0^1 C_\kappa(s) J_{c,d}(s) \dd s & \sigma \int_0^1 C_\kappa(s) G_{c,d}(s) \dd s & 0\\
        & & &  \sigma^2 \int J_{c,d}^2 & \sigma^2 \int J_{c,d} G_{c,d} & 0\\
& & & & \sigma^2\int G_{c,d}^2 & 0\\
        & & & & & \Gamma_p
    \end{pmatrix},
\end{align*}
where $\sigma^2=\sigma^2_\varepsilon/(1-\rho_1 -\ldots -\rho_{p_0})^2 $ is the long-run variance of $\{u_t\}$, and $\Gamma_p$ is the $p\times p$ matrix of the variances and autocovariances of orders up to $p$ of $\{u_t\}$:
$
\Gamma_p\equiv E z_{t-1,p} z_{t-1,p}'
$.
Therefore, $(H_n A_n H_n)^{-1}=O_p(1)$, and the result in part (a) follows for the case of $T=1$, $\kappa\ne\emptyset$, and $p>0$. The cases where $T_0=T=0$, $\kappa_0=\kappa=\emptyset$, and $p=p_0=0$ can be handled similarly.

For part (b), suppose that the terms in $\kappa^*\subset \kappa_0$ are omitted from $\kappa$. Let $C_{\kappa^*}(t/n)$ denote the vector of the corresponding deterministic cyclical components. Furthermore, suppose that $p_0>p$, and let $z_{t-1,p}^*=(\Delta^2 y_{t-p-1},\ldots,\Delta^2 y_{t-p_0})'$ denote the vector of $\Delta^2 y_t$ terms omitted from $z_{t-1,p}$. Let $\gamma_n^*$ and $\nu_n^*$ denote the coefficients of $C_{\kappa^*}(t/n)$ and $z_{t-1,p}^*$ respectively. Now,
\begin{align*}
    n^{-1}SSR_n(1,\kappa,p)=n^{-1}\sum (\varepsilon_t + C_{\kappa^*}(t/n)'\gamma_n^*+{z^*}'_{t-1,p}\nu_n^*)^2-n^{-1}{v_n^*}'A_n^{-1}v_n,
\end{align*}
where $A_n$ is as defined in the proof of part (a), and 
\begin{align*}
     v_n^*&\equiv \begin{pmatrix}
        \sum \varepsilon_t +  {\gamma_n^*}'\sum C_{\kappa^*}(t/n) +{\nu_n^*}'\sum{z^*}_{t-1,p} \\
        \sum \varepsilon_t \cdot (t/n)  + {\gamma_n^*}'\sum C_{\kappa^*}(t/n)\cdot(t/n)+{\nu^*}'_n\sum{z^*}_{t-1,p}\cdot (t/n)\\
        \sum \varepsilon_t C_\kappa(t/n)' + {\gamma_n^*}'\sum C_{\kappa^*}(t/n) C_{\kappa}(t/n)+{\nu_n^*}'\sum{z_{t-1,p}^*} C_{\kappa}(t/n)\\
        \sum \varepsilon_t y_{t-1}  + {\gamma_n^*}'\sum C_{\kappa^*}(t/n) y_{t-1}+{\nu_n^*}'\sum{z_{t-1,p}^*} y_{t-1}\\
         \sum \varepsilon_t \Delta y_{t-1}   + {\gamma_n^*}'\sum C_{\kappa^*}(t/n) \Delta y_{t-1}+{\nu_n^*}'\sum{z_{t-1,p}^*} \Delta y_{t-1}\\
         \sum \varepsilon_t z_{t-1,p} + {\gamma_n^*}'\sum C_{\kappa^*}(t/n) z_{t-1,p}+{\nu_n^*}'\sum{z_{t-1,p}^*}z_{t-1,p}\\
     \end{pmatrix}.
\end{align*}
Next,
\begin{align*}
    n^{-1/2}H_n v_n^*& = \begin{pmatrix}
        n^{-1}\sum \varepsilon_t +  {\gamma_n^*}' n^{-1}\sum C_{\kappa^*}(t/n) +{\nu_n^*}' n^{-1}\sum{z^*}_{t-1,p}  \\
       n^{-1} \sum \varepsilon_t \cdot (t/n)  + {\gamma_n^*}' n^{-1}\sum C_{\kappa^*}(t/n)\cdot(t/n)+{\nu^*}'_n n^{-1}\sum{z^*}_{t-1,p}\cdot (t/n) \\
        n^{-1}\sum \varepsilon_t C_\kappa(t/n) + {\gamma_n^*}'n^{-1}\sum C_{\kappa^*}(t/n) C_{\kappa}(t/n)+{\nu_n^*}'n^{-1}\sum{z_{t-1,p}^*} C_{\kappa}(t/n)\\
        n^{-5/2}\sum \varepsilon_t y_{t-1}  + {\gamma_n^*}'n^{-5/2}\sum C_{\kappa^*}(t/n) y_{t-1}+{\nu_n^*}'n^{-5/2}\sum{z_{t-1,p}^*} y_{t-1}\\
         n^{-3/2}\sum \varepsilon_t \Delta y_{t-1}   + {\gamma_n^*}'n^{-3/2}\sum C_{\kappa^*}(t/n) \Delta y_{t-1}+{\nu_n^*}'n^{-3/2}\sum{z_{t-1,p}^*} \Delta y_{t-1}\\
         n^{-1}\sum \varepsilon_t z_{t-1,p} + {\gamma_n^*}'n^{-1}\sum C_{\kappa^*}(t/n) z_{t-1,p}+{\nu_n^*}'n^{-1}\sum{z_{t-1,p}^*}z_{t-1,p}\\
     \end{pmatrix}\notag\\
     &\to_d \begin{pmatrix}
     \int_0^1 {\gamma^*}'C_{\kappa^*}(s)\dd s \\
        \int_0^1 s{\gamma^*}'C_{\kappa^*}(s)\dd s  \\
         \int_0^1 {\gamma^*}'C_{\kappa^*}(s)C_{\kappa}(s)\dd s \\
          \sigma \int_0^1 {\gamma^*}'C_{\kappa^*}(s) J_{s,d}(s) \dd s\\
          \sigma \int_0^1 {\gamma^*}'C_{\kappa^*}(s) G_{s,d}(s) \dd s\\
          \Gamma_{p,*}\nu^*
     \end{pmatrix},\label{eq:S.v_star}
\end{align*}
where $\gamma^*\equiv\lim_{n\to\infty}\gamma_n=(1-\rho_1-\ldots-\rho_p)(\eta_{1k},\eta_{2k}:k\in\kappa^*)'$, $\Gamma_{p,*}\equiv Ez_{t-1,p}{z^*}_{t-1,p}'$, and $\nu^*\equiv \lim_{n\to\infty} \nu^*_n $. It follows that
\begin{align*}
    n^{-1}SSR_n(1,\kappa,p)\to_d \sigma^2_\varepsilon 
    +\int_0^1 \Big({\gamma^*}'C_{\kappa^*}(s) - X(s)'\varsigma\Big)^2\dd s+{\nu^*}'\Big(\Gamma_* -\Gamma_{p,*}' \Gamma_p^{-1} \Gamma_{p,*}\Big){\nu^*},
\end{align*}
where 
\begin{align*}
    X(s)&\equiv \Big(1,s,C_\kappa(s)',J_{c,d}(s),G_{c,d}(s)\Big)',\\
    \varsigma&\equiv \Big(\int_0^1 X(s)X(s)' \dd s\Big)^{-1} \int_0^1X(s) {\gamma^*}'C_{\kappa^*}(s) \dd s,\\
    \Gamma_*&\equiv Ez_{t-1,p}^*{z^*}_{t-1,p}'.
\end{align*}

Similar arguments can be used to show that when $T_0=1$,
\begin{align*}
    n^{-1}SSR_n(0,\kappa,p)\to_d \sigma^2_\varepsilon 
    +\int_0^1 \Big({\gamma^*}'C_{\kappa^*}(s)+\beta \cdot s - X(s)'\varsigma\Big)^2\dd s+{\nu^*}'\Big(\Gamma_* -\Gamma_{p,*}' \Gamma_p^{-1} \Gamma_{p,*}\Big){\nu^*},
\end{align*}
where now 
\begin{align*}
    X(s)&\equiv \Big(1,C_\kappa(s)',J_{c,d}(s),G_{c,d}(s)\Big)',\\
    \varsigma&\equiv \Big(\int_0^1 X(s)X(s)' \dd s\Big)^{-1} \int_0^1X(s) \Big({\gamma^*}'C_{\kappa^*}(s) +\beta \cdot s\Big)\dd s,
\end{align*}
and $\beta\equiv \lim_{n\to\infty}\beta_n$.
\end{proof}
